\documentclass[preprint,a4paper,12pt,number]{elsarticle}

\usepackage{latexsym}
\usepackage{amssymb}
\usepackage{amsmath}
\usepackage{url}
\usepackage[all]{xy}
\newcommand{\ignore}[1]{} 
\newcommand{\too}{\longrightarrow}
\newcommand{\arro}[1]{\xrightarrow{#1}} 
\newcommand{\hoo}{\hookrightarrow} 
\newcommand{\del}{\mathsf{del}}
\newcommand{\rolldel}{\mathsf{rolldel}} 
\newcommand{\comp}{\:|\:} 

\newtheorem{theorem}{Theorem}
\newtheorem{lemma}[theorem]{Lemma}

\newtheorem{corollary}[theorem]{Corollary}
\newdefinition{definition}[theorem]{Definition}
\newdefinition{example}[theorem]{Example}
\newproof{proof}{Proof}

\usepackage{color}

\newcommand{\blue}[1]{{\color{black} #1}}  

\makeatletter
\newenvironment{prog}{\vspace{0.7ex}\par
\setlength{\parindent}{0.7cm}
\obeylines\@vobeyspaces\tt}{\vspace{0.7ex}\noindent
}
\makeatother
\newcommand{\startprog}{\begin{prog}}
\newcommand{\stopprog}{\end{prog}\noindent}

\newcommand{\id}{id}

\def\defemb#1#2{\expandafter\def\csname #1\endcsname
                              {\relax\ifmmode #2\else\hbox{$#2$}\fi}}
\defemb{cA}{{\cal A}}
\defemb{cB}{{\cal B}}
\defemb{cC}{{\cal C}}
\defemb{cD}{{\cal D}}
\defemb{cE}{{\cal E}}
\defemb{cF}{{\cal F}}
\defemb{cG}{{\cal G}}
\defemb{cH}{{\cal H}}
\defemb{cI}{{\cal I}}
\defemb{cJ}{{\cal J}}
\defemb{cL}{{\cal L}}
\defemb{cM}{{\cal M}}
\defemb{cN}{{\cal N}}
\defemb{cO}{{\cal O}}
\defemb{cP}{{\cal P}}
\defemb{cQ}{{\cal Q}}
\defemb{cR}{{\cal R}}
\defemb{cS}{{\cal S}}
\defemb{cT}{{\cal T}}
\defemb{cU}{{\cal U}}
\defemb{cV}{{\cal V}}
\defemb{cX}{{\cal X}}
\defemb{cZ}{{\cal Z}}

\newcommand{\dom}{{\cD}om}

\newcommand{\nil}{[\:]}

\renewcommand{\phi}{\varphi}

\newcommand{\ol}[1]{\overline{#1}}  

\newcommand{\cons}{\!:\!}

\newcommand{\h}{\mathit{h}}
\renewcommand{\k}{\lambda}
\newcommand{\init}{\mathsf{init}}
\newcommand{\final}{\mathsf{final}}
\newcommand{\rlh}{\rightleftharpoons}
\newcommand{\lh}{\leftharpoondown}
\newcommand{\rh}{\rightharpoonup}
\newcommand{\lhh}{\leftharpoondown\hspace{-1.8ex}\leftharpoondown}

\let\l=\langle
\let\r=\rangle
\def \tuple#1{\langle #1 \rangle}

\long\def\comment#1{}

\begin{document}

\begin{frontmatter}

\title{A Theory of Reversibility for Erlang\tnoteref{t1}}
\tnotetext[t1]{This work has been partially supported by
  MINECO/AEI/FEDER (EU) under grants TIN2013-44742-C4-1-R and
  TIN2016-76843-C4-1-R, by the \emph{Generalitat Valenciana} under
  grant PROMETEO-II/2015/013 (SmartLogic), by the COST Action
  IC1405 on Reversible Computation - extending horizons of computing, and by JSPS KAKENHI Grant Number JP17H01722.
  Adri\'an Palacios was partially supported by the EU (FEDER) and the
  Spanish \emph{Ayudas para contratos predoctorales para la
    formaci\'on de doctores} and \emph{Ayudas a la movilidad
    predoctoral para la realizaci\'on de estancias breves en centros
    de I+D}, MINECO (SEIDI), under FPI grants BES-2014-069749 and
  EEBB-I-16-11469.
  Ivan Lanese was partially supported by INdAM as a member of GNCS (Gruppo Nazionale per il Calcolo Scientifico).
  Part of this research was done while the third and fourth authors
  were visiting Nagoya and Bologna Universities; they gratefully
  acknowledge their hospitality.
  Finally, we thank Salvador Tamarit and the anonymous reviewers for
  their helpful suggestions and comments.\\[1ex]
  $\copyright$ 2018. This manuscript version is made available under
  the CC-BY-NC-ND 4.0 license \url{http://creativecommons.org/licenses/by-nc-nd/4.0/}
}
  \author[blg]{Ivan Lanese}
  \ead{ivan.lanese@gmail.com}
  \author[ngy]{Naoki Nishida}
  \ead{nishida@i.nagoya-u.ac.jp}
  \author[vlc]{Adri\'an Palacios} \ead{apalacios@dsic.upv.es}
  \author[vlc]{Germ\'an Vidal\corref{cor}} \ead{gvidal@dsic.upv.es}
  \address[blg]{Focus Team, University of Bologna/INRIA\\
  Mura Anteo Zamboni, 7, Bologna, Italy}
  \address[ngy]{Graduate School of Informatics, Nagoya University\\
    Furo-cho, Chikusa-ku, 4648603 Nagoya, Japan}
  \address[vlc]{MiST, DSIC, Universitat Polit\`ecnica de Val\`encia\\
    Camino de Vera, s/n, 46022 Valencia, Spain}

\cortext[cor]{Corresponding author.}


\begin{abstract}
  In a reversible language, any forward computation can be undone by a
  finite sequence of backward steps. Reversible computing has been
  studied in the context of different programming languages and
  formalisms, where it has been used for testing and verification,
  among others.
  In this paper, we consider a subset of Erlang, a functional and
  concurrent programming language based on the actor model. We present
  a formal semantics for reversible computation in this language and
  prove its main properties, including its causal consistency. We also
  build on top of it a rollback operator that can be used to undo the actions
  of a process up to a given checkpoint.\\[1ex]

  To appear in the \emph{Journal of
Logical and Algebraic Methods in Programming} (Elsevier).
\end{abstract}


\end{frontmatter}

\section{Introduction} \label{intro}

Let us consider that the operational semantics of a programming
language is specified by a state transition relation $R$ such that
$R(s,s')$ holds if the state $s'$ is reachable---in one step---from
state $s$.
Then, we say that a programming language (or formalism) is
\emph{reversible} if there exists a constructive algorithm that can be
used to recover the predecessor state $s$ from $s'$.
In general, such a property does not hold for most programming
languages and formalisms.
We refer the interested reader to, e.g.,
\cite{Bennett00,Fra05,Yok10,YAG08} for a high level account of the
principles of reversible computation.

The notion of \emph{reversible computation} was first introduced in
Landauer's seminal work \cite{Lan61} and, then, further improved by
Bennett \cite{Ben73} in order to avoid the generation of ``garbage''
data.  The idea underlying these works is that any programming
language or formalism can be made reversible by adding the
\emph{history} of the computation to each state, which is usually
called a \emph{Landauer embedding}. Although carrying the history of
a computation might seem infeasible because of its size, there are
several successful proposals that are based on this idea. In
particular, one can restrict the original language or apply a number
of analysis in order to restrict the required information in the
history as much as possible, as in, e.g., \cite{MHNHT07,NPV16,TA15} in
the context of a functional language.

In this paper, we aim at introducing a form of reversibility in the
context of a programming language that follows the actor model
(concurrency based on message passing), a first-order subset of the
concurrent and functional language Erlang \cite{AVW96}. Previous
approaches have mainly considered reversibility in---mostly
synchronous---concurrent calculi like CCS \cite{DK04,DK05} and
$\pi$-calculus~\cite{CKV13}; a general framework for reversibility of
algebraic process calculi \cite{PU07}, or the recent approach to
reversible session-based $\pi$-calculus \cite{TY15}. However, we can
only find a few approaches that considered the reversibility of
\emph{asynchronous} calculi, e.g., Cardelli and Laneve's reversible
structures \cite{CL11}, and reversible extensions of the concurrent
functional language $\mu$Oz~\cite{LLMS12}, of a higher-order
asynchronous $\pi$-calculus~\cite{LMS16}, and of the coordination
language $\mu$Klaim~\cite{GLMT15}. In the last two cases, a form
of control of the backward execution using a rollback operator has
also been studied~\cite{LMSS11,GLMT15}. In the case of $\mu$Oz,
reversibility has been exploited for debugging~\cite{GLM14}.

To the best of our knowledge, our work is the first one that considers
reversibility in the context of the functional, concurrent, and
distributed language Erlang. Here, given a running Erlang system
consisting of a pool of interacting processes, possibly distributed in
several computers, we aim at allowing a \emph{single} process to undo
its actions in a stepwise manner, including the interactions with
other processes, following a rollback fashion. In this context, we
must ensure \emph{causal consistency} \cite{DK04}, i.e., an action
cannot be undone until all the actions that depend on it have already
been undone. E.g., if a process $\mathrm{p1}$ spawns a process $\mathrm{p2}$, we cannot undo
the spawning of process $\mathrm{p2}$ until all the actions performed by the
process $\mathrm{p2}$ are undone too. This is particularly challenging in an
asynchronous and distributed setting, where ensuring causal
consistency for backward computations is far from trivial.

In this paper, we consider a simple Erlang-like language that can be
seen as a subset of \emph{Core Erlang} \cite{CGJLNPV04}.  We present
the following contributions:
\begin{itemize}
\item First, we introduce an appropriate semantics for the
  language. In contrast to previous semantics like that in
  \cite{CMRT13tr} which were monolithic, ours is modular, which
  simplifies the definition of a reversible extension. Here, we follow
  some of the ideas in~\cite{SFB10}, e.g., the use of a global mailbox
  (there called ``ether'').
  There are also some differences though.
  In the semantics of \cite{SFB10}, at the expression level, the
  semantics of a receive statement is, in
  principle, infinitely branching, since their formulation allows for
  an infinite number of possible queues and selected messages (see
  \cite[page 53]{Fre01} for a detailed explanation).  This
  source of nondeterminism is avoided in our semantics.  

\item We then introduce a reversible semantics that can go both
  forward and backward (basically, a Landauer embedding), in a
  nondeterministic fashion, called an \emph{uncontrolled} reversible
  semantics according to the terminology in~\cite{LMT14}. Here, we
  focus on the concurrent actions (namely, process spawning, message
  sending and receiving) and, thus, we do not define a reversible
  semantics for the functional component of the language;
  rather, we assume that the state of the process---the current
  expression and its environment---is stored in the history after each
  execution step. This approach could be improved following, e.g., the
  techniques presented in \cite{MHNHT07,NPV16,TA15}.  We state and formally
  prove several properties of the semantics and, particularly, its
  causal consistency.

\item Finally, we add control to the reversible semantics by
  introducing a \emph{rollback operator} that can be used to undo 
  the actions of a given process until a given
  checkpoint---introduced by the programmer---is reached. In order to
  ensure causal consistency, the rollback action might be propagated
  to other, dependent processes.
\end{itemize}
This paper is an extended version of \cite{NPV16b}. Compared to
\cite{NPV16b}, we introduce an uncontrolled reversible semantics and
prove a number of fundamental theoretical properties, including its
causal consistency. The rollback semantics, originally introduced in
\cite{NPV16b}, has been refined and improved (see
Section~\ref{sec:relwork} for more details).

The paper is organised as follows. The syntax and semantics of the
considered language are presented in Sections~\ref{sec:syntax} and
\ref{sec:semantics}, respectively. Our (uncontrolled) reversible
semantics is then introduced in Section~\ref{sec:revsem}, while the
rollback operator is defined in Section~\ref{sec:rollbacksem}.  A
proof-of-concept implementation of the reversible semantics is
described in Section~\ref{sec:imple}.  Finally, some related work is
discussed in Section~\ref{sec:relwork}, and Section~\ref{sec:conc}
concludes and points out some directions for future work.

\section{Language Syntax} \label{sec:syntax} 

In this section, we present the syntax of a first-order concurrent and
distributed functional language that follows the actor model. Our
language is equivalent to a subset of Core Erlang \cite{CGJLNPV04}.

\begin{figure}[t]
  \begin{center}
  $
  \begin{array}{rcl@{~~~~~~}l}
    \mathit{module} & ::= & \mathsf{module} ~ Atom = 
    \mathit{fun}_1~\ldots~\mathit{fun}_n\\
    {\mathit{fun}} & ::= & \mathit{fname} = \mathsf{fun}~(\mathit{Var}_1,\ldots,\mathit{Var}_n) \to expr \\
    {\mathit{fname}} & ::= & Atom/\mathit{Integer} \\
    lit & ::= & Atom \mid \mathit{Integer} \mid \mathit{Float} \mid
                \mathit{Pid} \mid \nil \\
    expr & ::= & \mathit{Var} \mid lit \mid \mathit{fname} \mid [expr_1|expr_2]
                 \mid   \{expr_1,\ldots,expr_n\} \\
    & \mid & \mathsf{call}~Op~(expr_1,\ldots,expr_n) 
    \mid \mathsf{apply}~\mathit{fname}~(expr_1,\ldots,expr_n) \\
    & \mid &
    \mathsf{case}~expr~\mathsf{of}~clause_1;\ldots;clause_m~\mathsf{end}\\
    & \mid & \mathsf{let}~\mathit{Var}=expr_1~\mathsf{in}~expr_2 
    \mid \mathsf{receive}~clause_1;\ldots;clause_n~\mathsf{end}\\
    & \mid & \mathsf{spawn}(\mathit{fname},[expr_1,\ldots,expr_n])  
     \mid expr \:!\: expr \mid \mathsf{self}()\\
    clause & ::= & pat ~\mathsf{when}~expr_1 \to expr_2
    \\
    pat & ::= & \mathit{Var} \mid lit \mid [pat_1|pat_2] \mid
    \{pat_1,\ldots,pat_n\} \\
  \end{array}
  $
  \end{center}
\caption{Language syntax rules} \label{ErlangSyntax}
\end{figure}

The syntax of the language can be found in Figure~\ref{ErlangSyntax}.
Here, a module is a sequence of function definitions, where each
function name $f/n$ (atom/arity) has an associated definition of the
form $\mathsf{fun}~(X_1,\ldots,X_n) \to e$. We consider that a program
consists of a single module for simplicity. The body of a function is
an \emph{expression}, which can include variables, 
literals, function names, lists, tuples, calls
to built-in functions---mainly arithmetic and relational operators---,
function applications, case expressions, let bindings, and receive
expressions; furthermore, we also include the functions
$\mathsf{spawn}$, ``$!$'' (for sending a message), and
$\mathsf{self}()$ that are usually considered built-ins in the Erlang
language. As is common practice, we assume that $X$ is a fresh
variable in a let binding of the form
$\mathsf{let}~X=\mathit{expr}_1~\mathsf{in}~\mathit{expr}_2$.

As shown by the syntax in Figure~\ref{ErlangSyntax}, we only consider
first-order expressions. Therefore, the first argument in applications
and spawns is a function name (instead of an arbitrary expression or
closure). Analogously, the first argument in calls is a built-in operation $Op$.

In this language, we distinguish expressions, patterns, and values. 
Here, \emph{patterns} are built from variables, 
literals, lists, and tuples, while \emph{values} are built from
literals, 
lists, and tuples, i.e., they are \emph{ground}---without
variables---patterns. Expressions are denoted by
$e,e',e_1,e_2,\ldots$, patterns by $pat$, $pat'$, $pat_1$, $pat_2,\ldots$ and
values by $v,v',v_1,v_2,\ldots$ Atoms are typically denoted with
roman letters, while variables start with an uppercase letter.
As it is common practice, a \emph{substitution} $\theta$ is a mapping
from variables to expressions, and $\dom(\theta) =
\{X\in\mathit{Var} \mid X \neq \theta(X)\}$ is its
domain.\!\footnote{Since we consider an eager language, variables are bound 
to values.}
Substitutions are usually denoted by sets of bindings like,
e.g., $\{X_1\mapsto v_1,\ldots,X_n\mapsto v_n\}$.
Substitutions are extended to morphisms from expressions to
expressions in the natural way.  
The identity substitution is denoted by $\id$. Composition of
substitutions is denoted by juxtaposition, i.e., $\theta\theta'$
denotes a substitution $\theta''$ such that $\theta''(X) =
\theta'(\theta(X))$ for all $X\in\mathit{Var}$. Also, we denote by
$\theta[X_1\mapsto v_1,\ldots,X_n\mapsto v_n]$ the \emph{update} of
$\theta$ with the mapping $\{X_1\mapsto v_1,\ldots,X_n\mapsto v_n\}$,
i.e., it denotes a new substitution $\theta'$ such that $\theta'(X) =
v_i$ if $X = X_i$, for some $i\in\{1,\ldots,n\}$, and $\theta'(X) =
\theta(X)$ otherwise. 

In a case expression
``$\mathsf{case}~e~\mathsf{of} ~ pat_1~\mathsf{when}~e_1 \to e'_1;
~\ldots; ~ pat_n~\mathsf{when}~e_n \to e'_n~~\mathsf{end} $''\!\!, we
first evaluate $e$ to a value, say $v$; then, we should find (if any)
the first clause $pat_i ~\mathsf{when}~e_i\to e'_i$ such that $v$
matches $pat_i$ (i.e., there exists a substitution $\sigma$ for the
variables of $pat_i$ such that $v=pat_i\sigma$) and $e_i\sigma$---the
\emph{guard}---reduces to $\emph{true}$; then, the case expression
reduces to $e'_i\sigma$. Note that guards can only contain calls to
built-in functions (typically, arithmetic and relational operators).

As for the concurrent features of the language, we consider that a
\emph{system} is a pool of processes that can only interact through
message sending and receiving (i.e., there is no shared memory). Each
process has an associated \textit{pid} (process identifier), which is
unique in a system. As in Erlang, we consider a specific type or
domain Pid for pids. Furthermore, in this work, we assume that pids can
only be introduced in a computation from the evaluation of functions
$\mathsf{spawn}$ and $\mathsf{self}$ (see below).
By abuse of notation, when no confusion can arise, we refer to a
process with its pid.

An expression of the form $\mathsf{spawn}(f/n,[e_1,\ldots,e_n])$ has,
as a \emph{side effect}, the creation of a new process, with a fresh
pid $p$, initialised with the expression
$\mathsf{apply}~f/n~(v_1,\ldots,v_n)$, where $v_1,\ldots,v_n$ are the
evaluations of $e_1,\ldots,e_n$, respectively; the expression
$\mathsf{spawn}(f/n,[e_1,\ldots,e_n])$ itself evaluates to the new pid
$p$.
The function $\mathsf{self}()$ just returns the pid of the current
process.  
An expression of the form $e_1\:!\: e_2$, where $e_1$ evaluates to a pid $p$ and $e_2$ to a value $v$, also evaluates to the value $v$ and,
as a side effect, the value $v$---the \emph{message}---will be stored in the
queue or \emph{mailbox} of process $p$ at some point in the future.

Finally, an expression
``$ \mathsf{receive}~pat_1~\mathsf{when}~e_1\to
e'_1;\ldots;pat_n~\mathsf{when}~e_n\to e'_n~~\mathsf{end} $''
traverses the messages in the process' queue until one of them matches
a branch in the receive statement; i.e., it should find the
\emph{first} message $v$ in the process' queue (if any) such that
$\mathsf{case}~v~\mathsf{of}~pat_1~\mathsf{when}~e_1\to
e'_1;\ldots;pat_n~\mathsf{when}~e_n\to e'_n~\mathsf{end}$ can be
reduced; then, the receive expression evaluates to the same expression
to which the above case expression would be evaluated, with the
additional side effect of deleting the message $v$ from the process'
queue.
If there is no matching message in the queue, the process suspends its
execution until a matching message arrives.

\begin{example} \label{ex1} Consider the program shown in
  Figure~\ref{fig:ex1}, where the symbol ``$\_$'' is used to denote an
  \emph{anonymous} variable, i.e., a variable whose name is not
  relevant. The computation starts with ``$ \mathsf{apply}~\mathrm{main}/0~()
  $.\!'' This creates a process, say $\mathrm{p1}$. Then, $\mathrm{p1}$
  spawns two new processes, say $\mathrm{p2}$ and $\mathrm{p3}$, and
  then sends the message $\mathrm{hello}$ to process $\mathrm{p3}$
  and the message $\{\mathrm{p3},\mathrm{world}\}$ to process
  $\mathrm{p2}$, which then resends $\mathrm{world}$ to
  $\mathrm{p3}$. Note that we consider that variables $P2$ and $P3$
  are bound to pids $\mathrm{p2}$ and $\mathrm{p3}$, respectively.

  In our language, there is no guarantee regarding which message
  arrives first to $\mathrm{p3}$, i.e., both interleavings (a) and (b)
  in Figure~\ref{fig:ex1aux} are possible (resulting in function
  $\mathrm{target}/0$ returning either
  $\{\mathrm{hello},\mathrm{world}\}$ or
  $\{\mathrm{world},\mathrm{hello}\}$).  This is coherent with the
  semantics of Erlang, where the only guarantee is that if two
  messages are sent from process $p$ to process $p'$, and both are
  delivered, then the order of these messages is
  kept.\footnote{Current implementations only guarantee this
    restriction within the same node though.}

  \begin{figure}[t]
    \[
    \begin{array}{r@{~}l@{~~}r@{~}l}
      \mathrm{main}/0 = \mathsf{fun}~()\to & \mathsf{let}~P2 =
      \mathsf{spawn}(\mathrm{echo}/0,\nil) \\
      & \mathsf{in}~\mathsf{let}~P3 = \mathsf{spawn}(\mathrm{target}/0,\nil) \\
      & \mathsf{in}~\mathsf{let}~\_ = P3\:!\: \mathrm{hello} \\
      & \mathsf{in}~P2\:!\:\{P3,\mathrm{world}\}\\[1ex]
      \mathrm{target}/0 = \mathsf{fun}~()\to & \mathsf{receive}\\
      & \hspace{3ex}A \to \mathsf{receive} \\
      & \hspace{11ex}B\to \{A,B\}\\
      & \hspace{8.5ex}\mathsf{end}& \\
      &\mathsf{end}\\[1ex]
      \mathrm{echo}/0 = \mathsf{fun}~()\to & \mathsf{receive} \\
      & \hspace{3ex}\{P,M\}\to P\:!\: M~\\
      &\mathsf{end} \\
    \end{array}
    \]
    \caption{A simple concurrent program} \label{fig:ex1}
  \end{figure}

  \begin{figure}[t]
    \footnotesize
     \[
     \xymatrix@R=8pt@C=2pt{
         \underline{\mathrm{p1}} \ar@{..}[d] 
        & \underline{\mathrm{p2}} \ar@{..}[ddd] &
        \underline{\mathrm{p3}} \ar@{..}[ddddd] \\
        \mathrm{p3}\:!\:\mathrm{hello} \ar@{..}[d] \ar[drr] & & \\
        \mathrm{p2}\:!\:\{\mathrm{p3},\mathrm{world}\} \ar[dr] & & \\
        & \mathrm{p3}\:!\:\mathrm{world} \ar[dr] & \\
        & & \\
        & & \\
      } 
\hspace{12ex}
      \xymatrix@R=8pt@C=2pt{
        & \underline{\mathrm{p1}} \ar@{..}[d] 
        & \underline{\mathrm{p2}} \ar@{..}[ddd] &
        \underline{\mathrm{p3}} \ar@{..}[ddddd] \\
        & \mathrm{p3}\:!\:\mathrm{hello} \ar@{..}[d] \ar `l /20pt [d] `[dddd] [ddddrr] & & \\
        & \mathrm{p2}\:!\:\{\mathrm{p3},\mathrm{world}\} \ar[dr] & & \\
        & & \mathrm{p3}\:!\:\mathrm{world} \ar[dr] & \\
        & & & \\
        & & & \\
      } 
      \]  
      \centering (a)\hspace{40ex}(b)
  \caption{Admissible interleavings in Example~\ref{ex1}} \label{fig:ex1aux}
  \end{figure}
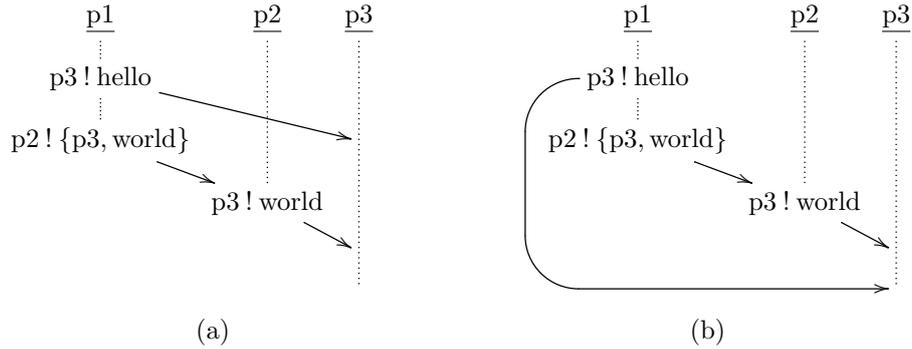
\end{example}

\section{The Language Semantics} \label{sec:semantics} 

In order to precisely set the framework for our proposal, in this
section we formalise the semantics of the considered language.

\begin{definition}[Process]
  A process is denoted by a tuple $\tuple{p,(\theta,e),q}$ where $p$
  is the pid of the process, $(\theta,e)$ is the control---which
  consists of an environment (a substitution) and an expression to be
  evaluated---and $q$ is the process' mailbox, a FIFO queue with the
  sequence of messages that have been sent to the process.

  We consider the following operations on local mailboxes. Given a
  message $v$ and a local mailbox $q$, we let $v:q$ denote a new
  mailbox with message $v$ on top of it (i.e., $v$ is the newer
  message). We also denote with $q\backslash\!\!\backslash v$ a new
  queue that results from $q$ by removing the oldest occurrence of
  message $v$ (which is not necessarily the oldest message in the
  queue).
\end{definition}
A running \emph{system} can then be seen as a pool of processes, which
we formally define as follows:

\begin{definition}[System]
  A system is denoted by $\Gamma;\Pi$, where $\Gamma$, the \emph{global
  mailbox}, is a multiset of pairs of the form
  $(target\_process\_pid,message)$, and $\Pi$ is a pool of processes,
  denoted by an expression of the form
  \[
    \tuple{p_1,(\theta_1,e_1),q_1} ~\comp \cdots
    \comp~\tuple{p_n,(\theta_n,e_n),q_n}
  \]
  where ``$\comp$'' denotes an associative and commutative operator.
  Given a global mailbox $\Gamma$, we let $\Gamma\cup\{(p,v)\}$ denote
  a new mailbox also including the pair $(p,v)$, where we use
  ``$\cup$'' as multiset union.

  We often denote a system by an expression of the form
  $ \Gamma; \tuple{p,(\theta,e),q}\comp\Pi $ to point out that
  $\tuple{p,(\theta,e),q}$ is an arbitrary process of the pool (thanks
  to the fact that ``$\comp$'' is associative and commutative).
\end{definition}
Intuitively, $\Gamma$ stores messages after they are sent, and before
they are inserted in the target mailbox, hence it models messages
which are in the network.  The use of $\Gamma$ (which is similar to
the ``ether'' in \cite{SFB10}) is needed to guarantee that all message
interleavings admissible in an asynchronous communication model (where
the order of messages is not preserved) can be generated by our
semantics.

In the following, we denote by $\ol{o_n}$ a sequence of syntactic
objects $o_1,\ldots,o_n$ for some $n$. We also write $\ol{o_{i,j}}$
for the sequence $o_i,\ldots,o_j$ when $i\leq j$ (and the empty
sequence otherwise).  We write $\ol{o}$ when the number of elements is
not relevant.

\begin{figure}[t]
  \footnotesize
  \[
  \begin{array}{c}
    (\mathit{Var}) ~ {\displaystyle \frac{}{\theta,X
        \arro{\tau} \theta,\theta(X)}} 

    \hspace{2ex}

    (\mathit{Tuple}) ~ {\displaystyle 
      \frac{\theta,e_i \arro{\ell}
        \theta',e'_i}{\theta,\{\ol{v_{1,{i-1}}},e_i,\ol{e_{{i+1},n}}\}
        \arro{\ell} \theta',
        \{\ol{v_{1,{i-1}}},e'_i,\ol{e_{{i+1},n}}\}}}\\[4ex] 

    (\mathit{List1})  ~{\displaystyle 
      \frac{\theta,e_1 \arro{\ell}
        \theta',e'_1}{\theta,[e_1|e_2]
        \arro{\ell} \theta',
        [e'_1|e_2]}} 

    \hspace{2ex}

    (\mathit{List2}) ~ {\displaystyle 
      \frac{\theta,e_2 \arro{\ell}
        \theta',e'_2}{\theta,[v_1|e_2]
        \arro{\ell} \theta',
        [v_1|e'_2]}} \\[4ex]

      (\mathit{Let1}) ~ {\displaystyle \frac{\theta,e_1
          \arro{\ell} \theta',e'_1 }{\theta,\mathsf{let}~ 
          X=e_1~\mathsf{in}~e_2
          \arro{\ell} \theta',\mathsf{let}~ 
          X=e'_1~\mathsf{in}~e_2}}

      \hspace{1ex}

      (\mathit{Let2}) ~ {\displaystyle \frac{}{\theta,\mathsf{let}~ 
          X=v~\mathsf{in}~e
          \arro{\tau} \theta[X\mapsto v],e}} \\[4ex]
 
     (\mathit{Case1}) ~ {\displaystyle
        \frac{\theta,e\arro{\ell}
          \theta',e'}{\begin{array}{l}
            \theta,\mathsf{case}~e~\mathsf{of}~cl_1;\ldots;cl_n~\mathsf{end}
        \arro{\ell}
        \theta',\mathsf{case}~e'~\mathsf{of}~cl_1;\ldots;cl_n~\mathsf{end}\\
      \end{array}}}\\[4ex]
  
      (\mathit{Case2}) ~ {\displaystyle
        \frac{\mathsf{match}(\theta,v,cl_1,\ldots,cl_n) = \tuple{\theta_i,e_i}}{\theta,\mathsf{case}~v~\mathsf{of}~cl_1;\ldots;cl_n~\mathsf{end}
        \arro{\tau} \theta\theta_i,e_i}} \\[4ex]

    (\mathit{Call1}) ~ {\displaystyle
    \frac{\theta,e_i\arro{\ell}
    \theta',e'_i~~~i\in\{1,\ldots,n\}}{\theta,\mathsf{call}~op~(\ol{v_{1,i-1}},e_i,\ol{e_{i+1,n}})
    \arro{\ell}
    \theta',\mathsf{call}~op~(\ol{v_{1,i-1}},e'_i,\ol{e_{i+1,n}})}} \\[5ex]

    (\mathit{Call2}) ~ {\displaystyle
      \frac{\mathsf{eval}(op,v_1,\ldots,v_n)=v}{\theta,\mathsf{call}~op~(v_1,\ldots,v_n)
        \arro{\tau} \theta,v}} \\[4ex]

    (\mathit{Apply1}) ~ {\displaystyle
      \frac{\theta,e_i\arro{\ell}
        \theta',e'_i~~~i\in\{1,\ldots,n\}}{\theta,\mathsf{apply}~a/n~(\ol{v_{1,i-1}},e_i,\ol{e_{i+1,n}}) 
        \arro{\ell}
        \theta',\mathsf{apply}~a/n~(\ol{v_{1,i-1}},e'_i,\ol{e_{i+1,n}})}}\\[5ex] 

    (\mathit{Apply2}) ~ {\displaystyle
      \frac{\mu(a/n) = \mathsf{fun}~(X_1,\ldots,X_n)\to e}{\theta,\mathsf{apply}~a/n~(v_1,\ldots,v_n)
        \arro{\tau} \theta\cup\{X_1\mapsto v_1,\ldots,X_n\mapsto v_n\},e}} 
  \end{array}
  \]
\caption{Standard semantics: evaluation of sequential expressions} \label{fig:seq-rules}
\end{figure}

\begin{figure}[t]
 \footnotesize
  \[
  \begin{array}{r@{~}c}
      (\mathit{Send1}) & {\displaystyle 
      \frac{\theta,e_1 \arro{\ell} \theta',e'_1}{\theta,e_1\:!\: e_2 \arro{\ell}
        \theta',e'_1\:!\: e_2} ~~~~(\mathit{Send2}) ~\frac{\theta,e_2 \arro{\ell} \theta',e'_2}{\theta,v_1\:!\: e_2 \arro{\ell}
        \theta',v_1\:!\: e'_2} 
      }\\[4ex]

      (\mathit{Send3}) & {\displaystyle
          \frac{}{\theta,v_1\:!\: v_2 \arro{\mathsf{send}(v_1,v_2)} \theta,v_2}
          }\\[4ex]

    (\mathit{Receive}) & {\displaystyle
      \frac{}{\theta,\mathsf{receive}~cl_1;\ldots;cl_n~\mathsf{end}
        \arro{\mathsf{rec}(\kappa,\ol{cl_n})}
        \theta,\kappa
      }
      }\\[4ex]

     (\mathit{Spawn1}) & {\displaystyle
       \frac{\theta,e_i\arro{\ell}
        \theta',e'_i~~~i\in\{1,\ldots,n\}}{\theta,\mathsf{spawn}(a/n,[\ol{v_{1,i-1}},e_i,\ol{e_{i+1,n}}])
         \arro{\ell} \theta',\mathsf{spawn}(a/n,[\ol{v_{1,i-1}},e'_i,\ol{e_{i+1,n}}])         
       }}\\[4ex]

     (\mathit{Spawn2}) & {\displaystyle
       \frac{}{\theta,\mathsf{spawn}(a/n,[\ol{v_n}])
         \arro{\mathsf{spawn}(\kappa,a/n,[\ol{v_n}])} \theta,\kappa         
       }}\\[4ex]

    (\mathit{Self}) & {\displaystyle
     \frac{}{\theta,\mathsf{self}() \arro{\mathsf{self}(\kappa)} \theta,\kappa}}
  \end{array}
  \]
\caption{Standard semantics: evaluation of concurrent expressions} \label{fig:concurrent-rules}
\end{figure}

The semantics is defined by means of two transition relations: $\too$
for expressions and $\hoo$ for systems. Let us first consider the
labelled transition relation
\[
\too\; : (Env,Exp)\times Label\times(Env,Exp)
\]
where $Env$ and $Exp$ are the domains of environments (i.e.,
substitutions) and expressions, respectively, and $Label$ denotes an
element of the set 
\[
\{\tau, \mathsf{send}(v_1,v_2),
\mathsf{rec}(\kappa,\ol{cl_n}), \mathsf{spawn}(\kappa,a/n,[\ol{v_n}]),
\mathsf{self}(\kappa)\}
\]
whose meaning will be explained below. We use $\ell$ to range over labels.
For clarity, we divide the transition rules of the semantics for
expressions in two sets: rules for sequential expressions are depicted
in Figure~\ref{fig:seq-rules}, while rules for concurrent ones are in
Figure~\ref{fig:concurrent-rules}.\footnote{By abuse, we include the
  rule for $\mathsf{self}()$ together with the concurrent actions.}
Note, however, that concurrent expressions can occur inside sequential
expressions.

Most of the rules are self-explanatory. In the following, we only
discuss some subtle or complex issues. In principle, the transitions
are labelled either with $\tau$ (a sequential reduction without side
effects) or with a label that identifies the reduction of a (possibly
concurrent) action with some side-effects. Labels are used in the
system rules (Figure~\ref{fig:system-rules}) to determine the
associated side effects and/or the information to be retrieved.

As in Erlang, we consider that the order of evaluation of the
arguments in a tuple, list, etc., is fixed from left to right.

For case evaluation, we assume an auxiliary function $\mathsf{match}$
which selects the first clause, $cl_i = (pat_i~\mathsf{when}~e'_i\to
e_i)$, such that $v$ matches $pat_i$, i.e., $v=\theta_i(pat_i)$, and
the guard holds, i.e., $ \theta\theta_i,e'_i \too^\ast \theta',true
$. As in Core Erlang, we assume that the patterns can only contain
fresh variables (but guards might have bound variables, thus we
pass the current environment $\theta$ to function $\mathsf{match}$).
Note that, for simplicity, we assume here that if the argument $v$
matches no clause then the evaluation is blocked.\footnote{This is not
  an issue in practice since, when an Erlang program is
  translated to the intermediate representation Core Erlang, a
  catch-all clause is added to every case expression in order to deal
  with pattern matching errors.}

Functions can either be defined in the program (in this case they are
invoked by $\mathsf{apply}$) or be a built-in (invoked by $\mathsf{call}$). In
the latter case, they are evaluated using the auxiliary function
$\mathsf{eval}$.
In rule $\mathit{Apply2}$, we consider that the mapping $\mu$ stores
all function definitions in the program, i.e., it maps every function name $a/n$ to a
copy of its definition $\mathsf{fun}~(X_1,\ldots,X_n)\to e$, where
$X_1,\ldots,X_n$ are (distinct) 
fresh variables and are the only variables that may occur free in $e$.
As for the applications, note that we only consider first-order
functions. In order to extend our semantics to also consider
higher-order functions, one should reduce the function name to a
\emph{closure} of the form $(\theta',\mathsf{fun}~(X_1,\ldots,X_n)\to
e)$.
We skip this extension since it is
orthogonal to our contribution.

Let us now consider the evaluation of concurrent expressions that
produce some side effect (Figure~\ref{fig:concurrent-rules}). Here, we
can distinguish two kinds of rules. On the one hand, we have rules
$\mathit{Send1}$, $\mathit{Send2}$ and $\mathit{Send3}$ for ``$!$''. In this case, we
know \emph{locally} what the expression should be reduced to (i.e.,
$v_2$ in rule $\mathit{Send3}$). For the remaining rules, this is not
known locally and, thus, we return a fresh distinguished symbol,
$\kappa$---by abuse, $\kappa$ is dealt with as a variable---so that
the system rules of Figure~\ref{fig:system-rules} 
will eventually bind $\kappa$ to its correct value:\footnote{Note that
  $\kappa$ takes values on the domain $expr \cup Pid$, in contrast to
  ordinary variables that can only be bound to values.
}
the selected expression in rule $\mathit{Receive}$ and a pid in
rules $\mathit{Spawn}$ and $\mathit{Self}$.
In these cases, the label of the transition contains all the information needed
by system rules to perform the evaluation at the system level, including the symbol $\kappa$.
This \emph{trick} allows us to keep the rules for expressions and
systems separated (i.e., the semantics shown in
Figures~\ref{fig:seq-rules} and \ref{fig:concurrent-rules} is mostly
independent from the rules in Figure~\ref{fig:system-rules}), in
contrast to other Erlang semantics, e.g., \cite{CMRT13tr}, where they are
combined into a single transition relation.

\begin{figure}[t]
  \footnotesize
  \[
  \begin{array}{r@{~~}c}
    (\mathit{Seq}) & {\displaystyle
      \frac{\theta,e\arro{\tau} \theta',e'
      }{\Gamma;\tuple{p,(\theta,e),q}\comp \Pi \hoo
      \Gamma;\tuple{p,(\theta',e'),q}\comp \Pi}
      }\\[4ex]

    (\mathit{Send}) & {\displaystyle
      \frac{\theta,e \arro{\mathsf{send}(p'',v)}
        \theta',e'}{\Gamma;\tuple{p,(\theta,e),q} 
        \comp \Pi \hoo \Gamma\cup (p'',v);\tuple{p,(\theta',e'),q}\comp \Pi}
      }\\[4ex]

      (\mathit{Receive}) & {\displaystyle
        \frac{\theta,e \arro{\mathsf{rec}(\kappa,\ol{cl_n})}
          \theta',e'~~~ \mathsf{matchrec}(\theta,\ol{cl_n},q) =
         (\theta_i,e_i,v)} {\Gamma;\tuple{p,(\theta,e),q}\comp \Pi \hoo
          \Gamma;\tuple{p,(\theta'\theta_i,e'\{\kappa\mapsto
            e_i\}),q\backslash\!\!\backslash v}\comp \Pi}
      }\\[4ex]
      
      (\mathit{Spawn}) & {\displaystyle
        \frac{\theta,e \arro{\mathsf{spawn}(\kappa,a/n,[\ol{v_n}])}
          \theta',e'~~~ p'~\mbox{is a fresh pid}}{\Gamma;\tuple{p,(\theta,e),q} 
          \comp \Pi \hoo \Gamma;\tuple{p,(\theta',e'\{\kappa\mapsto p'\}),q}\comp \tuple{p',(\id,\mathsf{apply}~a/n~(\ol{v_n})),\nil} 
          \comp \Pi}
      }\\[4ex]

    (\mathit{Self}) & {\displaystyle
      \frac{\theta,e \arro{\mathsf{self}(\kappa)} \theta',e'}{\Gamma;\tuple{p,(\theta,e),q} 
        \comp \Pi \hoo \Gamma;\tuple{p,(\theta',e'\{\kappa\mapsto p\}),q} 
        \comp \Pi }
      }\\[4ex]

    (\mathit{Sched}) & {\displaystyle
      \frac{~}{\Gamma\cup\{(p,v)\};\tuple{p,(\theta,e),q}\comp\Pi 
          \hoo \Gamma;\tuple{p,(\theta,e),v\cons q}\comp\Pi}
      }
  \end{array}
  \]
\caption{Standard semantics: system rules} \label{fig:system-rules}
\end{figure}

Finally, we consider the system rules, which are depicted in
Figure~\ref{fig:system-rules}. In most of the transition rules, we
consider an arbitrary system of the form
$ \Gamma;\tuple{p,(\theta,e),q}\:\comp\: \Pi $, where $\Gamma$ is
the global mailbox and $\tuple{p,(\theta,e),q}\:\comp\: \Pi$ is a
pool of processes that contains at least one process
$\tuple{p,(\theta,e),q}$.
Let us briefly describe the system rules. 

Rule $\mathit{Seq}$ just updates the control $(\theta, e)$ of the
considered process when a sequential expression is reduced using the
expression rules.

Rule $\mathit{Send}$ adds the pair $(p'',v)$ to the global mailbox
$\Gamma$ instead of adding it to the queue of process $p''$. This is
necessary to ensure that all possible message interleavings are
correctly modelled (as discussed in Example~\ref{ex1}).
Observe that $e'$ is usually different from
$v$ since $e$ may have different nested operators. E.g., if $e$ has
the form ``$\mathsf{case}~\mathrm{p}\:!\:
v~\mathsf{of}~\{\ldots\}$,\!'' then $e'$ will be ``$\mathsf{case}~
v~\mathsf{of}~\{\ldots\}$'' with label $\mathsf{send}(\mathrm{p},v)$.

In rule $\mathit{Receive}$, we use the auxiliary function
$\mathsf{matchrec}$ to evaluate a receive expression.  The main
difference w.r.t.\ $\mathsf{match}$ is that $\mathsf{matchrec}$ also
takes a queue $q$ and returns the selected message $v$.  More
precisely, function $\mathsf{matchrec}$ scans the queue $q$ looking
for the \emph{first} message $v$ matching a pattern of the receive
statement.
Then, $\kappa$ is bound to the expression in the selected clause,
$e_i$, and the environment is extended with the matching substitution.
If no message in the queue $q$ matches any clause, then the rule is
not applicable and the selected process cannot be reduced (i.e., it
suspends). As in case expressions, we assume that the patterns can
only contain fresh variables.

The rules presented so far allow one to store messages in the global mailbox, but not to remove messages from it.
This is precisely the task of the scheduler, which is
modelled by rule $\mathit{Sched}$. This rule nondeterministically
chooses a pair $(p,v)$ in the global mailbox $\Gamma$ and delivers the
message $v$ to the target process $p$. Here, we deliberately ignore
the restriction mentioned in Example~\ref{ex1}: ``the messages
sent---directly---between two given processes arrive in the same order
they were sent", since current implementations only guarantee it
within the same node.
In practice, ignoring this restriction amounts
to consider that each process is potentially run in a different node.
An alternative definition ensuring this
  restriction can be found in \cite{NPV16b}.

\begin{example} \label{ex1b} Consider again the program shown in
  Example~\ref{ex1}. Figures~\ref{fig:ex1b1} and \ref{fig:ex1b2} show
  a derivation from ``$ \mathsf{apply}~\mathrm{main}/0~() $'' where the call to
  function $\mathrm{target}$ reduces to
  $\{\mathrm{world},\mathrm{hello}\}$, as discussed in
  Example~\ref{ex1} (i.e., the interleaving shown in
  Figure~\ref{fig:ex1} (b)). Processes' pids are denoted with
  $\mathrm{p1}$, $\mathrm{p2}$ and $\mathrm{p3}$. For clarity, we
  label each transition step with the applied rule and underline the
  reduced expression.

  \begin{figure}[p]
    \scriptsize
    \[
      \begin{array}{l@{~}l@{~}l}
       & \{\:\}; & \tuple{\mathrm{p1},
          (\id,\underline{\mathsf{apply}~\mathrm{main}/0~()}),\nil} \\[1ex]

        \hoo_\mathit{Seq} & \{\:\}; & \tuple{\mathrm{p1},
          (\id,\mathsf{let}~P2=\underline{\mathsf{spawn}(\mathrm{echo}/0,\nil)}~\mathsf{in}~\ldots),\nil} \\[1ex]

        \hoo_\mathit{Spawn} & \{\:\}; & \tuple{\mathrm{p1},
          (\id,\underline{\mathsf{let}~P2=\mathrm{p2}~\mathsf{in}~\ldots}),\nil}\\
        & & \comp~\tuple{\mathrm{p2},(\id,\mathsf{apply}~\mathrm{echo}/0~\nil),\nil} \\[1ex]

        \hoo_\mathit{Seq} & \{\:\}; & \tuple{\mathrm{p1},
          (\{P2\mapsto\mathrm{p2}\},\mathsf{let}~P3=\underline{\mathsf{spawn}(\mathrm{target}/0,\nil)}~\mathsf{in}~\ldots),\nil}\\
        & & \comp~\tuple{\mathrm{p2},(\id,\mathsf{apply}~\mathrm{echo}/0~\nil),\nil} \\[1ex] 

        \hoo_\mathit{Spawn} & \{\:\}; & \tuple{\mathrm{p1},
          (\{P2\mapsto\mathrm{p2}\},\underline{\mathsf{let}~P3=\mathrm{p3}~\mathsf{in}~\ldots},\nil)}\\
        & & \comp~\tuple{\mathrm{p2},(\id,\mathsf{apply}~\mathrm{echo}/0~\nil),\nil} \\
        & & \comp~\tuple{\mathrm{p3},(\id,\mathsf{apply}~\mathrm{target}/0~\nil),\nil} \\[1ex] 

        \hoo_\mathit{Seq} & \{\:\}; & \tuple{\mathrm{p1},
          (\{P2\mapsto\mathrm{p2},P3\mapsto\mathrm{p3}\},\mathsf{let}~\_=\underline{P3}\:
                                    !\:\mathrm{hello}~\mathsf{in}~\ldots,\nil)}\\
        & & \comp~\tuple{\mathrm{p2},(\id,\mathsf{apply}~\mathrm{echo}/0~\nil),\nil} \\
        & & \comp~\tuple{\mathrm{p3},(\id,\mathsf{apply}~\mathrm{target}/0~\nil),\nil} \\[1ex] 

        \hoo_\mathit{Seq} & \{\:\}; & \tuple{\mathrm{p1},
          (\{P2\mapsto\mathrm{p2},P3\mapsto\mathrm{p3}\},\mathsf{let}~\_=\underline{\mathrm{p3}\:
                                    !\:\mathrm{hello}}~\mathsf{in}~\ldots,\nil)}\\
        & & \comp~\tuple{\mathrm{p2},(\id,\mathsf{apply}~\mathrm{echo}/0~\nil),\nil} \\
        & & \comp~\tuple{\mathrm{p3},(\id,\mathsf{apply}~\mathrm{target}/0~\nil),\nil} \\[1ex] 

        \hoo_\mathit{Send} & \{m_1\}; & \tuple{\mathrm{p1},
          (\{P2\mapsto\mathrm{p2},P3\mapsto\mathrm{p3}\},\underline{\mathsf{let}~\_=\mathrm{hello
                                        }~\mathsf{in}~\ldots},\nil)}\\
        & & \comp~\tuple{\mathrm{p2},(\id,\mathsf{apply}~\mathrm{echo}/0~\nil),\nil} \\
        & & \comp~\tuple{\mathrm{p3},(\id,\mathsf{apply}~\mathrm{target}/0~\nil),\nil} \\[1ex] 

        \hoo_\mathit{Seq} & \{m_1\}; & \tuple{\mathrm{p1},
          (\{P2\mapsto\mathrm{p2},P3\mapsto\mathrm{p3}\},\underline{P2}\:!\:\{P3,\mathrm{world}\},\nil)}\\
        & & \comp~\tuple{\mathrm{p2},(\id,\mathsf{apply}~\mathrm{echo}/0~\nil),\nil} \\
        & & \comp~\tuple{\mathrm{p3},(\id,\mathsf{apply}~\mathrm{target}/0~\nil),\nil} \\[1ex] 

        \hoo_\mathit{Seq} & \{m_1\}; & \tuple{\mathrm{p1},
          (\{P2\mapsto\mathrm{p2},P3\mapsto\mathrm{p3}\},\mathrm{p2}\:!\:\{\underline{P3},\mathrm{world}\},\nil)}\\
        & & \comp~\tuple{\mathrm{p2},(\id,\mathsf{apply}~\mathrm{echo}/0~\nil),\nil} \\
        & & \comp~\tuple{\mathrm{p3},(\id,\mathsf{apply}~\mathrm{target}/0~\nil),\nil} \\[1ex] 

        \hoo_\mathit{Seq} & \{m_1\}; & \tuple{\mathrm{p1},
          (\{P2\mapsto\mathrm{p2},P3\mapsto\mathrm{p3}\},\underline{\mathrm{p2}\:!\:\{\mathrm{p3},\mathrm{world}\}},\nil)}\\
        & & \comp~\tuple{\mathrm{p2},(\id,\mathsf{apply}~\mathrm{echo}/0~\nil),\nil} \\
        & & \comp~\tuple{\mathrm{p3},(\id,\mathsf{apply}~\mathrm{target}/0~\nil),\nil} \\[1ex] 

        \hoo_\mathit{Send} & \{m_1,m_2\}; & \tuple{\mathrm{p1},
          (\{P2\mapsto\mathrm{p2},P3\mapsto\mathrm{p3}\},\{\mathrm{p3},\mathrm{world}\},\nil)}\\
        & & \comp~\tuple{\mathrm{p2},(\id,\underline{\mathsf{apply}~\mathrm{echo}/0~\nil}),\nil} \\
        & & \comp~\tuple{\mathrm{p3},(\id,\mathsf{apply}~\mathrm{target}/0~\nil),\nil} \\[1ex] 

        \hoo_\mathit{Seq} & \{m_1,m_2\}; & \tuple{\mathrm{p1},
          (\{P2\mapsto\mathrm{p2},P3\mapsto\mathrm{p3}\},\{\mathrm{p3},\mathrm{world}\},\nil)}\\
        & &
              \comp~ \tuple{\mathrm{p2},(\id,\mathsf{receive}~\{P,M\}\to P\:!\:M~\mathsf{end}),\nil} \\
        & & \comp~\tuple{\mathrm{p3},(\id, \underline{\mathsf{apply}~\mathrm{target}/0~\nil}),\nil} \\[1ex] 

        \hoo_\mathit{Seq} & \{m_1,m_2\}; & \tuple{\mathrm{p1},
          (\{P2\mapsto\mathrm{p2},P3\mapsto\mathrm{p3}\},\{\mathrm{p3},\mathrm{world}\},\nil)}\\
        & &
              \comp~ \tuple{\mathrm{p2},(\id,\mathsf{receive}~\{P,M\}\to P\:!\:M~\mathsf{end}),\nil} \\
        & &
              \comp~ \tuple{\mathrm{p3},(\id,\mathsf{receive}~A\to\ldots~\mathsf{end}),\nil} \\
      \end{array}
    \]
    \caption{A derivation from ``$ \mathsf{apply}~\mathrm{main}/0~() $'', where
      $m_1 = (\mathrm{p3},\mathrm{hello})$,
      $m_2 = (\mathrm{p2},\{\mathrm{p3},\mathrm{world}\}
      )$, and $m_3 = (\mathrm{p3},\mathrm{world})$ (part
      1/2)} \label{fig:ex1b1}
  \end{figure}

  \begin{figure}[p]
    \scriptsize
    \[
      \begin{array}{l@{~}l@{~}l}
        \hoo_\mathit{Sched} & \{m_1\}; & \tuple{\mathrm{p1},
          (\{P2\mapsto\mathrm{p2},P3\mapsto\mathrm{p3}\},\{\mathrm{p3},\mathrm{world}\},\nil)}\\
        & &
              \comp~ \tuple{\mathrm{p2},(\id,\underline{\mathsf{receive}~\{P,M\}\to P\:!\:M~\mathsf{end}}), [\{\mathrm{p3},\mathrm{world}\}]} \\
        & & \comp~ \tuple{\mathrm{p3},(\id,\mathsf{receive}~A\to\ldots~\mathsf{end}),\nil} \\[1ex] 

        \hoo_\mathit{Receive} & \{m_1\}; & \tuple{\mathrm{p1},
          (\{P2\mapsto\mathrm{p2},P3\mapsto\mathrm{p3}\},\{\mathrm{p3},\mathrm{world}\},\nil)}\\
        & &
              \comp~ \tuple{\mathrm{p2},(\{P\mapsto \mathrm{p3},M\mapsto \mathrm{world}\},\underline{P}\:!\:M), \nil} \\
        & & \comp~ \tuple{\mathrm{p3},(\id,\mathsf{receive}~A\to\ldots~\mathsf{end}),\nil} \\[1ex] 

        \hoo_\mathit{Seq} & \{m_1\}; & \tuple{\mathrm{p1},
          (\{P2\mapsto\mathrm{p2},P3\mapsto\mathrm{p3}\},\{\mathrm{p3},\mathrm{world}\},\nil)}\\
        & &
              \comp~ \tuple{\mathrm{p2},(\{P\mapsto \mathrm{p3},M\mapsto \mathrm{world}\},\mathrm{p3}\:!\:\underline{M}), \nil} \\
        & & \comp~ \tuple{\mathrm{p3},(\id,\mathsf{receive}~A\to\ldots~\mathsf{end}),\nil} \\[1ex] 

        \hoo_\mathit{Seq} & \{m_1\}; & \tuple{\mathrm{p1},
          (\{P2\mapsto\mathrm{p2},P3\mapsto\mathrm{p3}\},\{\mathrm{p3},\mathrm{world}\},\nil)}\\
        & &
              \comp~ \tuple{\mathrm{p2},(\{P\mapsto \mathrm{p3},M\mapsto \mathrm{world}\},\underline{\mathrm{p3}\:!\:\mathrm{world}}), \nil} \\
        & & \comp~ \tuple{\mathrm{p3},(\id,\mathsf{receive}~A\to\ldots~\mathsf{end}),\nil} \\[1ex] 

        \hoo_\mathit{Send} & \{m_1,m_3\}; & \tuple{\mathrm{p1},
          (\{P2\mapsto\mathrm{p2},P3\mapsto\mathrm{p3}\},\{\mathrm{p3},\mathrm{world}\},\nil)}\\
        & &
              \comp~ \tuple{\mathrm{p2},(\{P\mapsto \mathrm{p3},M\mapsto \mathrm{world}\},\mathrm{world}), \nil} \\
        & & \comp~ \tuple{\mathrm{p3},(\id,\mathsf{receive}~A\to\ldots~\mathsf{end}),\nil} \\[1ex] 

        \hoo_\mathit{Sched} & \{m_1\}; & \tuple{\mathrm{p1},
          (\{P2\mapsto\mathrm{p2},P3\mapsto\mathrm{p3}\},\{\mathrm{p3},\mathrm{world}\},\nil)}\\
        & &
              \comp~ \tuple{\mathrm{p2},(\{P\mapsto \mathrm{p3},M\mapsto \mathrm{world}\},\mathrm{world}), \nil} \\
        & & \comp~ \tuple{\mathrm{p3},(\id,\underline{\mathsf{receive}~A\to\ldots~\mathsf{end}}),[\mathrm{world}]} \\[1ex] 

        \hoo_\mathit{Receive} & \{m_1\}; & \tuple{\mathrm{p1},
          (\{P2\mapsto\mathrm{p2},P3\mapsto\mathrm{p3}\},\{\mathrm{p3},\mathrm{world}\},\nil)),\nil}\\
        & &
              \comp~ \tuple{\mathrm{p2},(\{P\mapsto \mathrm{p3},M\mapsto \mathrm{world}\},\mathrm{world}), \nil} \\
        & & \comp~ \tuple{\mathrm{p3},(\{A\mapsto \mathrm{world} \},\mathsf{receive}~B\to \{A,B\}~\mathsf{end}),\nil} \\[1ex] 

        \hoo_\mathit{Sched} & \{\:\}; & \tuple{\mathrm{p1},
          (\{P2\mapsto\mathrm{p2},P3\mapsto\mathrm{p3}\},\{\mathrm{p3},\mathrm{world}\},\nil)),\nil}\\
        & &
              \comp~ \tuple{\mathrm{p2},(\{P\mapsto \mathrm{p3},M\mapsto \mathrm{world}\},\mathrm{world}), \nil} \\
        & & \comp~ \tuple{\mathrm{p3},(\{A\mapsto \mathrm{world} \},\underline{\mathsf{receive}~B\to \{A,B\}~\mathsf{end}}),[\mathrm{hello}]} \\[1ex] 

       \hoo_\mathit{Receive} & \{\:\}; & \tuple{\mathrm{p1},
          (\{P2\mapsto\mathrm{p2},P3\mapsto\mathrm{p3}\},\{\mathrm{p3},\mathrm{world}\},\nil)),\nil}\\
        & &
              \comp~ \tuple{\mathrm{p2},(\{P\mapsto \mathrm{p3},M\mapsto \mathrm{world}\},\mathrm{world}), \nil} \\
        & & \comp~ \tuple{\mathrm{p3},(\{A\mapsto
            \mathrm{world},B\mapsto \mathrm{hello} \},\{\underline{A},B\}),\nil} \\[1ex] 

       \hoo_\mathit{Seq} & \{\:\}; & \tuple{\mathrm{p1},
          (\{P2\mapsto\mathrm{p2},P3\mapsto\mathrm{p3}\},\{\mathrm{p3},\mathrm{world}\},\nil)),\nil}\\
        & &
              \comp~ \tuple{\mathrm{p2},(\{P\mapsto \mathrm{p3},M\mapsto \mathrm{world}\},\mathrm{world}), \nil} \\
        & & \comp~ \tuple{\mathrm{p3},(\{A\mapsto
            \mathrm{world},B\mapsto \mathrm{hello} \},\{\mathrm{world},\underline{B}\}),\nil} \\[1ex] 

       \hoo_\mathit{Seq} & \{\:\}; & \tuple{\mathrm{p1},
          (\{P2\mapsto\mathrm{p2},P3\mapsto\mathrm{p3}\},\{\mathrm{p3},\mathrm{world}\},\nil)),\nil}\\
        & &
              \comp~ \tuple{\mathrm{p2},(\{P\mapsto \mathrm{p3},M\mapsto \mathrm{world}\},\mathrm{world}), \nil} \\
        & & \comp~ \tuple{\mathrm{p3},(\{A\mapsto
            \mathrm{world},B\mapsto \mathrm{hello} \},\{\mathrm{world},\mathrm{hello}\}),\nil} \\[1ex]

      \end{array}
    \]
    \caption{A derivation from ``$ \mathsf{apply}~\mathrm{main}/0~() $'', where
      $m_1 = (\mathrm{p3},\mathrm{hello})$,
      $m_2 = (\mathrm{p2},\{\mathrm{p3},\mathrm{world}\}
      )$, and $m_3 = (\mathrm{p3},\mathrm{world})$ (part
      2/2)} \label{fig:ex1b2}
  \end{figure}
  
\end{example}

\subsection{Erlang Concurrency} \label{sec:concur}

In order to define a causal-consistent reversible semantics for Erlang
we need not only an interleaving semantics such as the one we just
presented, but also a notion of concurrency (or, equivalently, the
opposite notion of conflict). While concurrency is a main feature of
Erlang, as far as we know no formal definition of the concurrency
model of Erlang exists in the literature. We propose below one such
definition.

Given systems $s_1,s_2$, we call $s_1 \hoo^\ast s_2$ a
\emph{derivation}.
One-step derivations are simply called \emph{transitions}. We use
$d,d',d_1,\ldots$ to denote derivations and $t, t', t_1,\ldots$ for
transitions.
We label transitions as follows: $s_1 \hoo_{p,r} s_2$
where\footnote{Note that $p,r$ in $\hoo_{p,r}$ are not parameters of
  the transition relation $\hoo$ but just labels with some information
  on the reduction step. This information becomes useful to formally
  define the notion of concurrent transitions.}
\begin{itemize}
\item $p$ is the pid of the selected process in the
  transition or of the process to which a message is delivered (if the
  applied rule is $\mathit{Sched})$;
\item $r$ is the label of the applied transition rule.
\end{itemize}
We ignore some labels when they are clear from the context.

Given a derivation $d = (s_1 \hoo^\ast s_2)$, we define $\init(d) =
s_1$ and $\final(d) = s_2$.
Two derivations, $d_1$ and $d_2$, are \emph{composable} if
$\final(d_1) = \init(d_2)$. In this case, we let $d_1;d_2$ denote their
composition with $d_1;d_2 = (s_1 \hoo s_2 \hoo \cdots \hoo s_n \hoo
s_{n+1} \hoo \cdots \hoo s_m)$ if $d_1 = (s_1 \hoo s_2 \hoo \cdots
\hoo s_n)$ and $d_2 = (s_n \hoo s_{n+1} \hoo \cdots \hoo s_m)$.
Two derivations, $d_1$ and $d_2$, are said \emph{coinitial} if
$\init(d_1) = \init(d_2)$, and \emph{cofinal} if $\final(d_1) =
\final(d_2)$.

We let $\epsilon_s$ denote the zero-step derivation
$s\hoo^\ast s$.

\begin{definition}[Concurrent transitions] \label{def:concurrent1}
  Given two coinitial transitions, $t_1 = (s \hoo_{p_1,r_1} s_1)$
  and $t_2 = (s \hoo_{p_2,r_2} s_2)$, we say that they are
  \emph{in conflict} if they consider the same
    process, i.e., $p_1 = p_2$, and either $r_1 = r_2 =
    \mathit{Sched}$ or one transition applies rule $\mathit{Sched}$
    and the other transition applies rule $\mathit{Receive}$.
  Two coinitial transitions are \emph{concurrent} if they are not in
  conflict.
\end{definition}
We show below that our definition of concurrent transitions makes sense.

\begin{lemma}[Square lemma] \label{lemma:square1} Given two coinitial
  concurrent transitions $t_1 = (s \hoo_{p_1,r_1} s_1)$ and
  $t_2 = (s \hoo_{p_2,r_2} s_2)$, there exist two cofinal
  transitions $t_2/t_1 = (s_1 \hoo_{p_2,r_2} s')$ and
  $t_1/t_2 = (s_2 \hoo_{p_1,r_1} s')$.  Graphically,
  \[
  \begin{minipage}{50ex}
  \xymatrix@C=50pt@R=20pt{
   s \ar@{^{(}->}[r]^{p_1,r_1} \ar@<1pt>@{^{(}->}[d]_{p_2,r_2} & s_1\\
   s_2 & 
  }
  \end{minipage}
  ~~
  \Longrightarrow
  ~~
  \begin{minipage}{50ex}
  \xymatrix@C=50pt@R=20pt{
   s \ar@{^{(}->}[r]^{p_1,r_1} \ar@{^{(}->}[d]_{p_2,r_2} & s_1 \ar@{^{(}->}[d]^{p_2,r_2}\\
   s_2 \ar@{^{(}->}[r]_{p_1,r_1} &  s'
  }
  \end{minipage}
  \]
\end{lemma}

\begin{proof}
  We have the following cases:
  \begin{itemize}
  \item Two transitions $t_1$ and $t_2$ where $r_1\neq \mathit{Sched}$
    and $r_2\neq \mathit{Sched}$. Trivially, they apply to different
    processes, i.e., $p_1\neq p_2$.  Then, we can easily prove that by
    applying rule $r_2$ to $p_1$ in $s_1$ and rule $r_1$ to $p_2$ in
    $s_2$ we have two transitions $t_1/t_2$ and $t_2/t_1$ which
    are cofinal.

  \item One transition $t_1$ which applies rule $r_1 = \mathit{Sched}$
    to deliver message $v_1$ to process $p_1 = \mathrm{p}$,
    and another transition which applies a rule $r_2$ different from
    $\mathit{Sched}$. All cases but $r_2 = \mathit{Receive}$ with $p_2
    = \mathrm{p}$
    are straightforward.
    This last case, though, cannot happen since
    transitions using rules $\mathit{Sched}$ and $\mathit{Receive}$
    are not concurrent.

  \item Two transitions $t_1$ and $t_2$ with rules $r_1 = r_2 =
    \mathit{Sched}$ delivering messages $v_1$ and
    $v_2$, respectively. Since the transitions are
    concurrent, they should deliver the messages to different
    processes, i.e., $p_1\neq p_2$. Therefore, we can see
    that delivering $v_2$ from $s_1$ and $v_1$ from
    $s_2$ we get two cofinal transitions. \qed
  \end{itemize}
\end{proof}
We remark here that other definitions of concurrent transitions are
possible. Changing the concurrency model would require to change the
stored information in the reversible semantics in order to preserve
causal consistency. We have chosen the notion above since it is
reasonably simple to define and to work with,
and captures most of the pairs of coinitial
transitions that satisfy the Square lemma.

\section{A Reversible Semantics for Erlang} \label{sec:revsem}

In this section, we introduce a reversible---uncontrolled---semantics
for the considered language. Thanks to the modular design of the
concrete semantics, the transition rules for the language expressions
need not be changed in order to define the reversible semantics.

To be precise, in this section we introduce two transition relations:
$\rh$ and $\lh$. The first relation, $\rh$, is a conservative
extension of the standard semantics $\hoo$
(Figure~\ref{fig:system-rules}) to also include some additional
information in the states, following a typical Landauer
embedding. We refer to $\rh$ as the \emph{forward} reversible
semantics (or simply the forward semantics). In contrast, the second
relation, $\lh$, proceeds in the backward direction, ``undoing''
actions step by step. We refer to $\lh$ as the backward (reversible)
semantics.  We denote the union $\rh\cup\lh$ by $\rlh$.

In the next section, we will introduce a rollback operator that starts
a reversible computation for a process. In order to avoid undoing all
actions until the beginning of the process, we will also let the
programmer introduce \emph{checkpoints}. Syntactically, they are
denoted with the built-in function \textsf{check}, which takes an
identifier $\mathtt{t}$ as an argument. Such identifiers are supposed
to be unique in the program.  Given an expression, $expr$, we can
introduce a checkpoint by replacing $expr$ with ``$
\mathsf{let}~X=\mathsf{check}(\mathtt{t})~\mathsf{in}~expr $''\!\!.  A
call of the form $\mathsf{check}(\mathtt{t})$ just returns
$\mathtt{t}$ (see below).
In the following, we consider that the rules to evaluate the language
expressions (Figures~\ref{fig:seq-rules} and
\ref{fig:concurrent-rules}) are extended with the following rule:
\[
  (\mathit{Check}) ~~ 
  {\displaystyle 
    \frac{}{\theta,\mathsf{check}(\mathtt{t}) \arro{\mathsf{check}(\mathtt{t})} \theta,\mathtt{t}}}
\]
In this section, we will mostly ignore checkpoints, but they will
become relevant in the next section.

  \begin{figure}[t]
    \footnotesize
     \[
     \xymatrix@R=8pt@C=10pt{
         \underline{\mathrm{p1}} \ar@{..}[dd] 
        & \underline{\mathrm{p2}} \ar@{..}[dd] &
        \underline{\mathrm{p3}} \ar@{..}[d] \\
        & & \mathrm{p2}\:!\:\mathrm{v} \ar[dl] \ar@{..}[ddd]\\
        \mathrm{p2}\:!\:\mathrm{v} \ar@{..}[dd] \ar[dr] & \fbox{$t_1$} \ar@{..}[d] & \\
        & \fbox{$t_2$} \ar@{..}[d] & \\
        & & \\
      } 
\hspace{12ex}
     \xymatrix@R=8pt@C=10pt{
         \underline{\mathrm{p1}} \ar@{..}[dd] 
        & \underline{\mathrm{p2}} \ar@{..}[dd] &
        \underline{\mathrm{p3}} \ar@{..}[d] \\
        & & \mathrm{p2}\:!\:\mathrm{v} \ar[ddl] \ar@{..}[ddd]\\
        \mathrm{p2}\:!\:\mathrm{v} \ar@{..}[dd] \ar[r] & \fbox{$t_1$}
        \ar@{..}[d] & \\
        & \fbox{$t_2$} \ar@{..}[d]  & \\
        & & \\
      } 
      \]  
      \centering (a)\hspace{30ex}(b)
     \[
     \xymatrix@R=8pt@C=30pt{
         \underline{\mathrm{p1}} \ar@{..}[d] 
        & \underline{\mathrm{p2}} \ar@{..}[dd] \\
        \mathrm{p2}\:!\:\mathrm{v} \ar[dr] \ar@{..}[d]& \\
        \mathrm{p2}\:!\:\mathrm{v} \ar[dr] \ar@{..}[dd] &  \fbox{$t_1$}\ar@{..}[d]
        \\
        & \fbox{$t_2$} \ar@{..}[d] \\
        & \\
      } 
\hspace{12ex}
     \xymatrix@R=8pt@C=30pt{
         \underline{\mathrm{p1}} \ar@{..}[d] 
        & \underline{\mathrm{p2}} \ar@{..}[dd] \\
        \mathrm{p2}\:!\:\mathrm{v} \ar[ddr] \ar@{..}[d]& \\
        \mathrm{p2}\:!\:\mathrm{v} \ar[r] \ar@{..}[dd] & \fbox{$t_1$} \ar@{..}[d] \\
        & \fbox{$t_2$} \ar@{..}[d] \\
        & \\
      } 
      \]  
      \centering (c)\hspace{30ex}(d)
      \caption{Interleavings and the need for unique identifiers for
        messages} \label{fig:timestamps}
\end{figure}
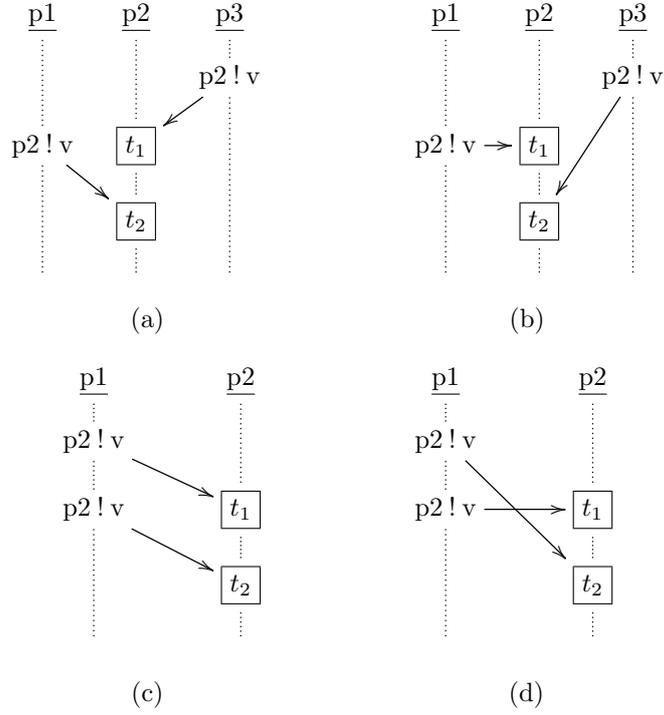

The most significant novelty in the forward semantics is that messages
now include a unique identifier (e.g., a timestamp $\k$). Let us illustrate
with some examples why we introduce these identifiers. Consider first
diagram (a) in Figure~\ref{fig:timestamps}, where two different
processes, $\mathrm{p1}$ and $\mathrm{p3}$, send the same message $v$
to process $\mathrm{p2}$. In order to undo the action $\mathrm{p2}\: !
\:\mathrm{v}$ in process $\mathrm{p3}$, one needs to first undo all
actions of $\mathrm{p2}$ up to $\fbox{$t_1$}$ (to ensure causal
consistency). However, currently, messages only store information
about the target process and the value sent, therefore it is not
possible to know whether it is safe to stop undoing actions at
$\fbox{$t_1$}$ or at $\fbox{$t_2$}$. Actually, the situations in
diagrams (a) and (b) are not distinguishable. In this case, it would
suffice to add the pid of the sender to every message in
order to avoid the confusion. However, this is not always
sufficient. Consider now diagram (c). Here, a process $\mathrm{p1}$
sends two identical messages to another process $\mathrm{p2}$ (which
is not unusual, say an ``ack'' after receiving a request). In this
case, in order to undo the first action $\mathrm{p2}\:!\:\mathrm{v}$ of
process $\mathrm{p1}$ one needs to undo all actions of process
$\mathrm{p2}$ up to $\fbox{$t_1$}$. However, we cannot distinguish
$\fbox{$t_1$}$ from $\fbox{$t_2$}$ unless some additional information
is taken into account (and considering triples of the form
$(source\_process\_pid,target\_process\_pid,message)$ would not
help). Therefore, one needs to introduce some unique identifier in
order to precisely distinguish case (c) from case (d).

Of course, we could have a less precise semantics where just the
message, $v$, is observable. However, that would make the backward
semantics unpredictable (e.g., we could often undo the ``wrong''
message delivery). Also, defining the corresponding notion of
\emph{conflicting} transitions (see Definition~\ref{def:concurrent} 
below) would be challenging, since one would like to have only a
conflict between the sending of a message $v$ and the ``last''
delivery of the same message $v$, which would be very
tricky. Therefore, in this paper, we prefer to assume that messages
can be uniquely distinguished.

\begin{figure}[p]
  \footnotesize
  \[
    \begin{array}{r@{~~}c}
    (\mathit{Seq}) & {\displaystyle
      \frac{\theta,e\arro{\tau} \theta',e'
      }{\Gamma;\tuple{p,\h,(\theta,e),q}\comp \Pi \rh
      \Gamma;\tuple{p,\tau(\theta,e)\cons\h,(\theta',e'),q}\comp \Pi}
      }\\[4ex]

    (\mathit{Check}) & {\displaystyle
      \frac{\theta,e \arro{\mathsf{check}(\mathtt{t})} \theta',e'}{\Gamma;\tuple{p,\h,(\theta,e),q} 
        \comp \Pi \rh \Gamma;\tuple{p,\mathsf{check}(\theta,e,\mathtt{t})\cons\h,(\theta',e'),q} 
        \comp \Pi }
      }\\[4ex]

    (\mathit{Send}) & {\displaystyle
      \frac{\theta,e \arro{\mathsf{send}(p'',v)}
        \theta',e'~~~\k~\mbox{is a fresh identifier}}{\Gamma;\tuple{p,\h,(\theta,e),q} 
        \comp \Pi \rh
           \Gamma\cup (p'',\{v,\k\});\tuple{p,\mathsf{send}(\theta,e,p'',\{v,\k\})\cons
           \h,(\theta',e'),q}\comp \Pi}
      }\\[4ex]

      (\mathit{Receive}) & {\displaystyle
        \frac{\theta,e \arro{\mathsf{rec}(\kappa,\ol{cl_n})}
          \theta',e'~~~ \mathsf{matchrec}(\theta,\ol{cl_n},q) = (\theta_i,e_i,\{v,\k\})}{\Gamma;\tuple{p,\h,(\theta,e),q}\comp \Pi \rh
          \Gamma;\tuple{p,\mathsf{rec}(\theta,e,\{v,\k\},q)\cons\h,(\theta'\theta_i,e'\{\kappa\mapsto e_i\}),q\backslash\!\!\backslash\{v,\k\}}\comp \Pi}
      }\\[4ex]
      
    (\mathit{Spawn}) & {\displaystyle
      \frac{\theta,e \arro{\mathsf{spawn}(\kappa,a/n,[\ol{v_n}])}
        \theta',e'~~~ p'~\mbox{is a fresh pid}}
{\begin{array}{ll}
          \Gamma;\tuple{p,\h,(\theta,e),q} 
        \comp \Pi \rh &
        \Gamma;\tuple{p,\mathsf{spawn}(\theta,e,p')\cons\h,(\theta',e'\{\kappa\mapsto
          p'\}),q} \\
        & \comp \tuple{p',\nil,(\id,\mathsf{apply}~a/n~(\ol{v_n})),\nil} 
        \comp \Pi
      \end{array}}
      }\\[6ex]

    (\mathit{Self}) & {\displaystyle
      \frac{\theta,e \arro{\mathsf{self}(\kappa)} \theta',e'}{\Gamma;\tuple{p,\h,(\theta,e),q} 
        \comp \Pi \rh \Gamma;\tuple{p,\mathsf{self}(\theta,e)\cons\h,(\theta',e'\{\kappa\mapsto p\}),q} 
        \comp \Pi }
      }\\[4ex]

      (\mathit{Sched}) & {\displaystyle
        \frac{~}{\Gamma\cup\{(p,\{v,\k\})\};\tuple{p,\h,(\theta,e),q}\comp\Pi 
          \rh \Gamma;\tuple{p,\h,(\theta,e),\{v,\k\}\cons q}\comp\Pi}
      }

  \end{array}
  \]
\caption{Forward reversible semantics} \label{fig:forwardsem}
\end{figure}

\begin{figure}[p]
  \footnotesize
  \[
  \begin{array}{r@{~~}c}
      (\mathit{\ol{Seq}}) & {\displaystyle
        \Gamma;
                            \tuple{p,\tau(\theta,e)\cons\h,(\theta',e'),q}
                            \:\comp\: \Pi
          \lh  \Gamma;\tuple{p,\h,(\theta,e),q}\:\comp\: \Pi
        }
      \\[2ex]

      (\mathit{\ol{Check}}) & {\displaystyle
        \Gamma;
                            \tuple{p,\mathsf{check}(\theta,e,\mathtt{t})\cons\h,(\theta',e'),q}
                            \:\comp\: \Pi
          \lh  \Gamma;\tuple{p,\h,(\theta,e),q}\:\comp\: \Pi
        } 
      \\[2ex]

    (\mathit{\ol{Send}}) & {\displaystyle
      \begin{array}{l}
        \Gamma\cup\{(p'',\{v,\k\})\};\tuple{p,\mathsf{send}(\theta,e,p'',\{v,\k\})\cons\h,(\theta',e'),q}\:\comp\:
        \Pi
        \lh \Gamma;\tuple{p,\h,(\theta,e),q} 
        \:\comp\: \Pi\\
    \end{array}
  }\\[3ex]

      (\mathit{\ol{Receive}}) & {\displaystyle
        \Gamma;\tuple{p,\mathsf{rec}(\theta,e,\{v,\k\},q)\cons\h,(\theta',e'), q\backslash\!\!\backslash\{v,\k\}}\:\comp\: \Pi
          \lh  \Gamma;\tuple{p,\h,(\theta,e),q}\:\comp\: \Pi
        }
      \\[2ex]
      
    (\mathit{\ol{Spawn}}) & {\displaystyle
      \begin{array}{l}
       \Gamma;\tuple{p,\mathsf{spawn}(\theta,e,p')\cons\h,(\theta',e'),q}
        \:\comp\: 
             \tuple{p',\nil,(\id,e''),\nil}
        \:\comp\: \Pi 
        \\
        \hspace{20ex}\lh
       \Gamma;\tuple{p,\h,(\theta,e),q}
        \:\comp\: \Pi 
      \end{array}
      }\\[3ex]

    (\mathit{\ol{Self}}) & {\displaystyle
       \Gamma;\tuple{p,\mathsf{self}(\theta,e)\cons\h,(\theta',e'),q} 
        \:\comp\: \Pi \lh
      \Gamma;\tuple{p,\h,(\theta,e),q} 
        \:\comp\: \Pi  
      }\\[2ex]

    (\mathit{\ol{Sched}}) & {\displaystyle
      \begin{array}{l}
      \Gamma;\tuple{p,\h,(\theta,e),\{v,\k\}\cons q} \: \comp\: \Pi \lh 
      \Gamma\cup(p,\{v,\k\});\tuple{p,\h,(\theta,e),q}\:\comp\:\Pi\\
      \hspace{20ex}\mbox{if the topmost $\mathsf{rec}(\ldots)$ item in
        $\h$ (if any) has the}\\ 
      \hspace{20ex}\mbox{form}~\mathsf{rec}(\theta',e',\{v',\k'\},q')~\mbox{with}~q'\backslash\!\!\backslash\{v',\k'\}\neq \{v,\k\}\cons q
    \end{array}
    }
  \end{array}
  \]
\caption{Backward reversible semantics} \label{fig:backwardsem}
\end{figure}

The transition rules of the forward reversible semantics can be found
in Figure~\ref{fig:forwardsem}. Processes now include a memory (or
\emph{history}) $\h$ that records the intermediate states of a process,
and messages have an associated unique identifier. In the memory, we use terms
headed by constructors $\tau$, $\mathsf{check}$, $\mathsf{send}$,
$\mathsf{rec}$, $\mathsf{spawn}$, and $\mathsf{self}$
to record the steps performed by the forward semantics. Note that
we could optimise the information stored in these terms by following a
strategy similar to that in \cite{MHNHT07,NPV16,TA15} for the
reversibility of functional expressions, but this is orthogonal to our
purpose in this paper, so we focus mainly on the concurrent actions.
Note also that the auxiliary function $\mathsf{matchrec}$ now deals
with messages of the form $\{v,\k\}$, which is a trivial
extension of the original function in the standard semantics by just
ignoring $\k$ when computing the first matching message.

\begin{example} \label{ex2} 

  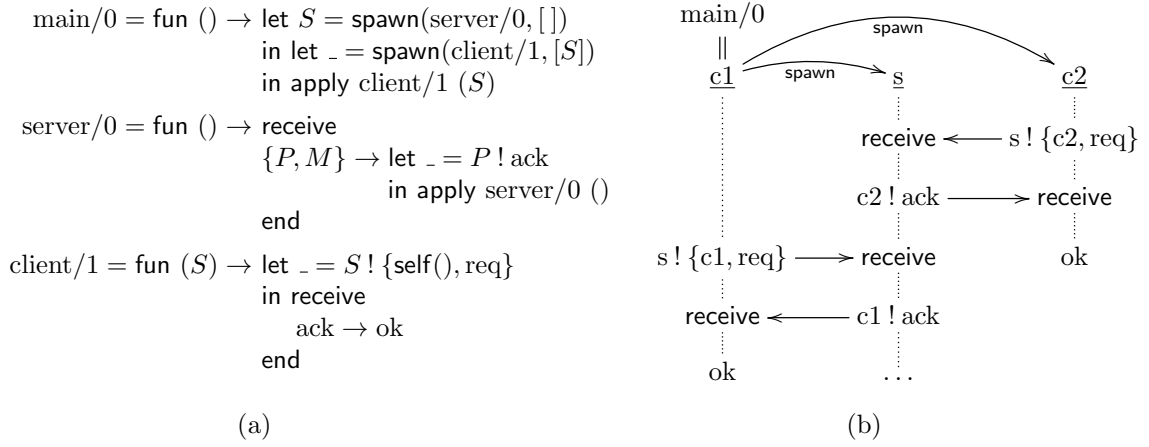
\begin{figure}[t]
    \footnotesize
    \centering
      \begin{minipage}{.45\linewidth}\hspace{-10ex}
        $
        \begin{array}{r@{~}ll}
        \mathrm{main}/0 = \mathsf{fun}~()\to & \mathsf{let}~S =
        \mathsf{spawn}(\mathrm{server}/0,\nil)\\
        & \mathsf{in}~\mathsf{let}~\_ = \mathsf{spawn}(\mathrm{client}/1,[S]) \\
        & \mathsf{in}~\mathsf{apply}~\mathrm{client}/1~(S)\\[1ex]

        \mathrm{server}/0 = \mathsf{fun}~()\to & \mathsf{receive}\\
        &\hspace{0ex}\{P,M\}\to \mathsf{let}~\_=P\:!\:
        \mathrm{ack}~\\
        &\hspace{11ex}\mathsf{in}~\mathsf{apply}~\mathrm{server}/0~()\\
        &\mathsf{end}\\[1ex]

        \mathrm{client}/1 = \mathsf{fun}~(S)\to & 
        \mathsf{let}~\_=S\:!\:\{\mathsf{self}(),\mathrm{req}\} \\
        & \mathsf{in}~\mathsf{receive}\\
        &\hspace{3ex} \mathrm{ack}\to \mathrm{ok}\\
        & \mathsf{end} \\
      \end{array}
      $
      \end{minipage}
      \hspace{5ex}
      \begin{minipage}{.35\linewidth}
      $
      \xymatrix@R=8pt@C=20pt{
        \mathrm{main}/0 \ar@{=}[d] & & \\
        \underline{\mathrm{c1}} \ar@{..}[ddd] \ar@/^/[r]_{\mathsf{spawn}}
        \ar@/^{8mm}/[rr]_{\mathsf{spawn}} 
        & \underline{\mathrm{s}} \ar@{..}[d] &
        \underline{\mathrm{c2}} \ar@{..}[d] \\
        & \mathsf{receive}
        \ar@{..}[d] &  \mathrm{s}\:!\:\{\mathrm{c2},\mathrm{req}\} \ar@{..}[d]
        \ar[l]\\
        & \mathrm{c2}\:!\: \mathrm{ack} \ar@{..}[d]
       \ar[r] & \mathsf{receive}
        \ar@{..}[d] \\
        \blue{\mathrm{s}\:!\: \{\mathrm{c1},\mathrm{req}\}} \ar@{..}[d] \ar[r] & \blue{\mathsf{receive}}
        \ar@{..}[d] &  \mathrm{ok} \\
        \blue{\mathsf{receive}} \ar@{..}[d]  &  \ar[l]  \blue{\mathrm{c1}\:!\:\mathrm{ack}} \ar@{..}[d] & \\
        \blue{\mathrm{ok}} & \ldots & \\ 
      } 
      $
    \end{minipage}\\[2ex]
    (a)\hspace{50ex}(b)
    \caption{A simple client-server} \label{fig:ex2-prog}
  \end{figure}

\begin{figure}[p]
    \scriptsize
    \[
      \begin{array}{l@{~}l@{~}l}
        & \{\:\}; & \l\mathrm{c1},\nil, (\id,C[\mathsf{apply}~\mathrm{main}/0~()]),\nil\r \\[1ex]

        \rh^\ast
        & \{\:\}; & \l\mathrm{c1},\nil, (\_,C[\mathsf{spawn}(\mathrm{server}/0,\nil)]),\nil\r \\[1ex]

        \rh_\mathit{Spawn} 
        & \{\:\}; & \l\mathrm{c1},[\mathsf{spawn}(\_,\_,\mathrm{s})],(\_,C[\underline{\mathsf{spawn}(\mathrm{client}/1,[\mathrm{s}])}]),\nil\r \\
        &&\comp\l\mathrm{s}, \nil,(\_,C[\mathsf{receive}~\{P,M\}\to\ldots]),\nil\r
        \\[1ex]

        \rh_\mathit{Spawn} 
        & \{\:\}; & \l\mathrm{c1},[\mathsf{spawn}(\_,\_,\mathrm{c2}),\mathsf{spawn}(\_,\_,\mathrm{s})],(\_,C[\mathrm{s}\:!\:\{\mathrm{c1},\mathrm{req}\}]),\nil\r \\
        &&\comp\l\mathrm{s}, \nil,(\_,C[\mathsf{receive}~\{P,M\}\to\ldots]),\nil\r\\
        &&\comp\l\mathrm{c2}, \nil,(\_,C[\underline{\mathrm{s}\:!\:\{\mathrm{c2}, \mathrm{req}\}}]),\nil\r
        \\[1ex]

        \rh_\mathit{Send} 
        & \{(\mathrm{s},m_1)\}; & \l\mathrm{c1},[\mathsf{spawn}(\_,\_,\mathrm{c2}),\mathsf{spawn}(\_,\_,\mathrm{s})],(\_,C[\mathrm{s}\:!\:\{\mathrm{c1},\mathrm{req}\}]),\nil\r \\
        &&\comp\l\mathrm{s}, \nil,(\_,C[\mathsf{receive}~\{P,M\}\to\ldots]),\nil\r\\
        &&\comp\l\mathrm{c2},[\mathsf{send}(\_,\_,\mathrm{s},m_1)],(\_,C[\mathsf{receive}~\mathrm{ack}\to \mathrm{ok}]),\nil\r
        \\[1ex]
 
       \rh_\mathit{Sched} 
        & \{\:\}; & \l\mathrm{c1},[\mathsf{spawn}(\_,\_,\mathrm{c2}),\mathsf{spawn}(\_,\_,\mathrm{s})],(\_,C[\mathrm{s}\:!\:\{\mathrm{c1},\mathrm{req}\}]),\nil\r \\
        &&\comp\l\mathrm{s}, \nil,(\_,C[\underline{\mathsf{receive}~\{P,M\}\to\ldots]}),[m_1]\r\\
        &&\comp\l\mathrm{c2},[\mathsf{send}(\_,\_,\mathrm{s},m_1)],(\_,C[\mathsf{receive}~\mathrm{ack}\to \mathrm{ok}]),\nil\r
        \\[1ex]

       \rh_\mathit{Receive} 
        & \{\:\}; & \l\mathrm{c1},[\mathsf{spawn}(\_,\_,\mathrm{c2}),\mathsf{spawn}(\_,\_,\mathrm{s})],(\_,C[\mathrm{s}\:!\:\{\mathrm{c1},\mathrm{req}\}]),\nil\r \\
        &&\comp\l\mathrm{s}, [\mathsf{rec}(\_,\_,m_1,[m_1])],(\_,C[\underline{\mathrm{c2}\:!\:\mathrm{ack}}]),\nil\r\\
        &&\comp\l\mathrm{c2},[\mathsf{send}(\_,\_,\mathrm{s},m_1)],(\_,C[\mathsf{receive}~\mathrm{ack}\to \mathrm{ok}]),\nil\r
        \\[1ex]

       \rh_\mathit{Send} 
        & \{(\mathrm{c2},m_2)\}; & \l\mathrm{c1},[\mathsf{spawn}(\_,\_,\mathrm{c2}),\mathsf{spawn}(\_,\_,\mathrm{s})],(\_,C[\mathrm{s}\:!\:\{\mathrm{c1},\mathrm{req}\}]),\nil\r \\
        &&\comp\l\mathrm{s}, [\mathsf{send}(\_,\_,\mathrm{c2},m_2),\mathsf{rec}(\_,\_,m_1,[m_1])],(\_,C[\mathsf{receive}~\{P,M\}\to\ldots]),\nil\r\\
        &&\comp\l\mathrm{c2},[\mathsf{send}(\_,\_,\mathrm{s},m_1)],(\_,C[\mathsf{receive}~\mathrm{ack}\to \mathrm{ok}]),\nil\r
        \\[1ex]

        \rh_\mathit{Sched} 
        & \{\;\}; & \l\mathrm{c1},[\mathsf{spawn}(\_,\_,\mathrm{c2}),\mathsf{spawn}(\_,\_,\mathrm{s})],(\_,C[\mathrm{s}\:!\:\{\mathrm{c1},\mathrm{req}\}]),\nil\r \\
        &&\comp\l\mathrm{s}, [\mathsf{send}(\_,\_,\mathrm{c2},m_2),\mathsf{rec}(\_,\_,m_1,[m_1])],(\_,C[\mathsf{receive}~\{P,M\}\to\ldots]),\nil\r\\
        &&\comp\l\mathrm{c2},[\mathsf{send}(\_,\_,\mathrm{s},m_1)],(\_,C[\underline{\mathsf{receive}~\mathrm{ack}\to \mathrm{ok}}]),[m_2]\r
        \\[1ex]

        \rh_\mathit{Receive} 
        & \{\;\}; & \l\mathrm{c1},[\mathsf{spawn}(\_,\_,\mathrm{c2}),\mathsf{spawn}(\_,\_,\mathrm{s})],(\_,C[\underline{\mathrm{s}\:!\:\{\mathrm{c1},\mathrm{req}\}}]),\nil\r \\
        &&\comp\l\mathrm{s}, [\mathsf{send}(\_,\_,\mathrm{c2},m_2),\mathsf{rec}(\_,\_,m_1,[m_1])],(\_,C[\mathsf{receive}~\{P,M\}\to\ldots]),\nil\r\\
        &&\comp\l\mathrm{c2},[\mathsf{rec}(\_,\_,m_2,[m_2]),\mathsf{send}(\_,\_,\mathrm{s},m_1)],(\_,\mathrm{ok}),\nil\r
        \\[1ex]

        \rh_\mathit{Send} 
        & \{(\mathrm{s},m_3)\}; & \l\mathrm{c1},[\mathsf{send}(\_,\_,\mathrm{s},m_3),\mathsf{spawn}(\_,\_,\mathrm{c2}),\mathsf{spawn}(\_,\_,\mathrm{s})],(\_,C[\mathsf{receive}~\mathrm{ack}\to\mathrm{ok}]),\nil\r \\
        &&\comp\l\mathrm{s}, [\mathsf{send}(\_,\_,\mathrm{c2},m_2),\mathsf{rec}(\_,\_,m_1,[m_1])],(\_,C[\mathsf{receive}~\{P,M\}\to\ldots]),\nil\r\\
        &&\comp\l\mathrm{c2},[\mathsf{rec}(\_,\_,m_2,[m_2]),\mathsf{send}(\_,\_,\mathrm{s},m_1)],(\_,\mathrm{ok}),\nil\r
        \\[1ex]

        \rh_\mathit{Sched} 
        & \{\;\}; & \l\mathrm{c1},[\mathsf{send}(\_,\_,\mathrm{s},m_3),\mathsf{spawn}(\_,\_,\mathrm{c2}),\mathsf{spawn}(\_,\_,\mathrm{s})],(\_,C[\mathsf{receive}~\mathrm{ack}\to\mathrm{ok}]),\nil\r \\
        &&\comp\l\mathrm{s}, [\mathsf{send}(\_,\_,\mathrm{c2},m_2),\mathsf{rec}(\_,\_,m_1,[m_1])],(\_,C[\underline{\mathsf{receive}~\{P,M\}\to\ldots}]),[m_3]\r\\
        &&\comp\l\mathrm{c2},[\mathsf{rec}(\_,\_,m_2,[m_2]),\mathsf{send}(\_,\_,\mathrm{s},m_1)],(\_,\mathrm{ok}),\nil\r
        \\[1ex]

        \rh_\mathit{Receive} 
        & \{\;\}; & \l\mathrm{c1},[\mathsf{send}(\_,\_,\mathrm{s},m_3),\mathsf{spawn}(\_,\_,\mathrm{c2}),\mathsf{spawn}(\_,\_,\mathrm{s})],(\_,C[\mathsf{receive}~\mathrm{ack}\to\mathrm{ok}]),\nil\r \\
        &&\comp\l\mathrm{s}, [\mathsf{rec}(\_,\_,m_3,[m_3]),\mathsf{send}(\_,\_,\mathrm{c2},m_2),\mathsf{rec}(\_,\_,m_1,[m_1])],(\_,C[\underline{\mathrm{c1}\:!\:\mathrm{ack}}]),\nil\r\\
        &&\comp\l\mathrm{c2},[\mathsf{rec}(\_,\_,m_2,[m_2]),\mathsf{send}(\_,\_,\mathrm{s},m_1)],(\_,\mathrm{ok}),\nil\r
        \\[1ex]

        \rh_\mathit{Send} 
        & \{(\mathrm{c1},m_4)\}; & \l\mathrm{c1},[\mathsf{send}(\_,\_,\mathrm{s},m_3),\mathsf{spawn}(\_,\_,\mathrm{c2}),\mathsf{spawn}(\_,\_,\mathrm{s})],(\_,C[\mathsf{receive}~\mathrm{ack}\to\mathrm{ok}]),\nil\r \\
        &&\comp\l\mathrm{s},
           [\mathsf{send}(\_,\_,\mathrm{c1},m_4),\mathsf{rec}(\_,\_,m_3,[m_3]),\mathsf{send}(\_,\_,\mathrm{c2},m_2),\mathsf{rec}(\_,\_,m_1,[m_1])],\\
        &&~~(\_,C[\mathsf{receive}~\{P,M\}\to\ldots]),\nil\r\\
        &&\comp\l\mathrm{c2},[\mathsf{rec}(\_,\_,m_2,[m_2]),\mathsf{send}(\_,\_,\mathrm{s},m_1)],(\_,\mathrm{ok}),\nil\r
        \\[1ex]

        \rh_\mathit{Sched} 
        & \{\;\}; & \l\mathrm{c1},[\mathsf{send}(\_,\_,\mathrm{s},m_3),\mathsf{spawn}(\_,\_,\mathrm{c2}),\mathsf{spawn}(\_,\_,\mathrm{s})],(\_,C[\underline{\mathsf{receive}~\mathrm{ack}\to\mathrm{ok}}]),[m_4]\r \\
        &&\comp\l\mathrm{s},
           [\mathsf{send}(\_,\_,\mathrm{c1},m_4),\mathsf{rec}(\_,\_,m_3,[m_3]),\mathsf{send}(\_,\_,\mathrm{c2},m_2),\mathsf{rec}(\_,\_,m_1,[m_1])],\\
        &&~~(\_,C[\mathsf{receive}~\{P,M\}\to\ldots]),\nil\r\\
        &&\comp\l\mathrm{c2},[\mathsf{rec}(\_,\_,m_2,[m_2]),\mathsf{send}(\_,\_,\mathrm{s},m_1)],(\_,\mathrm{ok}),\nil\r
        \\[1ex]

        \rh_\mathit{Receive} 
        & \{\;\}; & \l\mathrm{c1},[\mathsf{rec}(\_,\_,m_4,[m_4]),\mathsf{send}(\_,\_,\mathrm{s},m_3),\mathsf{spawn}(\_,\_,\mathrm{c2}),\mathsf{spawn}(\_,\_,\mathrm{s})],(\_,\mathrm{ok}),\nil\r \\
        &&\comp\l\mathrm{s},
           [\mathsf{send}(\_,\_,\mathrm{c1},m_4),\mathsf{rec}(\_,\_,m_3,[m_3]),\mathsf{send}(\_,\_,\mathrm{c2},m_2),\mathsf{rec}(\_,\_,m_1,[m_1])],\\
        &&~~(\_,C[\mathsf{receive}~\{P,M\}\to\ldots]),\nil\r\\
        &&\comp\l\mathrm{c2},[\mathsf{rec}(\_,\_,m_2,[m_2]),\mathsf{send}(\_,\_,\mathrm{s},m_1)],(\_,\mathrm{ok}),\nil\r
        \\
      \end{array}
    \]
    \caption{A derivation under the forward semantics, with
      $m_1 = \{\{\mathrm{c2},\mathrm{req}\},1\}$,
      $m_2 = \{\mathrm{ack},2\}$,
      $m_3 = \{\{\mathrm{c1},\mathrm{req}\},3\}$, and $m_4 = \{\mathrm{ack},4\}$.} \label{fig:ex2-der}
  \end{figure}
  Let us consider the program shown in Figure~\ref{fig:ex2-prog} (a),
  together with the execution trace sketched in
  Figure~\ref{fig:ex2-prog} (b).  
  Figure~\ref{fig:ex2-der}
  shows a high level account of the corresponding derivation under the
  forward semantics. For clarity, we consider the following
  conventions:
  \begin{itemize}
  \item Processes $\mathrm{client1}$, $\mathrm{client2}$ and
    $\mathrm{server}$ are denoted with $\mathrm{c1}$, $\mathrm{c2}$
    and $\mathrm{s}$, respectively.
  \item In the processes, we do not show the current
    environment. Moreover, we use the notation $C[e]$ to denote that
    $e$ is the redex to be reduced next and $C[\;]$ is an arbitrary
    (possibly empty) context. We also underline the selected redex
    when there are more than one (e.g., a redex in each process).
  \item In the histories, some arguments are denoted by ``$\_$'' since
    they are not relevant in the current derivation.
  \item Finally, we only show the steps performed with rules
    $\mathit{Spawn}$, $\mathit{Send}$, $\mathit{Receive}$ and
    $\mathit{Sched}$; the transition relation is labelled with the
    applied rule.
  \end{itemize}

\end{example}
We now prove that the forward semantics $\rh$ is a
conservative extension of the standard semantics $\hoo$.

In order to state the result, we let $\del(s)$ denote the system that
results from $s$ by removing the histories of the processes; formally,
$\del(\Gamma;\Pi) = \Gamma;\del'(\Pi)$, where
\[
  \begin{array}{lll}
    \del'(\tuple{p,\h,(\theta,e),q}) & = &
                                                  \tuple{p,(\theta,e),q}
    \\
    \del'(\tuple{p,\h,(\theta,e),q}\:\comp\:\Pi) & = &
                                                  \tuple{p,(\theta,e),q}\:\comp\:\del'(\Pi)
    \\
  \end{array}
\]
where we assume that $\Pi$ is not empty. 

We can now state the conservative extension result.

\begin{theorem} \label{th:conservative}
  Let $s_1$ be a system of the reversible semantics without
  occurrences of ``$\mathsf{check}$'' and $s'_1=\del(s_1)$ a system of
  the standard semantics.
  Then, $s'_1 \hoo^\ast s'_2$ iff $s_1 \rh^\ast s_2$
  and $\del(s_2) = s'_2$.
\end{theorem}

\begin{proof}
  The proof is straightforward since the transition rules of the
  forward semantics in Figure~\ref{fig:forwardsem} are just annotated
  versions of the corresponding rules in
  Figure~\ref{fig:system-rules}.  The only tricky point is noticing
  that the introduction of unique identifiers for messages does not
  change the behaviour of rule $\mathit{Receive}$ since function
  $\mathsf{matchrec}$ always returns the oldest occurrence (in terms
  of position in the queue) of the selected message. \qed
\end{proof}
The transition rules of the backward semantics are shown in
Figure~\ref{fig:backwardsem}.  In general, all rules restore the
control (and, if it applies, also the queue) of the
process. Nevertheless, let us briefly discuss a few particular
situations:
\begin{itemize}
\item First, observe that rule $\ol{\mathit{Send}}$ can only be
  applied when the message sent is in the global
  mailbox. If this is not the case (i.e., the message has been
  delivered using rule $\mathit{Sched}$), then we should first apply
  backward steps to the receiver process until, eventually, the
  application of rule $\ol{\mathit{Sched}}$ puts the message back into
  the global mailbox and rule $\ol{\mathit{Send}}$ becomes
  applicable. This is required to ensure causal consistency. In the
  next section, we will introduce a particular strategy that achieves
  this effect in a controlled manner.

\item A similar situation occurs with rule
  $\ol{\mathit{Spawn}}$. Given a process $p$ with a history item
  $\mathsf{spawn}(\theta,e,p')$, rule $\ol{\mathit{Spawn}}$ cannot be
  applied until the history and the queue of process $p'$ are both empty. Therefore, one
  should first apply a number of backward steps to process $p'$
  in order to be able to undo the $\mathsf{spawn}$
  item.
  We note that there is no need to require that no message targeting
  the process $p'$ (which would become an \emph{orphan} message) is in the
  global mailbox: in order to send such a message the pid $p'$ is
  needed, hence the sending of the message depends on the
  $\mathsf{spawn}$ and, thus, it must be undone beforehand.

\item Observe too that rule $\ol{\mathit{Receive}}$ can only be
  applied when the queue of the process is exactly the same queue that
  was obtained after applying the corresponding (forward)
  $\mathit{Receive}$ step. This is necessary in order to ensure that
  the restored queue is indeed the right one (note that adding the
  message to an arbitrary queue would not work since we do not know
  the ``right'' position for the message).

\item In principle, there is some degree of freedom in the application
  of rule $\ol{\mathit{Sched}}$ since it does not interfere with the
  remaining rules, except for $\ol{\mathit{Receive}}$ and other
  applications of $\ol{\mathit{Sched}}$. Therefore, the application of
  rule $\ol{\mathit{Sched}}$ can be switched with the application of
  any other backward rule except for $\ol{\mathit{Receive}}$ or another
  $\ol{\mathit{Sched}}$. The fact that two $\ol{\mathit{Sched}}$
  (involving the same process) do not commute is ensured since
  $\ol{\mathit{Sched}}$ always applies to the most recent message of a
  queue.
  The fact that a $\ol{\mathit{Sched}}$ and a $\ol{\mathit{Receive}}$
  do not commute is ensured since the side condition of
  $\ol{\mathit{Sched}}$ checks that there is no $\mathsf{rec}(\ldots)$
  item in the history of the process that can be used to apply rule
  $\ol{Receive}$ with the current queue. 
  Hence, their applicability conditions do not overlap.
\end{itemize}

\begin{example} \label{ex2b}
  \begin{figure}[p]
    \scriptsize
    \[
      \begin{array}{l@{~}l@{~}l}
        & \{\;\}; & \l\mathrm{c1},[\mathsf{rec}(\_,\_,m_4,[m_4]),\mathsf{send}(\_,\_,\mathrm{s},m_3),\mathsf{spawn}(\_,\_,\mathrm{c2}),\mathsf{spawn}(\_,\_,\mathrm{s})],(\_,\mathrm{ok}),\nil\r \\
        &&\comp\l\mathrm{s},
           [\mathsf{send}(\_,\_,\mathrm{c1},m_4),\mathsf{rec}(\_,\_,m_3,[m_3]),\mathsf{send}(\_,\_,\mathrm{c2},m_2),\mathsf{rec}(\_,\_,m_1,[m_1])],\\
        &&~~(\_,C[\mathsf{receive}~\{P,M\}\to\ldots]),\nil\r\\
        &&\comp\l\mathrm{c2},[\underline{\mathsf{rec}(\_,\_,m_2,[m_2])},\mathsf{send}(\_,\_,\mathrm{s},m_1)],(\_,\mathrm{ok}),\nil\r
        \\[1ex]

        \lh_{\ol{Receive}} 
        & \{\;\}; & \l\mathrm{c1},[\underline{\mathsf{rec}(\_,\_,m_4,[m_4])},\mathsf{send}(\_,\_,\mathrm{s},m_3),\mathsf{spawn}(\_,\_,\mathrm{c2}),\mathsf{spawn}(\_,\_,\mathrm{s})],(\_,\mathrm{ok}),\nil\r \\
        &&\comp\l\mathrm{s},
           [\mathsf{send}(\_,\_,\mathrm{c1},m_4),\mathsf{rec}(\_,\_,m_3,[m_3]),\mathsf{send}(\_,\_,\mathrm{c2},m_2),\mathsf{rec}(\_,\_,m_1,[m_1])],\\
        &&~~(\_,C[\mathsf{receive}~\{P,M\}\to\ldots]),\nil\r\\
        &&\comp\l\mathrm{c2},[\mathsf{send}(\_,\_,\mathrm{s},m_1)],(\_, C[\mathsf{receive}~\mathrm{ack}\to \mathrm{ok}]),[m_2]\r
        \\[1ex]

        \lh_{\ol{Receive}} 
        & \{\;\}; & \l\mathrm{c1},[\mathsf{send}(\_,\_,\mathrm{s},m_3),\mathsf{spawn}(\_,\_,\mathrm{c2}),\mathsf{spawn}(\_,\_,\mathrm{s})],(\_, C[\mathsf{receive}~\mathrm{ack}\to \mathrm{ok}]),[\underline{m_4}]\r \\
        &&\comp\l\mathrm{s},
           [\mathsf{send}(\_,\_,\mathrm{c1},m_4),\mathsf{rec}(\_,\_,m_3,[m_3]),\mathsf{send}(\_,\_,\mathrm{c2},m_2),\mathsf{rec}(\_,\_,m_1,[m_1])],\\
        &&~~(\_,C[\mathsf{receive}~\{P,M\}\to\ldots]),\nil\r\\
        &&\comp\l\mathrm{c2},[\mathsf{send}(\_,\_,\mathrm{s},m_1)],(\_, C[\mathsf{receive}~\mathrm{ack}\to \mathrm{ok}]),[m_2]\r
        \\[1ex]

        \lh_{\ol{Sched}} 
        & \{(\mathrm{c1},m_4)\}; & \l\mathrm{c1},[\mathsf{send}(\_,\_,\mathrm{s},m_3),\mathsf{spawn}(\_,\_,\mathrm{c2}),\mathsf{spawn}(\_,\_,\mathrm{s})],(\_, C[\mathsf{receive}~\mathrm{ack}\to \mathrm{ok}]),\nil\r \\
        &&\comp\l\mathrm{s},
           [\mathsf{send}(\_,\_,\mathrm{c1},m_4),\mathsf{rec}(\_,\_,m_3,[m_3]),\mathsf{send}(\_,\_,\mathrm{c2},m_2),\mathsf{rec}(\_,\_,m_1,[m_1])],\\
        &&~~(\_,C[\mathsf{receive}~\{P,M\}\to\ldots]),\nil\r\\
        &&\comp\l\mathrm{c2},[\mathsf{send}(\_,\_,\mathrm{s},m_1)],(\_, C[\mathsf{receive}~\mathrm{ack}\to \mathrm{ok}]),[\underline{m_2}]\r
        \\[1ex]

        \lh_{\ol{Sched}} 
        & \{(\mathrm{c2},m_2),(\mathrm{c1},m_4)\}; & \l\mathrm{c1},[\mathsf{send}(\_,\_,\mathrm{s},m_3),\mathsf{spawn}(\_,\_,\mathrm{c2}),\mathsf{spawn}(\_,\_,\mathrm{s})],(\_, C[\mathsf{receive}~\mathrm{ack}\to \mathrm{ok}]),\nil\r \\
        &&\comp\l\mathrm{s},
           [\underline{\mathsf{send}(\_,\_,\mathrm{c1},m_4)},\mathsf{rec}(\_,\_,m_3,[m_3]),\mathsf{send}(\_,\_,\mathrm{c2},m_2),\mathsf{rec}(\_,\_,m_1,[m_1])],\\
        &&~~(\_,C[\mathsf{receive}~\{P,M\}\to\ldots]),\nil\r\\
        &&\comp\l\mathrm{c2},[\mathsf{send}(\_,\_,\mathrm{s},m_1)],(\_, C[\mathsf{receive}~\mathrm{ack}\to \mathrm{ok}]),\nil\r
        \\[1ex]

        \lh_{\ol{Send}} 
        & \{(\mathrm{c2},m_2)\}; & \l\mathrm{c1},[\mathsf{send}(\_,\_,\mathrm{s},m_3),\mathsf{spawn}(\_,\_,\mathrm{c2}),\mathsf{spawn}(\_,\_,\mathrm{s})],(\_, C[\mathsf{receive}~\mathrm{ack}\to \mathrm{ok}]),\nil\r \\
        &&\comp\l\mathrm{s},
           [\underline{\mathsf{rec}(\_,\_,m_3,[m_3])},\mathsf{send}(\_,\_,\mathrm{c2},m_2),\mathsf{rec}(\_,\_,m_1,[m_1])],(\_, C[\mathrm{c1}\:!\:\mathrm{ack}]),\nil\r\\
        &&\comp\l\mathrm{c2},[\mathsf{send}(\_,\_,\mathrm{s},m_1)],(\_, C[\mathsf{receive}~\mathrm{ack}\to \mathrm{ok}]),\nil\r
        \\[1ex]

        \lh_{\ol{Receive}} 
        & \{(\mathrm{c2},m_2)\}; & \l\mathrm{c1},[\mathsf{send}(\_,\_,\mathrm{s},m_3),\mathsf{spawn}(\_,\_,\mathrm{c2}),\mathsf{spawn}(\_,\_,\mathrm{s})],(\_, C[\mathsf{receive}~\mathrm{ack}\to \mathrm{ok}]),\nil\r \\
        &&\comp\l\mathrm{s},
           [\mathsf{send}(\_,\_,\mathrm{c2},m_2),\mathsf{rec}(\_,\_,m_1,[m_1])],(\_, C[\mathsf{receive}~\{P,M\}\to\ldots]),[\underline{m_3}]\r\\
        &&\comp\l\mathrm{c2},[\mathsf{send}(\_,\_,\mathrm{s},m_1)],(\_, C[\mathsf{receive}~\mathrm{ack}\to \mathrm{ok}]),\nil\r
        \\[1ex]

        \lh_{\ol{Sched}} 
        & \{(\mathrm{s},m_3),(\mathrm{c2},m_2)\}; & \l\mathrm{c1},[\mathsf{send}(\_,\_,\mathrm{s},m_3),\mathsf{spawn}(\_,\_,\mathrm{c2}),\mathsf{spawn}(\_,\_,\mathrm{s})],(\_, C[\mathsf{receive}~\mathrm{ack}\to \mathrm{ok}]),\nil\r \\
        &&\comp\l\mathrm{s},
           [\underline{\mathsf{send}(\_,\_,\mathrm{c2},m_2)},\mathsf{rec}(\_,\_,m_1,[m_1])],(\_, C[\mathsf{receive}~\{P,M\}\to\ldots]),\nil\r\\
        &&\comp\l\mathrm{c2},[\mathsf{send}(\_,\_,\mathrm{s},m_1)],(\_, C[\mathsf{receive}~\mathrm{ack}\to \mathrm{ok}]),\nil\r
        \\[1ex]

        \lh_{\ol{Send}} 
        & \{(\mathrm{s},m_3)\}; & \l\mathrm{c1},[\mathsf{send}(\_,\_,\mathrm{s},m_3),\mathsf{spawn}(\_,\_,\mathrm{c2}),\mathsf{spawn}(\_,\_,\mathrm{s})],(\_, C[\mathsf{receive}~\mathrm{ack}\to \mathrm{ok}]),\nil\r \\
        &&\comp\l\mathrm{s},
           [\underline{\mathsf{rec}(\_,\_,m_1,[m_1])}],(\_, C[\mathrm{c2}\:!\:\mathrm{ack}]),\nil\r\\
        &&\comp\l\mathrm{c2},[\mathsf{send}(\_,\_,\mathrm{s},m_1)],(\_, C[\mathsf{receive}~\mathrm{ack}\to \mathrm{ok}]),\nil\r
        \\[1ex]

        \lh_{\ol{Receive}} 
        & \{(\mathrm{s},m_3)\}; & \l\mathrm{c1},[\mathsf{send}(\_,\_,\mathrm{s},m_3),\mathsf{spawn}(\_,\_,\mathrm{c2}),\mathsf{spawn}(\_,\_,\mathrm{s})],(\_, C[\mathsf{receive}~\mathrm{ack}\to \mathrm{ok}]),\nil\r \\
        &&\comp\l\mathrm{s},\nil,(\_, C[\mathsf{receive}~\{P,M\}\to\ldots]),[\underline{m_1}]\r\\
        &&\comp\l\mathrm{c2},[\mathsf{send}(\_,\_,\mathrm{s},m_1)],(\_, C[\mathsf{receive}~\mathrm{ack}\to \mathrm{ok}]),\nil\r
        \\[1ex]

        \lh_{\ol{Sched}} 
        & \{(\mathrm{s},m_1),(\mathrm{s},m_3)\}; & \l\mathrm{c1},[\mathsf{send}(\_,\_,\mathrm{s},m_3),\mathsf{spawn}(\_,\_,\mathrm{c2}),\mathsf{spawn}(\_,\_,\mathrm{s})],(\_, C[\mathsf{receive}~\mathrm{ack}\to \mathrm{ok}]),\nil\r \\
        &&\comp\l\mathrm{s},\nil,(\_, C[\mathsf{receive}~\{P,M\}\to\ldots]),\nil\r\\
        &&\comp\l\mathrm{c2},[\underline{\mathsf{send}(\_,\_,\mathrm{s},m_1)}],(\_, C[\mathsf{receive}~\mathrm{ack}\to \mathrm{ok}]),\nil\r
        \\[1ex]

        \lh_{\ol{Send}} 
        & \{(\mathrm{s},m_3)\}; & \l\mathrm{c1},[\underline{\mathsf{send}(\_,\_,\mathrm{s},m_3)},\mathsf{spawn}(\_,\_,\mathrm{c2}),\mathsf{spawn}(\_,\_,\mathrm{s})],(\_, C[\mathsf{receive}~\mathrm{ack}\to \mathrm{ok}]),\nil\r \\
        &&\comp\l\mathrm{s},\nil,(\_, C[\mathsf{receive}~\{P,M\}\to\ldots]),\nil\r\\
        &&\comp\l\mathrm{c2},\nil,(\_, C[\mathrm{s}\:!\:\{\mathrm{c2},\mathrm{req}\}]),\nil\r
        \\[1ex]

        \lh_{\ol{Send}} 
        & \{\;\}; & \l\mathrm{c1},[\underline{\mathsf{spawn}(\_,\_,\mathrm{c2})},\mathsf{spawn}(\_,\_,\mathrm{s})],(\_, C[\mathrm{s}\:!\:\{\mathrm{c1},\mathrm{req}\}]),\nil\r \\
        &&\comp\l\mathrm{s},\nil,(\_, C[\mathsf{receive}~\{P,M\}\to\ldots]),\nil\r\\
        &&\comp\l\mathrm{c2},\nil,(\_, C[\mathrm{s}\:!\:\{\mathrm{c2},\mathrm{req}\}]),\nil\r
        \\[1ex]

        \lh_{\ol{Spawn}} 
        & \{\;\}; & \l\mathrm{c1},[\underline{\mathsf{spawn}(\_,\_,\mathrm{s})}],(\_, C[\mathsf{spawn}(\mathrm{client}/1,[\mathrm{s}])]),\nil\r \\
        &&\comp\l\mathrm{s},\nil,(\_, C[\mathsf{receive}~\{P,M\}\to\ldots]),\nil\r\\[1ex]

        \lh_{\ol{Spawn}} 
        & \{\;\}; & \l\mathrm{c1},\nil,(\_, C[\mathsf{spawn}(\mathrm{server}/0,\nil)]),\nil\r \\[1ex]

        \lh^\ast 
        & \{\;\}; & \l\mathrm{c1},\nil,(\_, C[\mathsf{apply}~\mathrm{main}/0~()]),\nil\r \\
      \end{array}
    \]
    \caption{A derivation under the backward semantics, with
      $m_1 = \{\{\mathrm{c2},\mathrm{req}\},1\}$,
      $m_2 = \{\mathrm{ack},2\}$,
      $m_3 = \{\{\mathrm{c1},\mathrm{req}\},3\}$, and $m_4 = \{\mathrm{ack},4\}$.} \label{fig:backder}
  \end{figure}

  Consider again the program shown in Figure~\ref{fig:ex2-prog}. By
  starting from the last system in the forward derivation shown in
  Figure~\ref{fig:ex2-der}, we may construct the backward derivation
  shown in Figure~\ref{fig:backder}. Observe that it does not strictly
  follow the inverse order of the derivation shown in
  Figure~\ref{fig:ex2-der}. Actually, a derivation that undoes the
  steps in the precise inverse order exists, but it is not the only
  possibility. We will characterise later on (see Corollary~\ref{cor:consequences}) which orders are allowed
  and which are not. In Figure~\ref{fig:backder}, besides following the same conventions of
  Example~\ref{ex2}, for clarity, we underline the selected history
  item to be undone or the element in the queue to be removed (when
  the applied rule is $\mathit{\ol{Sched}}$).
\end{example}
\subsection{Properties of the Uncontrolled Reversible Semantics}
In the following, we prove several properties of our reversible
semantics, including its \emph{causal consistency}, an essential
property for reversible concurrent calculi \cite{DK04}.

Given systems $s_1,s_2$, we call $s_1 \rh^\ast s_2$ a \emph{forward}
derivation and $s_2 \lh^\ast s_1$ a \emph{backward} derivation.  A
derivation potentially including both forward and backward steps is
denoted by $s_1 \rlh^\ast s_2$.
We label transitions as follows: $s_1 \rlh_{p,r,k} s_2$ where
\begin{itemize}
\item $p,r$ are the pid of the selected process and the label of the
  applied rule, respectively, as in Section~\ref{sec:concur},
\item $k$ is a history item if the applied rule was different from
  $\mathit{Sched}$ and $\ol{\mathit{Sched}}$, and 
\item $k=\mathsf{sched}(\{v,\k\})$ when the applied rule was
  $\mathit{Sched}$ or $\ol{\mathit{Sched}}$, where $\{v,\k\}$ is the
  message delivered or put back into $\Gamma$. Note that this
  information is available when applying the rule.
\end{itemize}
We ignore some labels when they are clear from the context.

We extend the definitions of functions $\init$ and $\final$ from
Section~\ref{sec:concur} to reversible derivations in the natural way.
The notions of composable, coinitial and cofinal derivations are
extended also in a straightforward manner.

Given a rule label $r$, we let $\ol{r}$ denote its reverse version, i.e.,
if $r=\mathit{Send}$ then $\ol{r} = \ol{\mathit{Send}}$ and vice
versa (if $r=\ol{\mathit{Send}}$ then $\ol{r} = \mathit{Send}$).
Also, given a transition $t$, we let $\ol{t} = (s' \lh_{p,\ol{r},k}
s)$ if $t= (s \rh_{p,r,k} s')$ and $\ol{t} = (s' \rh_{p,r,k} s)$ if
$t= (s \lh_{p,\ol{r},k} s')$. We say that $\ol{t}$ is the \emph{inverse} of $t$.
This notation is naturally extended to derivations. 
We let $\epsilon_s$ denote the zero-step derivation
$s\rlh^\ast s$.

In the following we restrict the attention to systems reachable from
the execution of a program:

\begin{definition}[Reachable systems]\label{def:reachable}
  A system is \emph{initial} if it is composed by a single process,
  and this process has an empty history and an empty queue; furthermore the
  global mailbox is empty. A system $s$ is \emph{reachable} if there
  exists an initial system $s_0$ and a derivation $s_0 \rlh^\ast s$
  using the rules corresponding to a given program.
\end{definition}
Moreover, for simplicity, we also consider an implicit, fixed program
in the technical results, that is we fix the function $\mu$ in the
semantics of expressions.

The next lemma proves that every forward (resp.\ backward) transition
can be undone by a backward (resp.\ forward) transition. 

\begin{lemma}[Loop lemma] \label{lemma:loop} For every pair of
  reachable systems, $s_1$ and $s_2$, we have $s_1 \rh_{p,r,k} s_2$ iff $s_2
  \lh_{p,\ol{r},k} s_1$.
\end{lemma}

\begin{proof}
  The proof is by case analysis on the applied rule.
  We discuss below the most interesting cases.
  \begin{itemize}
    \item Rule $\mathit{Sched}$: notice that the queue of a process is changed only
      by rule $\mathit{Receive}$ (which removes messages) and
      $\mathit{Sched}$ (which adds messages). Since, after the last
      $\mathit{Receive}$ at least one message has been added, then the
      side condition of rule $\ol{\mathit{Sched}}$ is always verified.
    \item Rule $\ol{\mathit{Seq}}$: one has to check that the restored
      control $(\theta,e)$ can indeed perform a sequential step to
      $(\theta',e')$. This always holds for reachable systems. An
      analogous check needs to be done for all backward rules. \qed
  \end{itemize}
\end{proof}
The following notion of concurrent transitions allows us to
characterise which actions can be switched without changing the
semantics of a computation. It extends the same notion from the
standard semantics (cf.\ Definition~\ref{def:concurrent1}) to the
reversible semantics.

\begin{definition}[Concurrent transitions] \label{def:concurrent}
  Given two coinitial transitions, $t_1 = (s \rlh_{p_1,r_1,k_1} s_1)$
  and $t_2 = (s \rlh_{p_2,r_2,k_2} s_2)$, we say that they are
  \emph{in conflict} if at least one of the following conditions
  holds:
  \begin{itemize}

  \item \textbf{both transitions are forward}, they consider the same
    process, i.e., $p_1 = p_2$, and either $r_1 = r_2 =
    \mathit{Sched}$ or one transition applies rule $\mathit{Sched}$
    and the other transition applies rule $\mathit{Receive}$.

  \item one is a \textbf{forward} transition that applies to a process
    $\mathrm{p}$, say $p_1=\mathrm{p}$, and the other one is a
    \textbf{backward} transition that undoes the creation of
    $\mathrm{p}$, i.e., $p_2 = \mathrm{p}'\neq\mathrm{p}$, $r_2=
    \ol{\mathit{Spawn}}$ and $k_2 =
    \mathsf{spawn}(\theta,e,\mathrm{p})$ for some control $(\theta,e)$;

  \item one is a \textbf{forward} transition that delivers a message
    $\{v,\k\}$ to a process $\mathrm{p}$, say $p_1 = \mathrm{p}$,
    $r_1 = \mathit{Sched}$ and $k_1 = \mathsf{sched}(\{v,\k\})$, and
    the other one is a \textbf{backward} transition that undoes the
    sending $\{v,\k\}$ to $\mathrm{p}$, i.e., $p_2 = \mathrm{p}'$
    (note that $\mathrm{p}=\mathrm{p}'$ if the message is sent to its
    own sender), $r_2 = \ol{\mathit{Send}}$ and
    $k_2 = \mathsf{send}(\theta,e,\mathrm{p},\{v,\k\})$ for some
    control $(\theta,e)$;

  \item one is a \textbf{forward} transition and the other one is a
    \textbf{backward} transition such that $p_1=p_2$ and either i)
    both applied rules are different from both $\mathit{Sched}$ and
    $\ol{\mathit{Sched}}$, i.e., $\{r_1,r_2\} \cap \{\mathit{Sched},\ol{\mathit{Sched}}\} = \emptyset$; ii) one rule is $\mathit{Sched}$ and the
    other one is $\ol{\mathit{Sched}}$; iii) one rule is
    $\mathit{Sched}$ and the other one is $\ol{\mathit{Receive}}$; or
    iv) one rule is $\ol{\mathit{Sched}}$ and the other one is
    $\mathit{Receive}$.
  \end{itemize}
  Two coinitial transitions are \emph{concurrent} if they are not in
  conflict. Note that two coinitial backward transitions are always concurrent.
\end{definition}
The following lemma (the counterpart of Lemma~\ref{lemma:square} for
the standard semantics) is a key result to prove the causal
consistency of the semantics.

\begin{lemma}[Square lemma] \label{lemma:square} Given two coinitial
  concurrent transitions $t_1 = (s \rlh_{p_1,r_1,k_1} s_1)$ and
  $t_2 = (s \rlh_{p_2,r_2,k_2} s_2)$, there exist two cofinal
  transitions $t_2/t_1 = (s_1 \rlh_{p_2,r_2,k_2} s')$ and
  $t_1/t_2 = (s_2 \rlh_{p_1,r_1,k_1} s')$.  Graphically,
  \[
  \begin{minipage}{50ex}
  \xymatrix@C=50pt@R=20pt{
   s \ar@<1pt>@^{->}[r]^{p_1,r_1,k_1} \ar@<1pt>@^{->}[d]_{p_2,r_2,k_2} & s_1 \ar@<1pt>@^{->}[l]\\
   s_2 \ar@<1pt>@^{->}[u]& 
  }
  \end{minipage}
  ~~
  \Longrightarrow
  ~~
  \begin{minipage}{50ex}
  \xymatrix@C=50pt@R=20pt{
   s \ar@<1pt>@^{->}[r]^{p_1,r_1,k_1} \ar@<1pt>@^{->}[d]_{p_2,r_2,k_2} & s_1 \ar@<1pt>@^{->}[d]^{p_2,r_2,k_2} \ar@<1pt>@^{->}[l]\\
   s_2 \ar@<1pt>@^{->}[r]_{p_1,r_1,k_1} \ar@<1pt>@^{->}[u] &  s' \ar@<1pt>@^{->}[u] \ar@<1pt>@^{->}[l]
  }
  \end{minipage}
  \]
\end{lemma}

\begin{proof}
  We distinguish the following cases depending on the applied rules:\\[1ex]
  (1) Two forward transitions. Then, we have the following cases:
  \begin{itemize}
  \item Two transitions $t_1$ and $t_2$ where $r_1\neq \mathit{Sched}$
    and $r_2\neq \mathit{Sched}$. Trivially, they apply to different
    processes, i.e., $p_1\neq p_2$.  Then, we can easily prove that by
    applying rule $r_2$ to $p_1$ in $s_1$ and rule $r_1$ to $p_2$ in
    $s_2$ we have two transitions $t_1/t_2$ and $t_2/t_1$ which
    produce the corresponding history items and are cofinal.

  \item One transition $t_1$ which applies rule $r_1 = \mathit{Sched}$
    to deliver message $\{v_1,\k_1\}$ to process $p_1 = \mathrm{p}$,
    and another transition which applies a rule $r_2$ different from
    $\mathit{Sched}$. All cases but $r_2 = \mathit{Receive}$ with $p_2
    = \mathrm{p}$ and $k_2 = \mathsf{rec}(\theta,e,\{v_2,\k_2\},q)$
    are straightforward. Note that $\k_1\neq \k_2$ since these
    identifiers are unique. Here, by applying rule $\mathit{Receive}$
    to $s_1$ and rule $\mathit{Sched}$ to $s_2$ we will end up with
    the same mailbox in $\mathrm{p}$ (since it is a FIFO
    queue). However, the history item
    $\mathsf{rec}(\theta,e,\{v_2,\k_2\},q')$ will be necessarily
    different since $q\neq q'$ by the application of rule
    $\mathit{Sched}$. This situation, though, cannot happen since
    transitions using rules $\mathit{Sched}$ and $\mathit{Receive}$
    are not concurrent.

  \item Two transitions $t_1$ and $t_2$ with rules $r_1 = r_2 =
    \mathit{Sched}$ delivering messages $\{v_1,\k_1\}$ and
    $\{v_2,\k_2\}$, respectively. Since the transitions are
    concurrent, they should deliver the messages to different
    processes, i.e., $p_1\neq p_2$. Therefore, we can easily prove
    that delivering $\{v_2,\k_2\}$ from $s_1$ and $\{v_1,\k_1\}$ from
    $s_2$ we get two cofinal transitions.
  \end{itemize}
  (2) One forward transition and one backward transition. Then, we
  distinguish the following cases:
  \begin{itemize}
  \item If the two transitions apply to the same process, i.e., $p_1 =
    p_2$, then, since they are concurrent, we can only have $r_1 = \mathit{Sched}$ and a rule
    different from both $\ol{\mathit{Sched}}$ and
    $\ol{\mathit{Receive}}$, or $r_1 = \ol{\mathit{Sched}}$ and a rule
    different from both $\mathit{Sched}$ and $\mathit{Receive}$. In
    these cases, the claim follows easily by a case distinction on the
    applied rules.

  \item Let us now consider that the transitions apply to different
    processes, i.e., $p_1\neq p_2$, and the applied rules are
    different from $\mathit{Sched},\ol{\mathit{Sched}}$. In this case,
    the claim follows easily except when one transition considers a
    process $\mathrm{p}$ and the other one undoes the spawning of the
    same process $\mathrm{p}$. This case, however, is not allowed
    since the transitions are concurrent.

  \item Finally, let us consider that the transitions apply to
    different processes, i.e., $p_1\neq p_2$, and that one transition
    applies rule $\mathit{Sched}$ to deliver a message $\{v,\k\}$ from
    sender $\mathrm{p}$ to receiver $\mathrm{p}'$, i.e., $p_1 =
    \mathrm{p}'$, $r_1 = \mathit{Sched}$ and $k_1 =
    \mathsf{sched}(\{v,\k\})$. In this case, the other transition
    should apply a rule $r_2$ different from $\ol{\mathit{Send}}$ with
    $k_2 = \mathsf{send}(\theta,e,\mathrm{p}',\{v,\k\})$ for some control
    $(\theta,e)$ since, otherwise, the transitions would not be
    concurrent. In any other case, one can easily prove that by
    applying $r_2$ to $s_1$ and $\mathit{Sched}$ to $s_2$ we get
    two cofinal transitions.

  \end{itemize}
  (3) Two backward transitions. We distinguish the following cases:
  \begin{itemize}
  \item If the two transitions apply to different processes, the
    claim follows easily.

  \item Let us now consider that they apply to the same process, i.e.,
    $p_1 = p_2$ and that the applied rules are different from
    $\ol{\mathit{Sched}}$. This case is not possible since, given a
    system, only one backward transition rule different from
    $\ol{\mathit{Sched}}$ is applicable (i.e., the one that
    corresponds to the last item in the history).

  \item Let us consider that both transitions apply to the same
    process and that both are applications of rule
    $\ol{\mathit{Sched}}$. This case is not possible since rule
    $\ol{\mathit{Sched}}$ can only take the newest message from the
    local queue of the process, and thus only one rule
    $\ol{\mathit{Sched}}$ can be applied to a given process.

  \item Finally, consider that both transitions apply to the same
    process and only one of them applies rule
    $\ol{\mathit{Sched}}$. In this case, the only non-trivial case is
    when the other applied rule is $\ol{\mathit{Receive}}$, since both
    change the local queue of the process. However, this case is not
    allowed by the backward semantics, since the conditions to apply
    rule $\ol{\mathit{Sched}}$ and rule $\ol{\mathit{Receive}}$ are
    non-overlapping. \qed
  \end{itemize}
\end{proof}
\begin{corollary}[Backward confluence]\label{cor:backconf}
  Given two backward derivations $s \lh^\ast s_1$ and $s \lh^\ast s_2$
  there exist $s_3$ and two backward derivations $s_1 \lh^\ast s_3$ and
  $s_2 \lh^\ast s_3$.
\end{corollary}
\begin{proof}
  By iterating the square lemma (Lemma~\ref{lemma:square}), noticing
  that backward transitions are always concurrent.
  This is a standard result for abstract relations (see, e.g.,
  \cite{BN98} and the original work by Rosen \cite{Ros73}), where
  confluence is implied by the \emph{diamond property} (the square
  lemma in our work). \qed
\end{proof}
The notion of concurrent transitions for the reversible semantics is a
natural extension of the same notion for the standard semantics:

\begin{lemma}
  Let $t_1$ and $t_2$ be two forward coinitial transitions using the
  reversible semantics, and let $t'_1$ and $t'_2$ be their counterpart
  in the standard semantics obtained by removing the histories and the
  unique identifiers for messages. Then, $t_1$ and $t_2$ are
  concurrent iff $t'_1$ and $t'_2$ are.
\end{lemma}

\begin{proof}
  The proof is straightforward since Definition~\ref{def:concurrent1}
  and the first case of Definition~\ref{def:concurrent} are perfectly
  analogous. \qed
\end{proof}
The next result is used to switch the successive application of two
transition rules. Let us note that previous proof schemes of causal
consistency (e.g., \cite{DK04}) did not include such a result,
directly applying the square lemma instead. In our case, this would
not be correct.

\begin{lemma}[Switching lemma] \label{lemma:switch} Given two
  composable transitions of the form $t_1 = (s_1
  \rlh_{p_1,r_1,k_1} s_2)$ and $t_2 = (s_2
  \rlh_{p_2,r_2,k_2} s_3)$ such that $\ol{t_1}$ and
  $t_2$ are concurrent, there exist a system $s_4$ and two composable
  transitions $t'_1 = (s_1 \rlh_{p_2,r_2,k_2} s_4)$ and
  $t'_2 = (s_4 \rlh_{p_1,r_1,k_1} s_3)$.
\end{lemma}

\begin{proof}
  First, using the loop lemma (Lemma~\ref{lemma:loop}), we have
  $\ol{t_1} = (s_2 \rlh_{p_1,\ol{r_1},k_1} s_1)$. Now, since
  $\ol{t_1}$ and $t_2$ are concurrent, by applying the square lemma
  (Lemma~\ref{lemma:square}) to $\ol{t_1} = (s_2
  \rlh_{p_1,\ol{r_1},k_1} s_1)$ and $t_2 = (s_2 \rlh_{p_2,r_2,k_2}
  s_3)$, there exists a system $s_4$ such that $\ol{t'_1} =
  \ol{t_1}/t_2 = (s_3 \rlh_{p_1,\ol{r_1},k_1} s_4)$ and $t'_2 =
  t_2/\ol{t_1} = (s_1 \rlh_{p_2,r_2,k_2} s_4)$. Using the loop lemma
  (Lemma~\ref{lemma:loop}) again, we have $t'_1 = t_1/t_2 = (s_4
  \rlh_{p_1,r_1,k_1} s_3)$, which concludes the proof. \qed
\end{proof}

\begin{corollary} \label{corollary:switching}
  Given two composable transitions $t_1 = (s_1 \rh_{p_1,r_1,k_1} s_2)$
  and $t_2 = (s_2 \lh_{p_2,r_2,k_2} s_3)$, there exist a system $s_4$
  and two composable transitions $t'_1 = ( s_1 \lh_{p_2,r_2,k_2} s_4)$
  and $t'_2 = (s_4 \rh_{p_1,r_1,k_1} s_3)$. Graphically, 
  \[
  \begin{minipage}{50ex}
    \xymatrix@C=50pt@R=10pt{
      s_1 \ar@^{->}[r]^{p_1,r_1,k_1} & s_2 \\
    & s_3 \ar@^{->}[u]_{p_2,r_2,k_2}
    }
  \end{minipage}
  ~~
  \Longrightarrow
  ~~
  \begin{minipage}{50ex}
    \xymatrix@C=50pt@R=10pt{
      s_1 \ar@^{->}[r]^{p_1,r_1,k_1} & s_2 \\
      s_4 \ar@^{->}[r]_{p_1,r_1,k_1} \ar@^{->}[u]_{p_2,r_2,k_2} & s_3 \ar@^{->}[u]_{p_2,r_2,k_2}}
  \end{minipage}
  \]
\end{corollary}
\begin{proof}
  The corollary follows by applying the switching lemma
  (Lemma~\ref{lemma:switch}), noticing that two backward transitions
  are always concurrent. \qed
\end{proof}
We now formally define the notion of causal equivalence between
derivations, in symbols $\approx$, as the least equivalence relation
between transitions closed under composition that obeys the following
rules:
\[
t_1; t_2/t_1 \approx t_2 ; t_1/t_2
~~~~~~
t;\ol{t} \approx \epsilon_{\init(t)} 
\]
Causal equivalence amounts to say that those derivations that only
differ for swaps of concurrent actions or the removal of successive
inverse actions are equivalent. Observe that any of the notations $t_1;
t_2/t_1$ and $t_2 ; t_1/t_2 $ requires $t_1$ and $t_2$ to be
concurrent. 

\begin{lemma}[Rearranging lemma] \label{lemma:rearrange} Given systems
  $s,s'$, if $d = (s \rlh^\ast s')$, then there exists a system $s''$
  such that $d' = (s \lh^\ast s'' \rh^\ast s')$ and $d\approx
  d'$. Furthermore, $d'$ is not longer than $d$.
\end{lemma}

\begin{proof}
  The proof is by lexicographic induction on the length of $d$
  and on the number of steps from the earliest pair of
  transitions in $d$ of the form $s_1 \rh s_2 \lh s_3$ to $s'$.  If there is
  no such pair we are done. If $s_1 = s_3$, then $s_1 \rh s_2 = \ol{(s_2
  \lh s_3)}$.  Indeed, if $s_1 \rh s_2$ adds an item to the
  history of some process then $s_2 \lh s_3$ should remove the same
  item. Otherwise, $s_1 \rh s_2$ is an application of rule
  $\mathit{Sched}$ and $s_2 \lh s_3$ should undo the scheduling of the
  same message.  Then, we can remove these two transitions and the
  claim follows by induction since the resulting derivation is shorter
  and $(s_1 \rh s_2 \lh s_3) \approx \epsilon_{s_1}$. Otherwise, we
  apply Corollary~\ref{corollary:switching} commuting $s_2 \lh s_3$
  with all forward transitions preceding it in $d$.  If one such
  transition is its inverse, then we reason as above.  Otherwise, we
  obtain a new derivation $d' \approx d$ which has the same length of
  $d$, and where the distance between the earliest pair of
  transitions in $d'$ of the form $s'_1 \rh s'_2 \lh s'_3$ and $s'$ has decreased.
  The claim follows then by the inductive hypothesis.\qed
\end{proof}
An interesting consequence of the rearranging lemma
is the following result, which states
that every system obtained by both forward and backward steps from an
initial system, is also reachable by a forward-only derivation:

\begin{corollary} \label{cor:lopstr} Let $s$ be an initial system. For
  each derivation $s \rlh^\ast s'$, there exists a forward derivation
  of the form $s \rh^\ast s'$.
\end{corollary}
The following auxiliary result is also needed for proving causal
consistency.

\begin{lemma}[Shortening lemma] \label{lemma:short}
  Let $d_1$ and $d_2$ be coinitial and cofinal derivations, such that
  $d_2$ is a forward derivation while $d_1$ contains at least one backward transition.
  Then,
  there exists a forward derivation $d'_1$ of length strictly less
  than that of $d_1$ such that $d'_1\approx d_1$.
\end{lemma}

\begin{proof}
  We prove this lemma by induction on the length of $d_1$.  By the
  rearranging lemma (Lemma \ref{lemma:rearrange}) there exist a
  backward derivation $d$ and a forward derivation $d'$ such that
  $d_1\approx d;d'$. Furthermore, $d;d'$ is not longer than $d_1$. Let
  $s_1 \lh_{p_1,\ol{r_1},k_1} s_2 \rh_{p_2,r_2,k_2} s_3$ be the only
  two successive transitions in $d;d'$ with opposite direction.
  We will show below that there is in $d'$ a transition $t$ which is the inverse
  of $s_1 \lh_{p_1,\ol{r_1},k_1} s_2$. Moreover, we can swap $t$ with all the
      transitions between $t$ and
      $s_1 \lh_{p_1,\ol{r_1},k_1} s_2$, in order to obtain
      a derivation in which
      $s_1 \lh_{p_1,\ol{r_1},k_1} s_2$ and $t$ are
      adjacent.\!\footnote{More precisely, the transition is not $t$, but a
    transition that applies the same rule to the same process and
    producing the same history item, but possibly applied to a
    different system.} To do so we use the switching lemma
      (Lemma~\ref{lemma:switch}), since for all transitions $t'$ in
      between, we have that $\ol{t'}$ and $t$ are concurrent (this is proved below too). When $s_1 \lh_{p_1,\ol{r_1},k_1} s_2$ and $t$
      are adjacent we can remove both of them using $\approx$. The
      resulting derivation is strictly shorter, thus the claim follows
      by the inductive hypothesis.\\

  Let us now prove the results used above.
  Thanks to the loop lemma (Lemma \ref{lemma:loop}) we have the
  derivations above iff we have two forward derivations which are
  coinitial (with $s_2$ as initial state) and cofinal: $\ol{d};d_2$ and
  $d'$.  We first consider the case where $\ol{r_1} \neq
  \ol{\mathit{Sched}}$.  Since the first transition of $\ol{d};d_2$,
  $\ol{(s_1 \lh_{p_1,\ol{r_1},k_1} s_2)}$, adds item $k_1$ to the
  history of $p_1$ and such an item is never removed (since the
  derivation is forward), then the same item $k_1$ has to be added
  also by a transition in $d'$, otherwise the two derivations cannot
  be cofinal.
  The earliest transition in
  $d'$ adding item $k_1$ is exactly $t$.

  Let us now justify that for each transition $t'$ before $t$ in $d'$ we have that $\ol{t'}$ and $t$ are
  concurrent. First, $t'$ is a forward transition and it should be
  applied to a process which is different from $p_1$, otherwise
  the item $k_1$ would be added by transition $t$ in the wrong position in the history of $p_1$.
  We
  consider the following cases:
  \begin{itemize}
  \item If $t'$ applies rule
    $\mathit{Spawn}$ to create a process $\mathrm{p}$, then $t$ should
    not apply to process $\mathrm{p}$ since the process $\mathrm{p_1}$ to which $t$ applies
    already existed before $t'$. Therefore, $\ol{t'}$
    and $t$ are concurrent.
  \item If $t'$ applies rule $\mathit{Send}$ to send a message to some
    process $\mathrm{p}$, then $t$ cannot deliver the same message
    since we know that $t$ is not a $\mathit{Sched}$ since it adds item $k_1$ to the history.
    Thus $\ol{t'}$ and $t$ are
    concurrent.

  \item If $t'$ applies some other rule,
    then $t'$ and $t$ are clearly concurrent.
  \end{itemize}
  Now, we consider the case $\ol{r_1} = \ol{\mathit{Sched}}$ with
  $k_1 = \mathsf{sched}(\{v,\k\})$, so that
  $\ol{(s_1 \lh_{p_1,\ol{\mathit{Sched}},k_1} s_2)}$ adds a message
  $\{v,\k\}$ to the queue of $p_1$.
  We now distinguish two cases according to whether there is in $\ol{d};d_2$ an application of rule $\mathit{Receive}$ to $p_1$ or not:
  \begin{itemize}
  \item If the forward derivation $\ol{d};d_2$ contains no application
    of rule $\mathit{Receive}$ to $p_1$ then, in the final
    state, the queue of process $p_1$ contains the message.
    Hence, $d'$ needs to contain a $\mathit{Sched}$ for the same message.
    The earliest such
    $\mathit{Sched}$ transition in $d'$ is exactly $t$.

    Let us now justify that for each transition $t'$ before $t$ in $d'$ we have that $\ol{t'}$ and $t$ are
  concurrent. Consider the case where $t'$ applies rule
      $\mathit{Sched}$ to deliver a different message to the same
      process $p_1$. Since no $\mathit{Receive}$ would be performed on
      $p_1$ then the queue will stay different, and the two
      derivations could not be cofinal, hence this case can never
      happen. In all the other cases the two transitions are concurrent.

  \item If the forward derivation $\ol{d};d_2$ contains at least an
    application of rule $\mathit{Receive}$ to $p_1$, let us consider the
    first such application.  This creates a history item $k_2$. In
    order for the two derivations to be cofinal, the same history item
    needs to be created in $d'$. The queue stored in $k_2$ has a
    suffix $\{v,\k\}\cons q$, hence also in $d'$ the first
    $\mathit{Sched}$ delivering a message to $p_1$ should deliver
    message $\{v,\k\}$.  Since there are no other $\mathit{Sched}$ nor
    $\mathit{Receive}$ targeting $p_1$ then the $\mathit{Sched}$
    delivering message $\{v,\k\}$ to $p_1$ is concurrent to all
    previous transitions as desired. \qed
\end{itemize}
\end{proof}
Finally, we can state and prove the causal consistency of our
reversible semantics. Intuitively speaking, it states that two
different derivations starting from the same initial state can reach
the same final state if and only if they are causal consistent. On the
one hand, it means that derivations which are causal consistent lead
to the same final state, hence it is not possible to distinguish such
derivations looking at their final states (as a consequence, also
their possible evolutions coincide). In particular, swapping two
concurrent transitions or doing and undoing a given transition has no
impact on the final state. On the other hand, derivations differing in
any other way are distinguishable by looking at their final state, e.g., the
final state keeps track of any past nondeterministic choice. In other
terms, causal consistency states that the amount of history
information stored is precisely what is needed to distinguish
computations which are not causal consistent, and no more.

\begin{theorem}[Causal consistency] \label{thm:causal}
  Let $d_1$ and $d_2$ be coinitial derivations. Then, $d_1\approx d_2$
  iff $d_1$ and $d_2$ are cofinal.
\end{theorem}

\begin{proof}
  By definition of $\approx$, if $d_1\approx d_2$, then they are
  coinitial and cofinal, so this direction of the theorem is verified.

  Now, we have to prove that, if $d_1$ and $d_2$ are coinitial and
  cofinal, then $d_1\approx d_2$. By the rearranging lemma
  (Lemma~\ref{lemma:rearrange}), we know that the two derivations can
  be written as the composition of a backward derivation, followed by
  a forward derivation, so we assume that $d_1$ and $d_2$ have this
  form. The claim is proved by lexicographic induction on the sum of
  the lengths of $d_1$ and $d_2$, and on the distance between the end
  of $d_1$ and the earliest pair of transitions $t_1$ in $d_1$ and
  $t_2$ in $d_2$ which are not equal. If all such transitions are
  equal, we are done. Otherwise, we have to consider three cases
  depending on the directions of the two transitions:
  \begin{enumerate}
  \item Consider that $t_1$ is a forward transition and $t_2$ is a
    backward one. Let us assume that $d_1 = d;t_1;d'$ and
    $d_2 = d;t_2;d''$. Here, we know that $t_1;d'$ is a forward
    derivation, so we can apply the shortening lemma
    (Lemma~\ref{lemma:short}) to the derivations $t_1;d'$ and
    $t_2;d''$ (since $d_1$ and $d_2$ are coinitial and cofinal, so are
    $t_1;d'$ and $t_2;d''$), and we have that $t_2;d''$ has a strictly
    shorter forward derivation which is causally equivalent, and so
    the same is true for $d_2$. The claim then follows by induction.

  \item Consider now that both $t_1$ and $t_2$ are forward
    transitions.  By assumption, the two transitions must be
    different. Let us assume first that they are not
    concurrent. Therefore, they should be applied to the same process
    and either both rules are $\mathit{Sched}$, or one is
    $\mathit{Sched}$ and the other one is $\mathit{Receive}$. In the
    first case, we get a contradiction to the fact that $d_1$ and
    $d_2$ are cofinal since both derivations are forward and, thus, we
    would either have a different queue in the process or different
    items $\mathsf{rec}(\ldots)$ in the history. In the second case,
    where we have one rule $\mathit{Sched}$ and one
    $\mathit{Receive}$, the situation is similar. Therefore, we can
    assume that $t_1$ and $t_2$ are concurrent transitions.

    We have two cases, according to whether $t_1$ is an application of
    $\mathit{Sched}$ or not. If it is not, let $t'_1$ be the
    transition in $d_2$ creating the same history item as $t_1$. Then,
    we have to prove that $t'_1$ can be switched back with all
    previous forward transitions.  This holds since no previous
    forward transition can add any history item to the same process,
    since otherwise the two derivations could not be cofinal.  Hence
    the previous forward transitions are applied to different
    processes and thus we never have a conflict since the only
    possible sources of conflict would be rules $\mathit{Spawn}$ and
    $\mathit{Sched}$, but this could not happen since, in this case,
    $t_1$ could not happen neither.

    If $t_1$ is an application of $\mathit{Sched}$ then we can find
    the transition $t'_1$ in $d_2$ scheduling the same message
    (otherwise the two derivations could not be cofinal), and show
    that it can be switched with all the previous transitions. If the
    previous transition targets a different process then the only
    possible conflicts are with rules $\mathit{Send}$ or
    $\mathit{Spawn}$, but in this case $t_1$ could not have been
    performed. If the previous transition targets the same process
    then the only possible conflicts are with rules $\mathit{Sched}$
    or $\mathit{Receive}$, but in this case the derivations could not
    be cofinal.

    Then, in all the cases, we can repeatedly apply the switching
    lemma (Lemma~\ref{lemma:switch}) to have a derivation causally
    equivalent to $d_2$ where $t_2$ and $t'_1$ are consecutive. The
    same reasoning can be applied in $d_1$, so we end up with
    consecutive transitions $t_1$ and $t'_2$. Finally, we can apply
    the switching lemma once more to $t_1;t'_2$ so that the first pair
    of different transitions is now closer to the end of the
    derivation. Hence the claim follows by the inductive hypothesis.

  \item Finally, consider that both $t_1$ and $t_2$ are backward
    transitions.
    By definition, we have that $t_1$ and $t_2$ are
    concurrent. Let us consider first that the rules applied in the
    transitions are different from $\ol{\mathit{Sched}}$. Then, we
    have that $t_1$ and $t_2$ cannot remove the same history item. Let
    $k_1$ be the history item removed by $t_1$. Since $d_1$ and $d_2$
    are cofinal, either there is another transition in $d_1$ that puts
    $k_1$ back in the history or there is a transition $t'_1$ in $d_2$
    removing the same history item $k_1$. In the first case,
    $\ol{t_1}$ should be concurrent to all the backward transitions
    following it but the ones that remove history items from the
    history of the same process. All the transitions of this kind have
    to be undone by corresponding forward transitions (since they are
    not possible in $d_2$). Consider the last such transition: we can
    use the switching lemma (Lemma~\ref{lemma:switch}) to make it the
    last backward transition. Similarly, the forward transition
    undoing it should be concurrent to all the previous forward
    transitions (the reason is the same as in the previous
    case). Thus, we can use the switching lemma again to make it the
    first forward transition. Finally, we can apply the simplification
    rule $t;\ol{t} \approx \epsilon_{\init(t)}$ to remove the two
    transitions, thus shortening the derivation. In the second case
    (there is a transition $t'_1$ in $d_2$ removing the same history
    item $k_1$), one can argue as in case (2) above.  The claim then
    follows by the inductive hypothesis.

    The case when at least one of the rules applied in the transitions
    is $\ol{\mathit{Sched}}$ follows by a similar reasoning by
    considering the respective queues instead of the histories. \qed
  \end{enumerate}
\end{proof}
We now show that, as a corollary of previous results, a transition can
be undone if and only if each of its consequences, if any, has been
undone. Formally, a \emph{consequence} of a forward transition $t$ is a
forward transition $t'$ that can only happen after $t$ has been
performed (assuming $t$ has not been undone in between). 
Hence $t'$ cannot be switched with $t$. 
E.g., consuming a message from the queue of a process (using
rule $\mathit{Receive}$) is a consequence of delivering this
message (using rule $\mathit{Sched}$). Similarly, every action
performed by a process is a consequence of spawning this process. 

\begin{corollary} \label{cor:consequences}
  Let
  $d = (s_1 \rlh \cdots \rlh s_n \rh s_{n+1} \rlh \cdots \rlh s_m)$ be
  a derivation, with $t=(s_n \rh_{p,r,k} s_{n+1})$ a forward
  transition. Then, transition $\ol{t}$ can be applied to $s_m$, i.e.,
  $s_m \lh_{p,\ol{r},k} s_{m+1}$ iff each consequence of $t$ in $d$,
  if any, has been undone in $d$.
\end{corollary}

\begin{proof}
  If each consequence $t'$ of $t$ in $d$ has been undone in $d$ then
  we can find $d' \approx d$ with no consequence of $t$, by moving
  each consequence $t'$ and its undoing $\ol{t'}$ close to each other
  (they can be switched using the switching lemma
  (Lemma~\ref{lemma:switch}) with all the transitions in between, but
  for further consequences which can be removed beforehand) and then
  applying $t';\ol{t'} \approx \epsilon_{\init(t')}$. Then we can find
  $d'' \approx d'$ where $t$ is the last transition, since $t$ is
  concurrent to all subsequent transitions, hence we can apply the
  switching lemma (Lemma~\ref{lemma:switch}) again. The thesis then
  follows by applying the loop lemma (Lemma~\ref{lemma:loop}).

  Assume now that transition $\ol{t}$ can be applied to $s_m$. Thanks to the
  rearranging lemma (Lemma~\ref{lemma:rearrange}) there is a
  derivation $d_b;d_f \approx d;\ol{t}$ where $d_b$ is a backward
  derivation and $d_f$ is a forward derivation. In order to transform
  $d;\ol{t}$ into $d_b;d_f$ we need to move $\ol{t}$ backward using the switching
  lemma (Lemma~\ref{lemma:switch}) until we find $t$. However, neither
  $t$ nor $\ol{t}$ can be switched with the consequences of $t$, hence
  the only possibility is that all the consequences $t'$ of $t$ can be
  removed using $t';\ol{t'} \approx \epsilon_{\init(t')}$ as above. \qed
\end{proof}

\section{Rollback Semantics} \label{sec:rollbacksem}

In this section, we introduce a (nondeterministic) ``undo'' operation
which has some similarities to, e.g., the rollback operator of
\cite{LMSS11,GLM14}.  Here, processes in ``rollback'' mode are annotated
using $\lfloor ~ \rfloor_\Psi$, where $\Psi$ is the set of requested
rollbacks. A typical rollback refers to a checkpoint that the backward
computation of the process has to go through before resuming its
forward computation. To be precise, we distinguish the following types
of rollbacks:
\begin{itemize}
\item $\#_\mathsf{ch}^\mathtt{t}$, where ``$\mathsf{ch}$'' stands for
  ``checkpoint'': a rollback to undo the actions of a process until a
  checkpoint with identifier $\mathtt{t}$ is reached;
  
\item $\#_\mathsf{sp}$, where ``$\mathsf{sp}$'' stands for ``spawn'':
  a rollback to undo \emph{all} the actions of a process, finally
  deleting it from the system;
  
\item $\#_\mathsf{sch}^{\k}$, where ``$\mathsf{sch}$'' stands for
  ``sched'': a rollback to undo the actions of a process until the
  delivery of a message $\{v,\k\}$ is undone.
\end{itemize}
In the following, in order to simplify the reduction rules, we
consider that our semantics satisfies the following \emph{structural equivalence}:
\[
  \begin{array}{ll}
    (\mathit{SC}) & {\displaystyle
                    \Gamma;
                    \lfloor\tuple{p,\h,(\theta,e),q}\rfloor_{\emptyset}
                    \:\comp\: \Pi
                         ~\equiv~
                         \Gamma;\tuple{p,\h,(\theta,e),q}\:\comp\: \Pi
                                                              }
  \end{array}
\]
Note that only the first of the rollback types above targets a
checkpoint. This kind of checkpoint is introduced nondeterministically
by the rule below, where we denote by $\lhh$ the new reduction
relation that models backward moves of the rollback semantics:
\[
  \begin{array}{ll}
    (\mathit{\ol{Undo}}) & {\displaystyle
                           \Gamma;
                           \lfloor\tuple{p,\h,(\theta,e),q}\rfloor_{\Psi}
                           \:\comp\: \Pi
                                \lhh
                                \Gamma;\lfloor\tuple{p,\h,(\theta,e),q}\rfloor_{\Psi\cup\{\#_\mathsf{ch}^\mathtt{t}\}}\:\comp\: \Pi
    }
    \\[1ex] & \mbox{if } \mathsf{check}(\theta',e',\mathtt{t}) \mbox{ occurs
      in } \h, ~\mbox{for some $\theta'$ and $e'$}
  \end{array}
\]
Only after this rule is applied steps can be undone, since default
computation in the rollback semantics is forward.

\begin{figure}[t]
  \footnotesize
 \[
  \hspace{-3ex}
  \begin{array}{r@{~~}c}
      (\mathit{\ol{Seq}}) & {\displaystyle
        \Gamma;
                            \lfloor\tuple{p,\tau(\theta,e)\cons\h,(\theta',e'),q}\rfloor_{\Psi}
                            \:\comp\: \Pi
          \lhh  \Gamma;\lfloor\tuple{p,\h,(\theta,e),q}\rfloor_{\Psi}\:\comp\: \Pi
        }
      \\[2ex]

      (\mathit{\ol{Check}}) & {\displaystyle
        \Gamma;
                            \lfloor\tuple{p,\mathsf{check}(\theta,e,\mathtt{t})\cons\h,(\theta',e'),q}\rfloor_{\Psi}
                            \:\comp\: \Pi
          \lhh  \Gamma;\lfloor\tuple{p,\h,(\theta,e),q}\rfloor_{\Psi\setminus\{\#_\mathsf{ch}^\mathtt{t}\}}\:\comp\: \Pi
        } 
      \\[2ex]

    (\mathit{\ol{Send1}}) & {\displaystyle
      \begin{array}{l}
        \Gamma\cup\{(p',\{v,\k\})\};\lfloor\tuple{p,\mathsf{send}(\theta,e,p',\{v,\k\})\cons\h,(\theta',e'),q}\rfloor_{\Psi}\:\comp\:
        \Pi
        \lhh \Gamma;\lfloor\tuple{p,\h,(\theta,e),q}\rfloor_{\Psi} 
        \:\comp\: \Pi\\
    \end{array}
  }\\[2ex]

    (\mathit{\ol{Send2}}) & {\displaystyle
      \begin{array}{l}
        \Gamma;\lfloor\tuple{p,\mathsf{send}(\theta,e,p',\{v,\k\})\cons\h,(\theta',e'),q}\rfloor_{\Psi}\:\comp\:
                                                                                                            \lfloor\tuple{p',\h',(\theta'',e''),q'}\rfloor_{\Psi'}\:\comp\:\Pi\\
        \lhh 
        \Gamma;\lfloor\tuple{p,\mathsf{send}(\theta,e,p',\{v,\k\})\cons\h,(\theta',e'),q}\rfloor_{\Psi}\:\comp\:
                                                                                                            \lfloor\tuple{p',\h',(\theta'',e''),q'}\rfloor_{\Psi'\cup\{\#_\mathsf{sch}^\k\}}\:\comp\:
        \Pi\\
        \hspace{20ex}\mbox{if}~(p',\{v,\k\})~\mbox{does not
        occur
          in}~\Gamma~\mbox{and}~\#_\mathsf{sch}^{\k}\not\in\Psi'
    \end{array}
  }\\[5ex]

      (\mathit{\ol{Receive}}) & {\displaystyle
        \Gamma;\lfloor\tuple{p,\mathsf{rec}(\theta,e,\{v,\k\},q)\cons\h,(\theta',e'), q\backslash\!\!\backslash\{v,\k\}}\rfloor_{\Psi}\:\comp\: \Pi
          \lhh  \Gamma;\lfloor\tuple{p,\h,(\theta,e),q}\rfloor_{\Psi}\:\comp\: \Pi
        }
      \\[2ex]
      
    (\mathit{\ol{Spawn1}}) & {\displaystyle
      \begin{array}{l}
       \Gamma;\lfloor\tuple{p,\mathsf{spawn}(\theta,e,p'')\cons\h,(\theta',e'),q}\rfloor_{\Psi}
        \:\comp\: 
             \lfloor\tuple{\nil,p'',(\theta'',e''),\nil}\rfloor_{\Psi'}
        \:\comp\: \Pi 
        \\
        \hspace{20ex}\lhh
       \Gamma;\lfloor\tuple{p,\h,(\theta,e),q}\rfloor_{\Psi}
        \:\comp\: \Pi 
      \end{array}
      }\\[3ex]

    (\mathit{\ol{Spawn2}}) & {\displaystyle
      \begin{array}{l}
       \Gamma;\lfloor\tuple{p,\mathsf{spawn}(\theta,e,p'')\cons\h,(\theta,e),q}\rfloor_{\Psi}
        \:\comp\: 
             \lfloor\tuple{p'',\h'',(\theta'',e''),q''}\rfloor_{\Psi'}
        \:\comp\: \Pi 
        \\
        \hspace{0ex}\lhh
       \Gamma;\lfloor\tuple{p,\mathsf{spawn}(\theta,e,p'')\cons\h,(\theta,e),q}\rfloor_{\Psi}
        \:\comp\: 
             \lfloor\tuple{p'',\h'',(\theta'',e''),q''}\rfloor_{\Psi'\cup\{\#_\mathsf{sp}\}}
        \:\comp\: \Pi \\
        \hspace{20ex}\mbox{if}~\h''\neq\nil \lor q''\neq\nil~\mbox{and}~\#_\mathsf{sp}\not\in\Psi'
      \end{array}
      }\\[4ex]

    (\mathit{\ol{Self}}) & {\displaystyle
       \Gamma;\lfloor\tuple{p,\mathsf{self}(\theta,e)\cons\h,(\theta',e'),q}\rfloor_{\Psi} 
        \:\comp\: \Pi \lhh
      \Gamma;\lfloor\tuple{p,\h,(\theta,e),q}\rfloor_{\Psi} 
        \:\comp\: \Pi  
      }\\[2ex]

    (\mathit{\ol{Sched}}) & {\displaystyle
      \begin{array}{l}
      \Gamma;\lfloor\tuple{p,\h,(\theta,e),\{v,\k\}\cons q}\rfloor_{\Psi} \: \comp\: \Pi \lhh 
      \Gamma\cup(p,\{v,\k\});\lfloor\tuple{p,\h,(\theta,e),q}\rfloor_{\Psi\setminus\{\#_\mathsf{sch}^\k\}}\:\comp\:\Pi\\
        \hspace{20ex} 
       \mbox{if the topmost $\mathsf{rec}(\ldots)$ item in
        $\h$ (if any) has the}\\ 
      \hspace{20ex}\mbox{form}~\mathsf{rec}(\theta',e',\{v',\k'\},q')~\mbox{with}~q'\backslash\!\!\backslash\{v',\k'\}\neq
        \{v,\k\}\cons q
    \end{array}
    }
  \end{array}
  \]
  \caption{Rollback semantics: backward reduction rules} \label{fig:rollsem}
\end{figure}

The backward rules of the rollback semantics are shown in
Figure~\ref{fig:rollsem}. Here, we assume that $\Psi\neq\emptyset$
(but $\Psi'$ might be empty).

Note that, while rollbacks to checkpoints are generated
nondeterministically by rule $\mathit{\ol{Undo}}$, the two other kinds
of checkpoints are generated by the backward reduction rules in order
to ensure causal consistency (in the sense of
Corollary~\ref{cor:consequences}). This is clarified by the discussion
below, where we briefly explain the main differences w.r.t.\ the
uncontrolled backward semantics:
\begin{itemize}
\item As in the uncontrolled semantics of
  Figure~\ref{fig:backwardsem}, the sending of a message can be undone
  when the message is still in the global mailbox (rule
  $\ol{\mathit{Send1}}$). Otherwise, one may need to first apply rule
  $\ol{\mathit{Send2}}$ in order to ``propagate'' the rollback mode to
  the receiver of the message, so that rules $\ol{\mathit{Sched}}$ and
  $\ol{\mathit{Send1}}$ can be eventually applied.

\item As for undoing the spawning of a process $p''$, rule
  $\ol{\mathit{Spawn1}}$ steadily applies when both the history and
  the queue of the spawned process $p''$ are empty, thus deleting both
  the history item in $p$ and the process $p''$. Otherwise, we apply
  rule $\ol{\mathit{Spawn2}}$ to propagate the rollback mode to
  process $p''$ so that, eventually, rule $\ol{\mathit{Spawn1}}$ can
  be applied.  

\item Finally, observe that rule $\ol{\mathit{Sched}}$ requires the
  same side condition as in the uncontrolled semantics. This is needed
  in order to avoid the commutation of rules $\mathit{\ol{Receive}}$
  and $\mathit{\ol{Sched}}$.
\end{itemize}
The rollback semantics is modeled by the relation $\looparrowright$,
which is defined as the union of the forward reversible relation $\rh$
(Figure~\ref{fig:forwardsem}) and the backward relation $\lhh$ defined
in Figure~\ref{fig:rollsem}. 
Note that, in contrast to the (uncontrolled) reversible 
semantics of Section~\ref{sec:revsem},
the rollback semantics given by the relation $\looparrowright$
has less nondeterministic choices:
all computations run forward except when
a rollback action demands some backward steps to recover a previous
state of a process (which can be propagated to other processes in
order to undo the spawning of a process or the sending of a message).

Note, however, that besides the introduction of rollbacks, there is
still some nondeterminism in the backward rules of the rollback semantics: on the one hand,
the selection of the process when there are several ongoing rollbacks is
nondeterministic; also, in many cases, both rule $\mathit{\ol{Sched}}$
and another rule are applicable to the same process. The semantics
could be made deterministic by using a particular strategy to select
the processes (e.g., round robin) and applying rule
$\mathit{\ol{Sched}}$ 
whenever possible (i.e., give to $\mathit{\ol{Sched}}$ a higher
priority than to the remaining backward rules).

\begin{example} \label{ex2r}
  Consider again the program shown in Figure~\ref{fig:ex2-prog}. Let
  us assume that function $\mathrm{main}/0$ is now defined as follows:
  \[      
    \begin{array}{r@{~}ll}
      \mathrm{main}/0 = \mathsf{fun}~()\to & \mathsf{let}~S =
                                             \mathsf{spawn}(\mathrm{server}/0,\nil)\\
                                           & \mathsf{in}~\mathsf{let}~\_ = \mathsf{spawn}(\mathrm{client}/1,[S]) \\
                                           &
                                             \mathsf{in}~\mathsf{let}~X = \mathsf{check}(\mathtt{t})\\
                                           & \mathsf{in}~\mathsf{apply}~\mathrm{client}/1~(S)
    \end{array}
  \]
  so that a checkpoint has been introduced after spawning the two
  processes: the server ($\mathrm{s}$) and one of the clients
  ($\mathrm{c2}$). Then, by repeating the same forward derivation
  shown in Figure~\ref{fig:ex2-der} (with the additional step to
  evaluate the checkpoint), we get the following final system:
  \[ 
      \begin{array}{l@{~}l@{~}l}
        & \{\;\}; &
                    \l\mathrm{c1},[\mathsf{rec}(\_,\_,m_4,[m_4]),\mathsf{send}(\_,\_,\mathrm{s},m_3),\mathsf{check}(\_,\_,\mathtt{t}),\mathsf{spawn}(\_,\_,\mathrm{c2}), \\
        &&~~~~~~\mathsf{spawn}(\_,\_,\mathrm{s})],(\_,\mathrm{ok}),\nil\r\\
        &&\comp\l\mathrm{s},
           [\mathsf{send}(\_,\_,\mathrm{c1},m_4),\mathsf{rec}(\_,\_,m_3,[m_3]),\mathsf{send}(\_,\_,\mathrm{c2},m_2),\\
        &&~~~~~~\mathsf{rec}(\_,\_,m_1,[m_1])], (\_,C[\mathsf{receive}~\{P,M\}\to\ldots]),\nil\r\\
        &&\comp\l\mathrm{c2},[\underline{\mathsf{rec}(\_,\_,m_2,[m_2])},\mathsf{send}(\_,\_,\mathrm{s},m_1)],(\_,\mathrm{ok}),\nil\r
      \end{array}
    \]
  Figure~\ref{fig:rollbackder} shows the steps performed by the
  rollback semantics in order to undo the steps of process
  $\mathrm{c1}$ until the checkpoint is reached. In
  Figure~\ref{fig:rollbackder} we follow the same conventions as in
  Examples~\ref{ex2} and \ref{ex2b}. Observe that we could also use
  the relation ``$\looparrowright$'' here in order to also perform
  some forward steps on process $\mathrm{c2}$, as it would happen in
  practice.

  \begin{figure}[p]
    \scriptsize
    \[
      \begin{array}{l@{~}l@{~}l}
        & \{\;\}; &
        \lfloor\l\mathrm{c1},[\underline{\mathsf{rec}(\_,\_,m_4,[m_4])},\mathsf{send}(\_,\_,\mathrm{s},m_3),\mathsf{check}(\_,\_,\mathtt{t}),\mathsf{spawn}(\_,\_,\mathrm{c2}),\mathsf{spawn}(\_,\_,\mathrm{s})],\\
        &&~~(\_,\mathrm{ok}),\nil\r\rfloor_{\{\#_\mathsf{ch}^\mathtt{t}\}}\\
        &&\comp\l\mathrm{s},
           [\mathsf{send}(\_,\_,\mathrm{c1},m_4),\mathsf{rec}(\_,\_,m_3,[m_3]),\mathsf{send}(\_,\_,\mathrm{c2},m_2), \mathsf{rec}(\_,\_,m_1,[m_1])], \\
        &&~~(\_,C[\mathsf{receive}~\{P,M\}\to\ldots]),\nil\r\\
        &&\comp\l\mathrm{c2},[\mathsf{rec}(\_,\_,m_2,[m_2]),\mathsf{send}(\_,\_,\mathrm{s},m_1)],(\_,\mathrm{ok}),\nil\r \\[1ex]

        \lhh_\mathit{\ol{Receive}} & \{\;\}; &
        \lfloor\l\mathrm{c1},[\underline{\mathsf{send}(\_,\_,\mathrm{s},m_3)},\mathsf{check}(\_,\_,\mathtt{t}),\mathsf{spawn}(\_,\_,\mathrm{c2}),\mathsf{spawn}(\_,\_,\mathrm{s})],\\
        &&~~(\_, C[\mathsf{receive}~\mathrm{ack}\to \mathrm{ok}]),[m_4]\r\rfloor_{\{\#_\mathsf{ch}^\mathtt{t}\}}\\
        &&\comp\l\mathrm{s},
           [\mathsf{send}(\_,\_,\mathrm{c1},m_4),\mathsf{rec}(\_,\_,m_3,[m_3]),\mathsf{send}(\_,\_,\mathrm{c2},m_2), \mathsf{rec}(\_,\_,m_1,[m_1])], \\
        &&~~(\_,C[\mathsf{receive}~\{P,M\}\to\ldots]),\nil\r\\
        &&\comp\l\mathrm{c2},[\mathsf{rec}(\_,\_,m_2,[m_2]),\mathsf{send}(\_,\_,\mathrm{s},m_1)],(\_,\mathrm{ok}),\nil\r \\[1ex]

        \lhh_\mathit{\ol{Send2}} & \{\;\}; &
        \lfloor\l\mathrm{c1},[\mathsf{send}(\_,\_,\mathrm{s},m_3),\mathsf{check}(\_,\_,\mathtt{t}),\mathsf{spawn}(\_,\_,\mathrm{c2}),\mathsf{spawn}(\_,\_,\mathrm{s})],\\
        &&~~(\_, C[\mathsf{receive}~\mathrm{ack}\to \mathrm{ok}]),[m_4]\r\rfloor_{\{\#_\mathsf{ch}^\mathtt{t}\}}\\
        &&\comp\lfloor\l\mathrm{s},
           [\underline{\mathsf{send}(\_,\_,\mathrm{c1},m_4)},\mathsf{rec}(\_,\_,m_3,[m_3]),\mathsf{send}(\_,\_,\mathrm{c2},m_2), \mathsf{rec}(\_,\_,m_1,[m_1])], \\
        &&~~(\_,C[\mathsf{receive}~\{P,M\}\to\ldots]),\nil\r\rfloor_{\{\#_ \mathsf{sch}^3\}}\\
        &&\comp\l\mathrm{c2},[\mathsf{rec}(\_,\_,m_2,[m_2]),\mathsf{send}(\_,\_,\mathrm{s},m_1)],(\_,\mathrm{ok}),\nil\r \\[1ex]

        \lhh_\mathit{\ol{Send2}} & \{\;\}; &
        \lfloor\l\mathrm{c1},[\mathsf{send}(\_,\_,\mathrm{s},m_3),\mathsf{check}(\_,\_,\mathtt{t}),\mathsf{spawn}(\_,\_,\mathrm{c2}),\mathsf{spawn}(\_,\_,\mathrm{s})],\\
        &&~~(\_, C[\mathsf{receive}~\mathrm{ack}\to \mathrm{ok}]),[\underline{m_4}]\r\rfloor_{\{\#_\mathsf{ch}^\mathtt{t}, \#_ \mathsf{sch}^4\}}\\
        &&\comp\lfloor\l\mathrm{s},
           [\mathsf{send}(\_,\_,\mathrm{c1},m_4),\mathsf{rec}(\_,\_,m_3,[m_3]),\mathsf{send}(\_,\_,\mathrm{c2},m_2), \mathsf{rec}(\_,\_,m_1,[m_1])], \\
        &&~~(\_,C[\mathsf{receive}~\{P,M\}\to\ldots]),\nil\r\rfloor_{\{\#_ \mathsf{sch}^3\}}\\
        &&\comp\l\mathrm{c2},[\mathsf{rec}(\_,\_,m_2,[m_2]),\mathsf{send}(\_,\_,\mathrm{s},m_1)],(\_,\mathrm{ok}),\nil\r \\[1ex]

        \lhh_\mathit{\ol{Sched}} & \{(\mathrm{c1},m_4)\}; &
        \lfloor\l\mathrm{c1},[\mathsf{send}(\_,\_,\mathrm{s},m_3),\mathsf{check}(\_,\_,\mathtt{t}),\mathsf{spawn}(\_,\_,\mathrm{c2}),\mathsf{spawn}(\_,\_,\mathrm{s})],\\
        &&~~(\_, C[\mathsf{receive}~\mathrm{ack}\to \mathrm{ok}]),\nil\r\rfloor_{\{\#_\mathsf{ch}^\mathtt{t}\}}\\
        &&\comp\lfloor\l\mathrm{s},
           [\underline{\mathsf{send}(\_,\_,\mathrm{c1},m_4)},\mathsf{rec}(\_,\_,m_3,[m_3]),\mathsf{send}(\_,\_,\mathrm{c2},m_2), \mathsf{rec}(\_,\_,m_1,[m_1])], \\
        &&~~(\_,C[\mathsf{receive}~\{P,M\}\to\ldots]),\nil\r\rfloor_{\{\#_ \mathsf{sch}^3\}}\\
        &&\comp\l\mathrm{c2},[\mathsf{rec}(\_,\_,m_2,[m_2]),\mathsf{send}(\_,\_,\mathrm{s},m_1)],(\_,\mathrm{ok}),\nil\r \\[1ex]

        \lhh_\mathit{\ol{Send1}} & \{ \;\}; &
        \lfloor\l\mathrm{c1},[\mathsf{send}(\_,\_,\mathrm{s},m_3),\mathsf{check}(\_,\_,\mathtt{t}),\mathsf{spawn}(\_,\_,\mathrm{c2}),\mathsf{spawn}(\_,\_,\mathrm{s})],\\
        &&~~(\_, C[\mathsf{receive}~\mathrm{ack}\to \mathrm{ok}]),\nil\r\rfloor_{\{\#_\mathsf{ch}^\mathtt{t}\}}\\
        &&\comp\lfloor\l\mathrm{s},
           [\underline{\mathsf{rec}(\_,\_,m_3,[m_3])},\mathsf{send}(\_,\_,\mathrm{c2},m_2), \mathsf{rec}(\_,\_,m_1,[m_1])],(\_,C[\mathrm{c1} \:!\: \mathrm{ack}]),\nil\r\rfloor_{\{\#_ \mathsf{sch}^3\}}\\
        &&\comp\l\mathrm{c2},[\mathsf{rec}(\_,\_,m_2,[m_2]),\mathsf{send}(\_,\_,\mathrm{s},m_1)],(\_,\mathrm{ok}),\nil\r \\[1ex]

        \lhh_\mathit{\ol{Receive}} & \{ \;\}; &
        \lfloor\l\mathrm{c1},[\mathsf{send}(\_,\_,\mathrm{s},m_3),\mathsf{check}(\_,\_,\mathtt{t}),\mathsf{spawn}(\_,\_,\mathrm{c2}),\mathsf{spawn}(\_,\_,\mathrm{s})],\\
        &&~~(\_, C[\mathsf{receive}~\mathrm{ack}\to \mathrm{ok}]),\nil\r\rfloor_{\{\#_\mathsf{ch}^\mathtt{t}\}}\\
        &&\comp\lfloor\l\mathrm{s},
           [\mathsf{send}(\_,\_,\mathrm{c2},m_2), \mathsf{rec}(\_,\_,m_1,[m_1])],(\_,C[\mathsf{receive}~\{P,M\}\to\ldots]),[\underline{m_3}]\r\rfloor_{\{\#_ \mathsf{sch}^3\}}\\
        &&\comp\l\mathrm{c2},[\mathsf{rec}(\_,\_,m_2,[m_2]),\mathsf{send}(\_,\_,\mathrm{s},m_1)],(\_,\mathrm{ok}),\nil\r \\[1ex]

        \lhh_\mathit{\ol{Sched}} & \{ (\mathrm{s},m_3)\}; &
        \lfloor\l\mathrm{c1},[\underline{\mathsf{send}(\_,\_,\mathrm{s},m_3)},\mathsf{check}(\_,\_,\mathtt{t}),\mathsf{spawn}(\_,\_,\mathrm{c2}),\mathsf{spawn}(\_,\_,\mathrm{s})],\\
        &&~~(\_, C[\mathsf{receive}~\mathrm{ack}\to \mathrm{ok}]),\nil\r\rfloor_{\{\#_\mathsf{ch}^\mathtt{t}\}}\\
        &&\comp\l\mathrm{s},
           [\mathsf{send}(\_,\_,\mathrm{c2},m_2), \mathsf{rec}(\_,\_,m_1,[m_1])],(\_,C[\mathsf{receive}~\{P,M\}\to\ldots]),\nil\r\\
        &&\comp\l\mathrm{c2},[\mathsf{rec}(\_,\_,m_2,[m_2]),\mathsf{send}(\_,\_,\mathrm{s},m_1)],(\_,\mathrm{ok}),\nil\r \\[1ex]

        \lhh_\mathit{\ol{Send1}} & \{ \;\}; &
        \lfloor\l\mathrm{c1},[\underline{\mathsf{check}(\_,\_,\mathtt{t})},\mathsf{spawn}(\_,\_,\mathrm{c2}),\mathsf{spawn}(\_,\_,\mathrm{s})],(\_, C[\mathrm{s}\:!\:\{\mathrm{c1},\mathrm{req}\}]),\nil\r\rfloor_{\{\#_\mathsf{ch}^\mathtt{t}\}}\\
        &&\comp\l\mathrm{s},
           [\mathsf{send}(\_,\_,\mathrm{c2},m_2), \mathsf{rec}(\_,\_,m_1,[m_1])],(\_,C[\mathsf{receive}~\{P,M\}\to\ldots]),\nil\r\\
        &&\comp\l\mathrm{c2},[\mathsf{rec}(\_,\_,m_2,[m_2]),\mathsf{send}(\_,\_,\mathrm{s},m_1)],(\_,\mathrm{ok}),\nil\r \\[1ex]

        \lhh_\mathit{\ol{Check}} & \{ \;\}; &
        \l\mathrm{c1},[\mathsf{spawn}(\_,\_,\mathrm{c2}),\mathsf{spawn}(\_,\_,\mathrm{s})],(\_, C[\mathsf{check}(\texttt{t})]),\nil\r\\
        &&\comp\l\mathrm{s},
           [\mathsf{send}(\_,\_,\mathrm{c2},m_2), \mathsf{rec}(\_,\_,m_1,[m_1])],(\_,C[\mathsf{receive}~\{P,M\}\to\ldots]),\nil\r\\
        &&\comp\l\mathrm{c2},[\mathsf{rec}(\_,\_,m_2,[m_2]),\mathsf{send}(\_,\_,\mathrm{s},m_1)],(\_,\mathrm{ok}),\nil\r 
      \end{array}
    \]
    \caption{A derivation under the backward reduction rules of
      Figure~\ref{fig:rollsem}, with
      $m_1 = \{\{\mathrm{c2},\mathrm{req}\},1\}$,
      $m_2 = \{\mathrm{ack},2\}$,
      $m_3 = \{\{\mathrm{c1},\mathrm{req}\},3\}$, and $m_4 = \{\mathrm{ack},4\}$.} \label{fig:rollbackder}
  \end{figure}
\end{example}
We state below the soundness of the rollback semantics.
In order to do it, we let $\rolldel(s)$ denote the system that
results from $s$ by removing ongoing rollbacks; formally,
$\rolldel(\Gamma;\Pi) = \Gamma;\rolldel'(\Pi)$, with
\[
  \begin{array}{lll}
    \rolldel'(\tuple{p,\h,(\theta,e),q}) & = &
                                                  \tuple{p,\h,(\theta,e),q}\\
    \rolldel'(\lfloor \tuple{p,\h,(\theta,e),q}\rfloor_{\Psi}) & = &
                                                  \tuple{p,\h,(\theta,e),q}\\
    \rolldel'(\tuple{p,\h,(\theta,e),q}\:\comp\:\Pi) & = &
                                                  \tuple{p,\h,(\theta,e),q}\:\comp\:\rolldel'(\Pi)\\
    \rolldel'(\lfloor\tuple{p,\h,(\theta,e),q}\rfloor_{\Psi}\:\comp\:\Pi) & = &
                                                  \tuple{p,\h,(\theta,e),q}\:\comp\:\rolldel'(\Pi)
    \\                                              
  \end{array}
\]
where we assume that $\Pi$ is not empty.
We also extend the definition of initial and reachable systems to the rollback semantics.

\begin{definition}[Reachable systems under the rollback
  semantics]\label{def:reachableroll} \mbox{}\\
  A system is \emph{initial} under the rollback semantics if it is
  composed by a single process with an empty set $\Psi$ of active
  rollbacks; furthermore, the history, the queue and the global
  mailbox are empty too. A system $s$ is
  \emph{reachable} under the rollback semantics if there exist an initial
  system $s_0$ and a derivation $s_0 \looparrowright^\ast s$
  using the rules corresponding to a given program.
\end{definition}

\begin{theorem}[Soundness] \label{thm:soundness-rollback}
  Let $s$ be a system reachable under the rollback semantics. If
  $s \looparrowright^\ast s'$, then $\rolldel(s) \rlh^\ast \rolldel(s')$.
\end{theorem}

\begin{proof}
  For forward transitions the proof is trivial since the forward rules
  are the same in both semantics, and they apply only to processes which
  are not under rollback. For backward transitions the proof is by case
  analysis on the applied rule, noting that the effect of structural
  equivalence is removed by $\rolldel$:

  \begin{itemize}
    \item Rule $\mathit{\ol{Undo}}$: the effect is removed by
      $\rolldel$, hence an application of this rule corresponds to a
      zero-step derivation under the uncontrolled semantics;
    \item Rules $\mathit{\ol{Seq}}$, $\mathit{\ol{Check}}$, $\mathit{\ol{Send1}}$, $\mathit{\ol{Receive}}$, $\mathit{\ol{Spawn1}}$, $\mathit{\ol{Self}}$ and $\mathit{\ol{Sched}}$: they are matched, respectively, by rules $\mathit{\ol{Seq}}$, $\mathit{\ol{Check}}$, $\mathit{\ol{Send}}$, $\mathit{\ol{Receive}}$, $\mathit{\ol{Spawn}}$, $\mathit{\ol{Self}}$ and $\mathit{\ol{Sched}}$ of the uncontrolled semantics;
      \item Rules $\mathit{\ol{Send2}}$ and $\mathit{\ol{Spawn2}}$:
        the effect is removed by $\rolldel$, hence an application of
        any of these rules corresponds to a zero-step derivation under the
        uncontrolled semantics.\qed
  \end{itemize}
\end{proof}
We can now show the completeness of the rollback semantics provided
that the involved process is in rollback mode:

\begin{lemma}[Completeness in rollback mode] \label{lemma:completeness-rollback}
  Let $s$ be a reachable system. If $s \lh s'$ then take any system
  $s_r$ such that $\rolldel(s_r)=s$ and where the process that
  performed the transition $s \lh s'$ is in rollback mode for a
  non-empty set of rollbacks. There exists $s'_r$ such that $s_r \lhh s'_r$
  and $\rolldel(s'_r)=s'$.
\end{lemma}

\begin{proof}
  The proof is by case analysis on the applied rule. Each step is
  matched by the homonymous rule, but for $\mathit{\ol{Send}}$ and
  $\mathit{\ol{Spawn}}$ which are matched by rules
  $\mathit{\ol{Send1}}$ and $\mathit{\ol{Spawn1}}$. \qed
\end{proof}
The following result illustrates the usefulness of the
rollback semantics:

\begin{lemma} \label{lemma:rev} Let us consider a forward derivation $d$
  of the form: 
  \[
  \begin{array}{l}
    \Gamma;\tuple{p,\h,(\theta,\mathsf{let}~X=\mathsf{check}(\mathtt{t})~\mathsf{in}~e),q}\comp\Pi \\
    \rh
    \Gamma; \tuple{p,\mathsf{check}(\theta,
    \mathsf{let}~X=\mathsf{check}(\mathtt{t})~\mathsf{in}~e,\mathtt{t})\cons\h,(\theta,
    \mathsf{let}~X=\mathtt{t}~\mathsf{in}~e),q}\comp\Pi \\
    \rh^\ast
    \Gamma';\tuple{p,\h',(\theta',e'),q'}\comp\Pi'
  \end{array}
  \]
  Then, there is a backward derivation $d'$ under the rollback
  semantics restoring process $p$:
  \[
    \begin{array}{l}
      \Gamma';\lfloor\tuple{p,\h',(\theta',e'),q'}\rfloor_{\{\#_\mathsf{ch}^\mathtt{t}\}}\comp\Pi'\\
      \lhh^\ast
      \Gamma''; \tuple{p,\h,(\theta,\mathsf{let}~X=\mathsf{check}(\mathtt{t})~\mathsf{in}~e),q}\comp\Pi''
    \end{array}  
  \]
\end{lemma}

\begin{proof}
  Trivially (by Theorem~\ref{thm:soundness-rollback}) the forward
  derivation $d$ can also be performed under the uncontrolled reversible
  semantics.   Now, by applying the loop lemma (Lemma~\ref{lemma:loop}) to each
  step of $d$, we have a backward derivation $\ol{d}$ of the form:
  \[
  \begin{array}{l}
    \Gamma';\tuple{p,\h',(\theta',e'),q'}\comp\Pi'\\
    \lh^\ast
    \Gamma;\tuple{p,\h,(\theta,\mathsf{let}~X=\mathsf{check}(\mathtt{t})~\mathsf{in}~e),q}\comp\Pi
  \end{array}
  \]
  Consider the relation $\leq$ on transitions of $\ol{d}$ defined as
  the reflexive and transitive closure of the following clauses:
  \begin{itemize}
  \item $t_1 \leq t_2$ if both $t_1$ and $t_2$ undo actions in the
    same process $p'$, and the transition undone by $t_2$ is a direct
    consequence of the one undone by $t_1$;
  \item $t_1 \leq t_2$ if $t_1$ undoes a spawn of process $p_2$ and $t_2$
    undoes the first transition of $p_2$;
  \item $t_1 \leq t_2$ if $t_1$ undoes the send of a message $\k$
    and $t_2$ undoes the scheduling of the same message.
  \end{itemize}
  Let us show that $\leq$ is a partial order. We only need to show
  that there are no cycles, but this follows from the fact that the
  total order given by $\ol{d}$ is compatible with $\leq$.

  We also notice that any two transitions which are not related by
  $\leq$ can be swapped using the switching lemma
  (Lemma~\ref{lemma:switch}).
  
  Then, there exists a derivation $\ol{d_r};\ol{d_u}$ such that $\ol{d_r}$
  contains all transitions $t$ such that $t_l \leq t$ where $t_l$ is
  the last transition in $\ol{d}$, 
  and only them. Since $\ol{d_u}$ contains no transition on $p$ we
  have that $\ol{d_r}$ is of the form:
  \[
  \begin{array}{l}
    \Gamma';\tuple{p,\h',(\theta',e'),q'}\comp\Pi'\\
    \lh^\ast
    \Gamma'';\tuple{p,\h,(\theta,\mathsf{let}~X=\mathsf{check}(\mathtt{t})~\mathsf{in}~e),q}\comp\Pi''
  \end{array}
  \]
  Using again the switching lemma (Lemma~\ref{lemma:switch}) one can
  transform $\ol{d_r}$ into a derivation $\ol{d'_r}$ obtained using
  the following execution strategy, where initially the active process
  is $p$, the termination condition is ``the checkpoint action $\mathtt{t}$ has been undone'', and the stack is empty:
  \begin{itemize}
  \item transitions of the active process are undone if possible, until
    the termination condition holds; if there is an occurrence of the
    active process in the stack and the termination condition for this
    process is matched because of the current transition undo, remove
    such occurrence from the stack (this remove does not follow the
    usual FIFO strategy for stacks);
  \item if the termination condition holds, then pop a new active
    process from the stack, if there are no processes on the stack
    then terminate;
  \item if no transition is possible for the active process then one of
    the two following subconditions should hold:
    \begin{enumerate}
      \item the active process needs to undo a spawn of a process
        which is not in the initial state: push the active process on
        the stack, and set the spawned process as new active process
        with termination condition ``all actions have been undone'';
      \item the active process needs to undo a send of a message $\k$
        which is not in the global mailbox: push the active process on
        the stack, and set the process to which message $\k$ has been
        scheduled as new active process with termination condition
        ``the scheduling of the message $\k$ has been undone'';
    \end{enumerate}
  \end{itemize}
  The switching lemma can be applied since this execution strategy is
  compatible with $\leq$. Now we show that the same execution strategy
  can be performed using the rollback semantics. We only need to show
  that the active process is in rollback mode, then the thesis will
  follow from the completeness in rollback mode
  (Lemma~\ref{lemma:completeness-rollback}). This can be shown by
  inspection of the execution strategy, considering the following
  invariant: the active process and all the processes on the stack are
  in rollback mode, and they have one checkpoint for each occurrence
  in the stack, plus one for the occurrence as active process. The
  invariant holds at the beginning since $p$ has one checkpoint
  corresponding to its termination condition. When the termination
  condition holds, a checkpoint is removed by rule
  $\mathit{\ol{Check}}$, $\mathit{\ol{Spawn1}}$, or
  $\mathit{\ol{Sched}}$.  When a new active process is selected, a new
  checkpoint is added by rule $\mathit{\ol{Spawn2}}$ or
  $\mathit{\ol{Send2}}$. \qed
\end{proof}
One can notice that in the lemma above only the process containing the
checkpoint is restored. We can restore the whole system to the
original configuration only if we restrict the forward derivation to
be a causal derivation, following the terminology in~\cite{DK05}.

\begin{definition}
  A forward derivation $d$ is causal iff all the transitions are 
  consequences of the first one.
\end{definition}
Hence, we have the following corollary:

\begin{corollary} Let us consider a causal derivation $d$
  of the form: 
  \[
  \begin{array}{l}
    \Gamma;\tuple{p,\h,(\theta,\mathsf{let}~X=\mathsf{check}(\mathtt{t})~\mathsf{in}~e),q}\comp\Pi \\
    \rh
    \Gamma; \tuple{p,\mathsf{check}(\theta,
    \mathsf{let}~X=\mathsf{check}(\mathtt{t})~\mathsf{in}~e,\mathtt{t})\cons\h,(\theta,
    \mathsf{let}~X=\mathtt{t}~\mathsf{in}~e),q}\comp\Pi \\
    \rh^\ast
    \Gamma';\tuple{p,\h',(\theta',e'),q'}\comp\Pi'
  \end{array}
  \]
  Then, there is a backward derivation $d'$ under the rollback
  semantics restoring the system to the original configuration:
  \[
    \begin{array}{l}
      \Gamma';\lfloor\tuple{p,\h',(\theta',e'),q'}\rfloor_{\{\#_\mathsf{ch}^\mathtt{t}\}}\comp\Pi'\\
      \lhh^\ast
      \Gamma; \tuple{p,\h,(\theta,\mathsf{let}~X=\mathsf{check}(\mathtt{t})~\mathsf{in}~e),q}\comp\Pi
    \end{array}  
  \]
\end{corollary}
\begin{proof}
  The proof follows the same strategy as the one of
  Lemma~\ref{lemma:rev}, noticing that $\ol{d_u}$ is empty hence
  $\Gamma = \Gamma''$ and $\Pi = \Pi''$. \qed
\end{proof}
While a derivation restoring the whole system exists, not all
derivations do so. More in general, given a set of rollbacks, it is
not the case that there is a unique system that is obtained by
executing backward transitions as far as possible (without executing
any $\mathit{\ol{Undo}}$). Indeed, the only nondeterminism is due to
the fact that $\mathit{\ol{Sched}}$ can commute with other
transitions, e.g., with $\mathit{\ol{Check}}$, which ends the
rollback. If we establish a policy for $\mathit{Sched}$ actions, and
we use the dual policy for undoing them, then the result is unique. A
sample policy could be that $\mathit{Sched}$ steps are performed as late as
possible, and dually undone as soon as possible. In such a setting we
have the following result:

\begin{lemma} \label{lemma:rollback-confluence} Let $s$ be a reachable
  system. If $s \lhh s_1$ and $s \lhh s_2$, both transitions use the
  same policy for $\ol{\mathit{Sched}}$, and the rules are different from
  $\mathit{\ol{Undo}}$, then there exists a system $s'$ such that $s_1
  \lhh^* s'$ and $s_2 \lhh^* s'$.
\end{lemma}

\begin{proof}
  Let us consider the case where both transitions are applied to the
  same process $p$. In this case, only one backward rule is applicable
  and the claim follows trivially. Note that the only case where more
  than one backward rule would be applicable is when one of the rules
  is $\ol{\mathit{Sched}}$ and the other one is a different rule, but
  this case is excluded by the fact that we consider a fixed policy
  for $\ol{\mathit{Sched}}$ as mentioned above.

  Consider now the case where each transition is applied to a
  different process, say $p_1$ and $p_2$, so that we have $s \lhh s_1$
  and $s \lhh s_2$. By the soundness of the backward reduction rules
  of the rollback semantics (Theorem~\ref{thm:soundness-rollback}), we
  have $\rolldel(s) \lh^* \rolldel(s_1)$ and $\rolldel(s) \lh^*
  \rolldel(s_2)$. Note that each of the derivations above has either
  length $1$ or $0$. We just consider the case where they have both
  length $1$, since the others are simpler.  By the square lemma
  (Lemma~\ref{lemma:square}), there exists a system $s''$ such that
  $\rolldel(s_1) \lh s''$ and $\rolldel(s_2) \lh s''$.  Now, we show
  that processes $p_2$ and $p_1$ are still in rollback mode in $s_1$
  and $s_2$, respectively. Here, the only case where the application
  of a backward rule to a process removes a rollback from a different
  process is $\ol{\mathit{Spawn}}$. Consider, e.g., that the rule
  applied to process $p_1$ is 
 $\ol{\mathit{Spawn}}$ and that the removed process is $p_2$. In this
  case, however, no backward rule could be applied to process $p_2$, so this
  case is not possible. Therefore, by applying the completeness of the
  rollback semantics, we have $s_1 \lhh s'_1$ and $s_2\lhh s'_2$ with
  $\rolldel(s'_1) = \rolldel(s'_2) = s''$. The thesis follows by
  noticing that the rollbacks in $s'_1$ and $s'_2$ coincide (in both
  the cases they are the rollbacks in $s$ minus the ones removed by
  the performed transitions, which are the same in both the cases)
  hence $s'_1=s'_2=s'$.  \qed
\end{proof}
The following result is an easy corollary of the previous lemma:

\begin{corollary}
  Let $s$ be a reachable system. If $s \lhh^\ast s_1 \not\lhh$ and $s
  \lhh^\ast s_2 \not\lhh$, both derivations use the same policy for
  $\ol{\mathit{Sched}}$, and never use rule $\mathit{\ol{Undo}}$, then $s_1 =
  s_2$.
\end{corollary}
\begin{proof}
  Analogously to the proof of Corollary~\ref{cor:backconf}, using
  standard results for confluence of abstract relations \cite{BN98},
  we have that Lemma~\ref{lemma:rollback-confluence} implies that
  there exists a system $s'$ such that $s_1 \lhh^\ast s'$ and $s_2
  \lhh^\ast s'$. Moreover, since both $s_1$ and $s_2$ are irreducible,
  we have $s_1 = s_2$.  \qed
\end{proof}

\section{Proof-of-concept Implementation of the Reversible Semantics} \label{sec:imple}

We have developed a proof-of-concept implementation of the uncontrolled
reversible semantics for Erlang that we presented in
Section~\ref{sec:semantics}.  This implementation is conveniently
bundled together with a graphical user interface (we refer to this as
``the application'') in order to facilitate the interaction of users
with the reversible semantics.  However, the application has been
developed in a modular way, so that it is possible to include the
implementation of the reversible semantics in other projects (e.g., it has been included in the
reversible debugger CauDEr~\cite{LNPV:FLOPS,cauder}).

\begin{figure}[t] \begin{center} \includegraphics[scale=0.4]{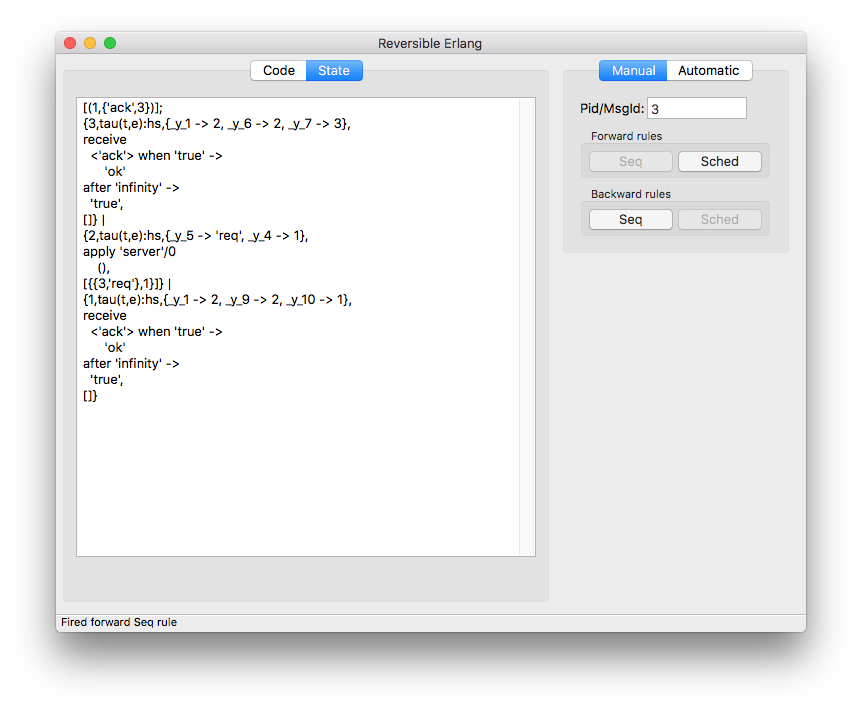}
  \end{center} \label{fig:screen} \vspace{-5ex} \caption{Screenshot of
    the application}\label{fig:screenshot} \end{figure}

Let us recall that our semantics is defined for a language that is
equivalent to Core Erlang \cite{CGJLNPV04}, a much simpler language
than Erlang. Not surprisingly, the implementation of our reversible
semantics is defined for Core Erlang as well. Prior to its
compilation, Erlang programs are translated to Core Erlang by
the Erlang/OTP system,
so that the resulting code is simplified. For instance, pattern
matching can occur almost anywhere in an Erlang program, whereas in
Core Erlang, pattern matching can only occur in case
statements. Nevertheless, directly writing Core Erlang programs 
would not be comfortable for the user,
since Core Erlang is only used as an intermediate language. Hence, our
implementation considers the Core Erlang code translated from the Erlang
program provided by the user.

The application works as follows: when it is started, the first step
is to select an Erlang source file. The selected source file is then
translated into Core Erlang, and the resulting code is shown in the
code window. Then, the user can choose any of the functions from the
module and write the arguments that she wants to evaluate the function
with. An initial system state, with an empty global mailbox and a
single process performing the specified function application, appears
on the state window when the user presses the start button, as shown
in Figure~\ref{fig:screenshot}. Now, the user is able to control
the system state by selecting the rules from the reversible semantics
that she wants to fire.

We have defined two different modes for controlling the reversible
semantics.  The first mode is a \emph{manual} mode, where the user
selects the rule to be fired for a particular process or message. Here,
the user is in charge of ``controlling'' the reversible semantics,
although this approach can rapidly become exhausting.
The second mode is the \emph{automatic} mode. Here, the user specifies
a number of steps and chooses a direction (forward or backward), and
the rules to be applied are selected at random---for the chosen
direction---until the specified number of steps is reached or no more
rules can be applied.  Alternatively, the user can move the state
forward up to a \emph{normalised system}. To normalise a system, one
must ignore the $\mathit{Sched}$ rule and apply only the other
rules. A normalised system is reached when no rule other than
$\mathit{Sched}$ can be fired.  Hence, in a normalised system, either
all processes are blocked (waiting for some message to arrive) or the
system state is final.  Normalising a system allows the user to
perform all the reductions that do not depend on the
network. Reductions depending on the network can then be performed one by one to understand their impact on the derivation. 

The release version (v1.0) of the application is fully written in
Erlang, and it is publicly available from
\texttt{https://github.com/mistupv/rev-erlang} under the MIT license.
Hence, the only requirement to build the application is to have
Erlang/OTP installed.  Besides, we have included some documentation and
a few examples to easily test the application.

\section{Related Work} \label{sec:relwork}

First, regarding the semantics of Erlang presented in
Section~\ref{sec:semantics}, we have some similarities with both
\cite{CMRT13tr} and \cite{SFB10}. In contrast to \cite{CMRT13tr},
which presents a monolithic semantics, our relation is split into
expression-level rules and system-level rules. This division eases the
presentation of a reversible semantics, since it only affects the
system-level rules. As for \cite{SFB10}, we follow the idea of
introducing a global mailbox (there called ``ether'') so that every
message passing communication can be decomposed into two steps: sending and
scheduling. Their semantics considers other features of Erlang (such as
links or monitors) but does not present the semantics of expressions,
as we do. Another difference lies in the fact that all \emph{side
  effects} are asynchronous in \cite{SFB10} (e.g., the spawning of a
process is asynchronous), a design decision that allows for a simpler
semantics. In our case, spawning a process is dealt with in a
synchronous manner, which is closer to the actual behaviour of
Erlang. Finally, as mentioned in Section~\ref{sec:semantics}, we
deliberately ignore the restriction that guarantees the order of
messages for any pair of given processes. This may increase the number
of possible interleavings, but we consider that it models better the
behaviour of current Erlang implementations.

Regarding reversibility, the approach presented in this paper is in
the line of work on causal-consistent reversibility \cite{DK04,PU07}
(see~\cite{LMT14} for a survey). In particular, our work is closer to
\cite{DK04}, since we also consider adding a \emph{memory} (a history
in our terminology) in order to make a computation
reversible. Moreover, our proof of causal consistency mostly follows
the proof scheme in \cite{DK04}. In contrast, we consider a
different concurrent language with asynchronous communication, while
communication in \cite{DK04} is synchronous. On the other hand,
\cite{PU07} does not introduce a memory but keeps the old actions
marked with a ``key''. As pointed out in \cite{PU07}, process
equivalence is easier to check than in \cite{DK04} (where one would
need to abstract away from the memories). Like \cite{DK04}, also
\cite{PU07} considers synchronous communication. Formalising the
Erlang semantics using a labelled transition relation as in
\cite{DK04,PU07} (rather than a reduction semantics, as we do in this
paper), and then defining a reversible extension would be an
interesting and challenging approach for further research.

Nevertheless, as mentioned in the Introduction, the closest to our
work is the debugging approach based on a rollback construct of
\cite{GLM14,GLMT15,LMSS11,LMS16,LLMS12}, but it is defined in the
context of a different language or formalism. Among the languages
considered in the works above, the closest to ours is
$\mu$Oz~\cite{LLMS12,GLM14}. A main difference is that $\mu$Oz is not
distributed: messages move atomically from the sender to a chosen
queue, and from the queue to the receiver. Each of the two actions is
performed by a specific process, hence naturally part of its
history. In our case, the scheduling action is not directly performed by a
process, and it is only potentially observed when the target process performs
the receive action (but not necessarily observed, e.g., if the message
does not match the patterns in the receive). The definition of the
notions of conflict and concurrency in this setting is, as a
consequence, much trickier than in $\mu$Oz. This difficulty carries
over to the definition of the history information that needs to be
tracked, and to how this information is exploited in the reversible
semantics (actually, this was one of the main difficulties we
encountered during our work). Furthermore, in the case of $\mu$Oz only
the uncontrolled semantics has been fully formalised~\cite{LLMS12},
while the controlled semantics and the corresponding results are only
sketched~\cite{GLM14}.

Also, we share some similarities with the checkpointing technique for
fault-tolerant distributed computing of \cite{FV05,KFV14}, although
the aim is different (they aim at defining a new language rather than
extending an existing one).

On the other hand, \cite{NY17} has very recently introduced a novel technique
for recovery in Erlang based on session types. Although the approach
is different, our rollback semantics could also be used for rollback
recovery. In contrast to \cite{NY17}, that only considers recovery of
processes as a whole, our approach could be helpful to design a more
fine grained recovery strategy.

Finally, as mentioned in the Introduction, this paper extends and improves
\cite{NPV16b} in different ways. Firstly, \cite{NPV16b} only presents a
rollback semantics. Here, we have introduced an uncontrolled reversible
semantics and have proved a number of fundamental theoretical properties,
including its causal consistency (no proofs of technical results are provided
in \cite{NPV16b}). Secondly, the reversible semantics in \cite{NPV16b} does
not consider messages' unique identifiers ($\k$), so that the problems
mentioned in Section~\ref{sec:revsem} are not avoided. Moreover, the process'
histories also include items for the applications of rule \textit{Sched},
which makes the underlying notion of concurrency unnecessarily restrictive. As
for the rollback semantics of \cite{NPV16b}, besides the points mentioned
above, it only considered one rollback for each process, while sets of
rollbacks are accepted in this work. Consequently, we have now reduced the
number of rules required to undo the sending of a message or to undo
the introduction of a checkpoint, so
that the rollback semantics is simpler. Furthermore, we have designed and
developed a proof-of-concept implementation in this paper that allowed us to
check the viability of the reversible semantics in practice.

\section{Conclusion and Future Work} \label{sec:conc}

We have defined a reversible semantics for a first-order subset of
Erlang that undoes the actions of a process step by step in a
sequential way. To the best of our knowledge, this is the first
attempt to define a reversible semantics for Erlang. In this work, we
have first introduced an uncontrolled, reversible semantics, and have
proved that it enjoys the usual properties (loop lemma, square lemma,
and causal consistency). Then, we have introduced a controlled version
of the backward semantics that can be used to model a rollback
operator that undoes the actions of a process up to a given
checkpoint. A proof-of-concept implementation shows that our approach
is indeed viable in practice.

As future work, we consider the definition of mechanisms to control
reversibility so that history information is stored only when needed
to perform a rollback.
This could be essential to extend Erlang with a new construct for
\emph{safe sessions}, where all the actions in a session can be undone
if the session aborts. Such a construct could have a great potential
to automate the fault-tolerance capabilities of the language Erlang.


\begin{thebibliography}{10}

\bibitem{AVW96}
J.~Armstrong, R.~Virding, C.~Wikstr\"om, and M.~Williams.
\newblock {\em {Concurrent programming in Erlang (2nd edition)}}.
\newblock Prentice Hall, 1996.

\bibitem{BN98}
F.~Baader and T.~Nipkow.
\newblock {\em Term Rewriting and All That}.
\newblock Cambridge University Press, 1998.

\bibitem{Ben73}
C.~Bennett.
\newblock Logical reversibility of computation.
\newblock {\em IBM Journal of Research and Development}, 17:525--532, 1973.

\bibitem{Bennett00}
C.~Bennett.
\newblock Notes on the history of reversible computation.
\newblock {\em {IBM} Journal of Research and Development}, 44(1):270--278,
  2000.

\bibitem{CMRT13tr}
R.~Caballero, E.~Mart{\'{\i}}n-Mart{\'{\i}}n, A.~Riesco, and S.~Tamarit.
\newblock A declarative debugger for concurrent erlang programs (extended
  version).
\newblock Technical Report SIC-15/13, Dpto.\ Sistemas Inform\'aticos y
  Computaci\'on, Universidad Complutense de Madrid, 2013.

\bibitem{CL11}
L.~Cardelli and C.~Laneve.
\newblock Reversible structures.
\newblock In F.~Fages, editor, {\em Proceedings of the 9th International
  Conference on Computational Methods in Systems Biology ({CMSB} 2011)}, pages
  131--140. {ACM}, 2011.

\bibitem{CGJLNPV04}
R.~Carlsson, B.~Gustavsson, E.~Johansson, T.~Lindgren, S.-O. Nystr\"om,
  M.~Pettersson, and R.~Virding.
\newblock Core erlang 1.0.3. language specification, 2004.
\newblock Available from \verb$https://www.it.uu.se/research/$
  \verb$group/hipe/cerl/doc/core_erlang-1.0.3.pdf$.

\bibitem{CKV13}
I.~Cristescu, J.~Krivine, and D.~Varacca.
\newblock A compositional semantics for the reversible p-calculus.
\newblock In {\em 28th Annual {ACM/IEEE} Symposium on Logic in Computer
  Science, {LICS} 2013}, pages 388--397. {IEEE} Computer Society, 2013.

\bibitem{DK04}
V.~Danos and J.~Krivine.
\newblock Reversible communicating systems.
\newblock In P.~Gardner and N.~Yoshida, editors, {\em Proc.\ of the 15th
  International Conference on Concurrency Theory ({CONCUR} 2004)}, volume 3170
  of {\em Lecture Notes in Computer Science}, pages 292--307. Springer, 2004.

\bibitem{DK05}
V.~Danos and J.~Krivine.
\newblock Transactions in {RCCS}.
\newblock In M.~Abadi and L.~{de~Alfaro}, editors, {\em Proc.\ of the 16th
  International Conference on Concurrency Theory ({CONCUR} 2005)}, volume 3653
  of {\em Lecture Notes in Computer Science}, pages 398--412. Springer, 2005.

\bibitem{FV05}
J.~Field and C.~A. Varela.
\newblock Transactors: a programming model for maintaining globally consistent
  distributed state in unreliable environments.
\newblock In J.~Palsberg and M.~Abadi, editors, {\em Proceedings of the 32nd
  {ACM} {SIGPLAN-SIGACT} Symposium on Principles of Programming Languages
  ({POPL} 2005)}, pages 195--208. {ACM}, 2005.

\bibitem{Fra05}
M.~P. Frank.
\newblock Introduction to reversible computing: motivation, progress, and
  challenges.
\newblock In N.~Bagherzadeh, M.~Valero, and A.~Ram{\'{\i}}rez, editors, {\em
  Proceedings of the Second Conference on Computing Frontiers}, pages 385--390.
  {ACM}, 2005.

\bibitem{Fre01}
L.-A. Fredlund.
\newblock {\em {A framework for reasoning about Erlang code}}.
\newblock PhD thesis, The Royal Institute of Technology, Sweeden, 2001.

\bibitem{GLM14}
E.~Giachino, I.~Lanese, and C.~A. Mezzina.
\newblock Causal-consistent reversible debugging.
\newblock In S.~Gnesi and A.~Rensink, editors, {\em Proc.\ of the 17th
  International Conference on Fundamental Approaches to Software Engineering
  ({FASE} 2014)}, volume 8411 of {\em Lecture Notes in Computer Science}, pages
  370--384. Springer, 2014.

\bibitem{GLMT15}
E.~Giachino, I.~Lanese, C.~A. Mezzina, and F.~Tiezzi.
\newblock Causal-consistent reversibility in a tuple-based language.
\newblock In M.~Daneshtalab, M.~Aldinucci, V.~Lepp{\"{a}}nen, J.~Lilius, and
  M.~Brorsson, editors, {\em Proceedings of the 23rd Euromicro International
  Conference on Parallel, Distributed, and Network-Based Processing, {PDP}
  2015}, pages 467--475. {IEEE} Computer Society, 2015.

\bibitem{KFV14}
P.~Kuang, J.~Field, and C.~A. Varela.
\newblock Fault tolerant distributed computing using asynchronous local
  checkpointing.
\newblock In E.~G. Boix, P.~Haller, A.~Ricci, and C.~Varela, editors, {\em
  Proceedings of the 4th International Workshop on Programming based on Actors
  Agents {\&} Decentralized Control (AGERE! 2014)}, pages 81--93. {ACM}, 2014.

\bibitem{Lan61}
R.~Landauer.
\newblock Irreversibility and heat generation in the computing process.
\newblock {\em IBM Journal of Research and Development}, 5:183--191, 1961.

\bibitem{LMSS11}
I.~Lanese, C.~A. Mezzina, A.~Schmitt, and J.~Stefani.
\newblock Controlling reversibility in higher-order pi.
\newblock In J.~Katoen and B.~K{\"{o}}nig, editors, {\em Proceedings of the
  22nd International Conference on Concurrency Theory ({CONCUR} 2011)}, volume
  6901 of {\em Lecture Notes in Computer Science}, pages 297--311. Springer,
  2011.

\bibitem{LMS16}
I.~Lanese, C.~A. Mezzina, and J.~Stefani.
\newblock Reversibility in the higher-order {\(\pi\)}-calculus.
\newblock {\em Theor. Comput. Sci.}, 625:25--84, 2016.

\bibitem{LMT14}
I.~Lanese, C.~A. Mezzina, and F.~Tiezzi.
\newblock Causal-consistent reversibility.
\newblock {\em Bulletin of the {EATCS}}, 114, 2014.

\bibitem{cauder}
I.~Lanese, N.~Nishida, A.~Palacios, and G.~Vidal.
\newblock {CauDEr} website.
\newblock URL: \url{https://github.com/mistupv/cauder}.

\bibitem{LNPV:FLOPS}
I.~Lanese, N.~Nishida, A.~Palacios, and G.~Vidal.
\newblock {CauDEr}: {A} causal-consistent reversible debugger for {E}rlang.
\newblock In J.~P. Gallagher and M.~Sulzmann, editors, {\em Proceedings of the
  14th International Symposium on Functional and Logic Programming (FLOPS
  2018)}, volume 10818 of {\em Lecture Notes in Computer Science}, pages
  247--263. Springer-Verlag, Berlin, 2018.

\bibitem{LLMS12}
M.~Lienhardt, I.~Lanese, C.~A. Mezzina, and J.~Stefani.
\newblock A reversible abstract machine and its space overhead.
\newblock In H.~Giese and G.~Rosu, editors, {\em Proceedings of the Joint 14th
  {IFIP} {WG} Int'l Conf.\ on Formal Techniques for Distributed Systems
  ({FMOODS} 2012) and the 32nd {IFIP} {WG} 6.1 International Conference
  ({FORTE} 2012)}, volume 7273 of {\em Lecture Notes in Computer Science},
  pages 1--17. Springer, 2012.

\bibitem{MHNHT07}
K.~Matsuda, Z.~Hu, K.~Nakano, M.~Hamana, and M.~Takeichi.
\newblock Bidirectionalization transformation based on automatic derivation of
  view complement functions.
\newblock In R.~Hinze and N.~Ramsey, editors, {\em Proc.\ of the 12th {ACM}
  {SIGPLAN} International Conference on Functional Programming, {ICFP} 2007},
  pages 47--58. {ACM}, 2007.

\bibitem{NY17}
R.~Neykova and N.~Yoshida.
\newblock Let it recover: multiparty protocol-induced recovery.
\newblock In P.~Wu and S.~Hack, editors, {\em Proceedings of the 26th
  International Conference on Compiler Construction, CC 2017}, pages 98--108.
  {ACM}, 2017.

\bibitem{NPV16}
N.~Nishida, A.~Palacios, and G.~Vidal.
\newblock Reversible term rewriting.
\newblock In D.~Kesner and B.~Pientka, editors, {\em 1st International
  Conference on Formal Structures for Computation and Deduction, {FSCD} 2016},
  volume~52 of {\em LIPIcs}, pages 28:1--28:18. Schloss Dagstuhl -
  Leibniz-Zentrum fuer Informatik, 2016.

\bibitem{NPV16b}
N.~Nishida, A.~Palacios, and G.~Vidal.
\newblock A reversible semantics for {E}rlang.
\newblock In M.~Hermenegildo and P.~L\'opez-Garc\'{\i}a, editors, {\em Proc.\
  of the 26th International Symposium on Logic-Based Program Synthesis and
  Transformation, LOPSTR 2016}, volume 10184 of {\em LNCS}, pages 259--274.
  Springer, 2017.
\newblock Preliminary version available from
  \texttt{https://arxiv.org/abs/1608.05521}.

\bibitem{PU07}
I.~Phillips and I.~Ulidowski.
\newblock Reversing algebraic process calculi.
\newblock {\em J. Log. Algebr. Program.}, 73(1-2):70--96, 2007.

\bibitem{Ros73}
B.~K. Rosen.
\newblock Tree-manipulating systems and {Church-Rosser} theorems.
\newblock {\em Journal of the ACM}, 20(1):160--187, 1973.

\bibitem{SFB10}
H.~Svensson, L.-A. Fredlund, and C.~B. Earle.
\newblock {A unified semantics for future Erlang}.
\newblock In {\em Proc.\ of the 9th ACM SIGPLAN workshop on Erlang}, pages
  23--32. ACM, 2010.

\bibitem{TA15}
M.~K. Thomsen and H.~B. Axelsen.
\newblock {Interpretation and programming of the reversible functional language
  RFUN}.
\newblock In {\em Proc.\ of the 27th International Symposium on Implementation
  and Application of Functional Languages (IFL 2015)}, pages 8:1 -- 8:13. ACM,
  2016.

\bibitem{TY15}
F.~Tiezzi and N.~Yoshida.
\newblock Reversible session-based pi-calculus.
\newblock {\em J. Log. Algebr. Meth. Program.}, 84(5):684--707, 2015.

\bibitem{Yok10}
T.~Yokoyama.
\newblock Reversible computation and reversible programming languages.
\newblock {\em Electronic Notes in Theoretical Computer Science},
  253(6):71--81, 2010.
\newblock Proc.\ of the Workshop on Reversible Computation (RC 2009).

\bibitem{YAG08}
T.~Yokoyama, H.~B. Axelsen, and R.~Gl{\"{u}}ck.
\newblock Principles of a reversible programming language.
\newblock In A.~Ram{\'{\i}}rez, G.~Bilardi, and M.~Gschwind, editors, {\em
  Proc.\ of the 5th Conference on Computing Frontiers}, pages 43--54. {ACM},
  2008.

\end{thebibliography}

\end{document}